\def\confversion{0}
\def\ifconf{\ifnum\confversion=1}
\def\ifnotconf{\ifnum\confversion=0}

\documentclass[11 pt]{article}

\def\showauthornotes{0}
\def\showkeys{0}
\def\showdraftbox{0}
\usepackage{xspace,xcolor,enumerate}
\usepackage{amsmath,amssymb}
\usepackage{amsthm}
\usepackage[toc,page]{appendix}
\usepackage{thmtools}
\usepackage{thm-restate}
\usepackage{color,graphicx}
\usepackage{boxedminipage}
\usepackage{makecell}

\ifnum\showkeys=1
\usepackage[color]{showkeys}
\fi

\definecolor{darkred}{rgb}{0.5,0,0}
\definecolor{darkgreen}{rgb}{0,0.35,0}
\definecolor{darkblue}{rgb}{0,0,0.55}

\usepackage[pdfstartview=FitH,pdfpagemode=UseNone,colorlinks,linkcolor=darkblue,filecolor=darkred,citecolor=darkgreen,urlcolor=darkred,pagebackref,bookmarks=false]{hyperref}

\usepackage[capitalise,nameinlink]{cleveref}
\usepackage[T1]{fontenc}
\usepackage{mathtools,dsfont,bbm}
\usepackage{mathpazo}

\ifnum\showauthornotes=1
\newcommand{\Authornote}[2]{{\sf\small\color{red}{[#1: #2]}}}
\newcommand{\Authorcomment}[2]{{\sf \small\color{gray}{[#1: #2]}}}
\newcommand{\Authorfnote}[2]{\footnote{\color{red}{#1: #2}}}
\else
\newcommand{\Authornote}[2]{}
\newcommand{\Authorcomment}[2]{}
\newcommand{\Authorfnote}[2]{}
\fi

\ifnum\showdraftbox=1
\newcommand{\draftbox}{\begin{center}
  \fbox{%
    \begin{minipage}{2in}%
      \begin{center}%
        \begin{Large}%
          \textsc{Working Draft}%
        \end{Large}\\
        Please do not distribute%
      \end{center}%
    \end{minipage}%
  }%
\end{center}
\vspace{0.2cm}}
\else
\newcommand{\draftbox}{}
\fi

\newtheorem{theorem}{Theorem}[section]

\newtheorem{observation}[theorem]{Observation}
\newtheorem{definition}[theorem]{Definition}

\newtheorem{lemma}[theorem]{Lemma}
\newtheorem{remark}[theorem]{Remark}
\newtheorem{proposition}[theorem]{Proposition}
\newtheorem{corollary}[theorem]{Corollary}
\newtheorem{claim}[theorem]{Claim}
\newtheorem{fact}[theorem]{Fact}

\newtheorem{example}[theorem]{Example}

\newtheorem{algo}[theorem]{Algorithm}

\def\FullBox{\hbox{\vrule width 6pt height 6pt depth 0pt}}

\def\qed{\ifmmode\qquad\FullBox\else{\unskip\nobreak\hfil
\penalty50\hskip1em\null\nobreak\hfil\FullBox
\parfillskip=0pt\finalhyphendemerits=0\endgraf}\fi}

\def\qedsketch{\ifmmode\Box\else{\unskip\nobreak\hfil
\penalty50\hskip1em\null\nobreak\hfil$\Box$
\parfillskip=0pt\finalhyphendemerits=0\endgraf}\fi}

\def\to{\rightarrow}
\def\eps{\varepsilon}
\def\epsilon{\varepsilon}

\def\eps{\epsilon}

\def\phi{\varphi}
\def\cal{\mathcal}

\def\psdgeq{\succeq} 
\newcommand{\defeq}{:=}

\newcommand{\ie}{i.e.,\xspace}

\newcommand{\etal}{et al.\xspace}

\newcommand{\mper}{\,.}
\newcommand{\mcom}{\,,}

\newcommand{\R}{{\mathbb R}}
\newcommand{\E}{{\mathbb E}}

\newcommand{\N}{{\mathbb{N}}}

\newcommand{\pmone}{\{-1,1\}\xspace}

\usepackage{nicefrac}

\newcommand{\abs}[1]{\ensuremath{\left\lvert #1 \right\rvert}}

\newcommand{\norm}[1]{\ensuremath{\left\lVert #1 \right\rVert}}

\newcommand{\Varsymb}{\mathrm{Var}}

\makeatletter

\def\Var#1{%
    \ProbabilityRender{\Varsymb}{#1}%
}

\def\ProbabilityRender#1#2{
  \@ifnextchar\bgroup%
  {\renderwithdist{#1}{#2}}
   {\singlervrender{#1}{#2}}
}
\def\singlervrender#1#2{%
   \ensuremath{\mathchoice
       {{#1}\left[ #2 \right]}
       {{#1}[ #2 ]}
       {{#1}[ #2 ]}
       {{#1}[ #2 ]}
   }
}
\def\renderwithdist#1#2#3{%
   \@ifnextchar\bgroup
   {\superfancyrender{#1}{#2}{#3}}
   {\ensuremath{\mathchoice
      {\underset{#2}{#1}\left[ #3 \right]}
      {{#1}_{#2}[ #3 ]}
      {{#1}_{#2}[ #3 ]}
      {{#1}_{#2}[ #3 ]}
     }
   }
}
\def\superfancyrender#1#2#3#4#5{
   \ensuremath{\mathchoice
      {\underset{#1}{{#1}}\left#4 #3 \right#5}
      {{#1}_{#2}#4 #3 #5}
      {{#1}_{#2}#4 #3 #5}
      {{#1}_{#2}#4 #3 #5}
   }
}
\makeatother

\newfont{\inhead}{eufm10 scaled\magstep1}

\newcommand{\calD}{{\cal D}}

\newcommand{\calF}{{\cal F}}

\newcommand{\calH}{{\cal H}}

\newcommand{\calL}{{\cal L}}
\newcommand{\calM}{{\cal M}}

\newcommand{\calP}{{\cal P}}

\newcommand{\calS}{{\cal S}}

\newcommand{\polylog}{{\mathrm{polylog}}}

\DeclareMathOperator{\tr}{\operatorname {tr}}

\newcommand{\el}{\ensuremath{\ell} }

\renewcommand{\iff}{\ensuremath{\Leftrightarrow}}

\newcommand{\1}[1]{\mathds{1}\left[#1\right]}

\newcommand{\littleoh}{\operatorname{o}}

\newcommand{\res}{{\sf Res}}

\newcommand{\Erdos}{Erd\H{o}s\xspace}
\newcommand{\Renyi}{R\'enyi\xspace}
\newcommand{\Lovasz}{Lov\'asz\xspace}

\usepackage[top=1in, bottom=1in, left=1in, right=1in]{geometry}
\usepackage{mathtools}

\usepackage{blkarray}
\usepackage{bigstrut}

\newtheorem*{claim*}{Claim}

\usepackage{thm-restate}

\DeclareMathOperator{\pE}{\widetilde{\E}}
\newcommand{\T}{\intercal}

\DeclareMathOperator{\Bernoulli}{Bernoulli}
\DeclareMathOperator{\Aut}{Aut}

\DeclareMathOperator{\Int}{Int}
\DeclareMathOperator{\conn}{conn}
\DeclareMathOperator{\linspan}{span}
\newcommand{\indset}[1]{1_{\overline{#1}}}
\newcommand{\psdleq}{\preceq}
\newcommand{\symfourier}{\textsf{sym}}
\newcommand{\mul}{{\sf mul}}

\newcommand{\Lribbon}{\calL}
\newcommand{\Mribbon}{\calM}
\newcommand{\midint}{\text{Mid}}

\newcommand{\EE}{\mathbb{E}}

\newcommand{\sig}{\sigma}

\newcommand{\al}{\alpha}
\newcommand{\lda}{\lambda}
\newcommand{\lam}{\lambda}
\newcommand{\Lda}{\Lambda}
\newcommand{\Lam}{\Lambda}

\newcommand{\gam}{\gamma}

\newcommand{\dsos}{D_{\text{SoS}}}
\newcommand{\id}{\textup{Id}}
\newcommand{\Id}{\textup{Id}}

\newcommand{\truncparam}{C\dsos\log n}

\ifnum\showauthornotes=1
\newcommand{\anote}[1]{{\sf\small\color{orange}{ [Aaron: #1] }}}
\newcommand{\cnote}[1]{{\sf\small\color{blue}{ [Chris: #1] }}}
\newcommand{\gnote}[1]{{\sf\small\color{red}{ [Goutham: #1] }}}
\newcommand{\jnote}[1]{{\sf\small\color{violet}{ [Jeff: #1] }}}
\newcommand{\mnote}[1]{{\sf\small\color{olive}{ [Madhur: #1] }}}
\else
\newcommand{\anote}[1]{}
\newcommand{\cnote}[1]{}
\newcommand{\gnote}[1]{}
\newcommand{\jnote}[1]{}
\newcommand{\mnote}[1]{}
\fi
\begin{document}

\title{Sum-of-Squares Lower Bounds for Sparse Independent Set}
\author{Chris Jones, Aaron Potechin, Goutham Rajendran, Madhur Tulsiani, Jeff Xu}

\author{
  Chris Jones\thanks{{\tt University of Chicago}. {\tt csj@uchicago.edu}. Supported in part by NSF grant CCF-2008920.}
  \and
  Aaron Potechin\thanks{{\tt University of Chicago}. {\tt potechin@uchicago.edu}. Supported in part by NSF grant CCF-2008920.}
  \and  
  Goutham Rajendran\thanks{{\tt University of Chicago}. {\tt goutham@uchicago.edu}. Supported in part by NSF grant CCF-1816372.}
  \and
  Madhur Tulsiani\thanks{{\tt Toyota Technological Institute at Chicago}. {\tt madhurt@ttic.edu}. Supported by NSF grant CCF-1816372.}
  \and
  Jeff Xu\thanks{{\tt Carnegie Mellon University}. {\tt jeffxusichao@cmu.edu}. Supported by NSF CAREER Award \#2047933.}
}

\maketitle
 \thispagestyle{empty}
\draftbox

\thispagestyle{empty}

\begin{abstract}
The Sum-of-Squares (SoS) hierarchy of semidefinite programs is a powerful algorithmic paradigm which captures state-of-the-art algorithmic guarantees for a wide array of problems. In the average case setting, SoS lower bounds provide strong evidence of algorithmic hardness or information-computation gaps. Prior to this work, SoS lower bounds have been obtained for problems in the ``dense" input regime, where the input is a collection of independent Rademacher or Gaussian random variables, while the sparse regime has remained out of reach. We make the first progress in this direction by obtaining strong SoS lower bounds for the problem of Independent Set on sparse random graphs.
We prove that with high probability over an \Erdos-R{\'e}nyi random graph $G\sim G_{n,\frac{d}{n}}$ with average degree $d>\log^2 n$, degree-$\dsos$ SoS fails to refute the existence of an independent set of size 
$k = \Omega\left(\frac{n}{\sqrt{d}(\log n)(\dsos)^{c_0}} \right)$ 
in $G$ (where $c_0$ is an absolute constant), whereas the true size of the largest independent set in $G$ is $O\left(\frac{n\log d}{d}\right)$.

Our proof involves several significant extensions of the techniques used for proving SoS lower bounds in the dense setting. Previous lower bounds are based on the pseudo-calibration heuristic of Barak \etal [FOCS 2016] which produces a candidate SoS solution using a planted distribution indistinguishable from the input distribution via low-degree tests. In the sparse case the natural planted distribution \emph{does} admit low-degree distinguishers, and we show how to adapt the pseudo-calibration heuristic to overcome this. 

Another notorious technical challenge for the sparse regime is the quest for matrix norm bounds. 
In this paper, we obtain new norm bounds for \emph{graph matrices} in the sparse setting.
While in the dense setting the norms of graph matrices are characterized by the size of the minimum vertex separator of the corresponding graph, this turns not to be the case for sparse graph matrices. Another contribution of our work is developing a new combinatorial understanding of structures needed to understand the norms of sparse graph matrices.
\end{abstract}
\newpage
\ifnotconf
\pagenumbering{roman}
\setcounter{tocdepth}{2}
	\tableofcontents
\clearpage
\fi

\pagenumbering{arabic}
\setcounter{page}{1}

\section{Introduction}
The Sum-of-Squares (SoS) hierarchy is a powerful convex programming
technique that has led to successful approximation and recovery algorithms for various problems in the past decade.
SoS captures the best-known approximation algorithms for several classical combinatorial optimization problems.
Some of the additional successes of SoS also include Tensor PCA~\cite{HSS15, MSS16} and Constraint Satisfaction Problems with additional structure ~\cite{BRS11, GS11}. 
SoS is a family of convex relaxations parameterized by  \textit{degree}; by taking larger degree, one gets a better approximation to the true optimum at the  expense of a larger SDP instance. Thus we are interested in the tradeoff between degree and approximation quality. For an introduction to Sum-of-Squares algorithms, see~\cite{BS16:sos-course, FKP19:survey}.

The success of SoS on the upper bound side has also conferred on it an important role for the investigation of algorithmic hardness.
Lower bounds for the SoS hierarchy provide strong unconditional hardness results for several optimization problems and are of particular interest when NP-hardness results are unavailable.
An important such setting is the study of average case complexity of optimization problems, where relatively few techniques exist for establishing NP-hardness results~\cite{ABB19}.
In this setting, a study of the SoS hierarchy not only provides a powerful benchmark for average-case complexity, but also helps in understanding the structural properties of the problems: what makes them algorithmically challenging? 
Important examples of such results include an improved understanding of sufficient conditions for average-case hardness of CSPs \cite{KMOW17} and lower bounds for the planted clique problem \cite{BHKKMP16}.

An important aspect of previous lower bounds for the SoS hierarchy is that they apply for the so-called \emph{dense setting}, which corresponds to cases when the input distribution can be specified by a collection of independent Rademacher or Gaussian variables. 
In the case of planted clique, this corresponds to the case when the input is a random graph distributed according to $G_{n, \frac12}$ i.e. specified by a collection of $\binom{n}{2}$ independent Rademacher variables. In the case of CSPs, one fixes the structure of the lower bound instance and only considers an instance to be specified by the signs of the literals, which can again be taken as uniformly random $\pmone$ variables. Similarly, recent results by a subset of the authors \cite{PR20} for tensor PCA apply when the input tensor has independent Rademacher or Gaussian entries. 
The techniques used to establish these lower bounds have proved difficult to extend to the case when the input distribution naturally corresponds to a \emph{sparse graph} (or more generally, when it is specified by a collection of independent sub-gaussian variables, with Orlicz norm $\omega(1)$ instead of $O(1)$).

In this paper we are interested in extending lower bound technology for SoS to the \textit{sparse setting}, where the input is a graph with average degree $d \leq n/2$.
We use as a case study the fundamental combinatorial optimization problem of independent set. For the dense case $d=n/2$, finding an independent set is equivalent to finding a clique and the paper~\cite{BHKKMP16} shows an average-case lower bound
against the Sum-of-Squares algorithm. We extend the techniques introduced there, namely pseudocalibration,
graph matrices, and the approximate decomposition into positive semidefinite matrices, in order to show the first average-case lower bound for the sparse setting.
We hope that the techniques developed in this paper offer a gateway for the analysis of SoS on other sparse problems. \cref{sec:open-problems} lists several such problems that are likely to benefit from an improved understanding of the sparse setting.

Sample $G \sim G_{n,\frac{d}{n}}$ as an \Erdos-\Renyi random
graph\footnote{Unfortunately our techniques do not work for a random $d$-regular
graph. See the open problems~(\cref{sec:open-problems}).} with average degree $d$,
where we think of $d \ll n$.
Specializing to the problem of independent set, a maximum independent set in $G$ has size:
\begin{fact}[\cite{COE15, DM11, DSS16}]
    W.h.p. the max independent set in $G$ has size $(1+o_d(1)) \cdot \frac{2\ln d}{d} \cdot n$.
\end{fact}

The value of the degree-2 SoS relaxation for independent set
equals the \Lovasz $\vartheta$ function, which is an upper bound on the independence number $\alpha(G)$, by virtue of being a relaxation.
For random graphs $G \sim G_{n, d/n}$ this value is larger by a factor of about $\sqrt{d}$ than the true value of $\alpha(G)$ with
high probability.

\begin{fact}[\cite{CO05}]
    W.h.p. $\vartheta(G) =\Theta(\frac{n}{\sqrt{d}} )$.
\end{fact}

We will prove that the value of higher-degree SoS is also on the order
of $n/\sqrt{d}$, rather than $n/d$, and thereby demonstrate that the information-computation gap against basic SDP/spectral algorithms persists against higher-degree SoS.

\subsection{Our main results}
\label{sec:formal-results}
The solution to the convex relaxations obtained via the SoS hierarchy can be specified by the so-called ``pseudoexpectation operator". 
\begin{definition}[Pseudoexpectation]\label{def:pseudoexpectation}
A degree-$D$ pseudoexpectation operator $\pE$ is a linear functional on polynomials of degree at most $D$ (in $n$ variables) such that $\pE[1] = 1$ and $\pE[f^2] \geq 0$ for every polynomial $f$ with degree at most $D/2$.
A pseudoexpectation is said to satisfy a polynomial constraint $g = 0$ if $\pE[f \cdot g] = 0$ for all polynomials $f$ when $\deg(f \cdot g) \leq D$.
\end{definition}
In considering relaxations for independent set of a graph $G = (V,E)$, with variables $x_v$ being the $0/1$ indicators of the independent set, the SoS relaxation searches for pseudoexpectation operators satisfying the polynomial
constraints
\[
\forall v \in V. \quad x_v^2 = x_v 
\qquad
\quad \text{and} \quad 
\qquad
\forall (u,v) \in E. \quad x_u x_v = 0 \mper
\]
The objective value of the convex relaxation is given by the quantity $\pE[\sum_{v \in V} x_v] = \sum_{v \in V} \pE[x_v]$.
For the results below, we say that an event occurs with high probability (w.h.p.) when it occurs with probability at least $1-O(1/n^c)$ for some $c>0$.
The following theorem states our main result.
%
%
%
%
%
%
\begin{theorem}
    \label{thm:main}
There is an absolute constant $c_0 \in \N$ such that for sufficiently large $n \in \N$ and ${d \in [(\log n)^2, n^{0.5}]}$, and parameters $k, \dsos$ satisfying
\[
k ~\leq~ \frac{n}{\dsos^{c_0}\cdot \log n \cdot d^{1/2}} \mcom
\]
it holds w.h.p. for $G = (V,E) ~\sim~ G_{n,~d/n}$ that there exists a degree-$\dsos$ pseudoexpectation satisfying
\[
\forall v \in V. \quad x_v^2 = x_v 
\qquad
\quad \text{and} \quad 
\qquad
\forall (u,v) \in E. \quad x_u x_v = 0 \mcom
\]
and objective value $ \pE[\sum_{v \in V} x_v] ~\geq~ (1-o(1)) k$.    
\end{theorem}
\begin{remark}
This is a non-trivial lower bound whenever $\dsos ~\leq~ \left(\frac{d^{1/2}}{\log n}\right)^{1/c_0}$.
\end{remark}
\begin{remark}
It suffices to set $c_0 = 20$ for our current proof.
We did not optimize the tradeoff in $\dsos$ with $k$, but we did optimize
the log factor (with the hope of eventually removing it).
\end{remark}

\begin{remark}
Using the same technique, we can prove an $n^{\Omega(\eps)}$ SoS-degree lower bound for all $d \in [\sqrt{n}, n^{1-\eps}]$.
\end{remark}

For $n^\eps \leq d \leq n^{0.5}$, the theorem gives a polynomial $n^\delta$ SoS-degree lower bound.
For smaller $d$, the bound is still strong against low-degree SoS, but
it becomes trivial as $\dsos$ approaches $(d^{1/2}/\log n)^{1/c_0}$ or $d$ approaches $(\log n)^2$ since $k$ matches the
size of the maximum independent set in $G$, hence there is an actual
distribution over independent sets of this size (the expectation operator for which is trivially is also a pseudoexpectation operator).

The above bound says nothing about the ``almost dense'' regime $d \in [n^{1-\eps}, n/2]$. To handle this regime,
we observe that our techniques, along with the ideas from the $\Omega(\log n)$-degree SoS bound from \cite{BHKKMP16} for the dense case,
prove a lower bound for any degree $d \geq n^\eps$.
\begin{theorem}
    \label{thm:informal-logd}
    For any $\eps_1, \eps_2 >0$ there is $\delta > 0$, such that for $d \in [n^{\eps_1}, n/2]$ and $k \leq \frac{n}{d^{1/2+\eps_2}}$, it holds w.h.p. for $G = (V,E) ~\sim~ G_{n,~d/n}$ that there exists a degree-$(\delta \log d)$ pseudoexpectation satisfying
\[
\forall v \in V. \quad x_v^2 = x_v 
\qquad
\quad \text{and} \quad 
\qquad
\forall (u,v) \in E. \quad x_u x_v = 0 \mcom
\]
and objective value $\pE[\sum_{v \in V}  x_v] ~\geq~ (1-o(1))k$.    
\end{theorem}
In particular, these theorems rule out polynomial-time certification (i.e. constant
degree SoS) for any $d \geq \polylog(n)$.



\subsection{Our approach}
Proving lower bounds for the case of sparse graphs requires extending the previous techniques for SoS lower bounds in multiple ways. The work closest to ours is the planted clique lower bound of \cite{BHKKMP16}.
The idea there is to view a random graph $G \sim G_{n, ~1/2}$ as a random input in $\pmone^{\binom{n}{2}}$, and develop a canonical method called ``pseudocalibration" for obtaining the pseudoexpectation $\pE$ as a function of $G$. The pseudocalibration method takes the low-degree Fourier coefficients of $\mu$ based on a \emph{different} distribution on inputs $G$ (with large planted cliques), and takes higher degree coefficients to be zero. This is based on the heuristic that distribution $G_{n, ~1/2}$ and the planted distribution are indistinguishable by low-degree tests. The pseudoexpectation obtained via this heuristic is then proved to be PSD (\ie to satisfy $\pE[f^2] \geq 0$) by carefully decomposing its representation as a (moment) matrix $\Lambda$. One then needs to estimate the norms of various terms in this decomposition, known as ``graph matrices", which are random matrices with entries as low-degree polynomials (in $\pmone^{\binom{n}{2}}$), and carefully group terms together to form PSD matrices.

Each of the above components require a significant generalization in the sparse case. To begin with, there is no good planted distribution to work with, as the natural planted distribution (with a large planted independent set) \emph{is distinguishable from $G_{n, ~d/n}$ via low-degree tests}! While we still use the natural planted distribution to compute \emph{some} pseudocalibrated Fourier coefficients, we also truncate (set to zero) several low-degree Fourier coefficients, in addition to the high-degree coefficients as in \cite{BHKKMP16}. In particular, when the Fourier coefficients correspond to subgraphs where certain vertex sets are disconnected (viewed as subsets of $\binom{n}{2}$), we set them to zero.
This is perhaps the most conceptually interesting part of the proof, and we hope
that the same ``connected truncation'' will be useful for other integrality gap constructions.

The technical machinery for understanding norm bounds, and obtaining PSD terms, also requires a significant update in the sparse case.
Previously, norm bounds for graph matrices were understood in terms of minimum vertex separators for the corresponding graphs, and arguments for obtaining PSD terms required working with the combinatorics of vertex separators~\cite{AMP20}. 
However, the number of vertices in a vertex separator turns out to be insufficient to control the relevant norm bounds in the sparse case. This is because of the fact that unlike random $\pm 1$ variables, their $p$-biased analogs no longer have Orlicz norm $O(1)$ but instead $O(\frac{1}{\sqrt{p}})$, which results in both the vertices as well as edges in the graph playing a role in the norm bounds. 
We then characterize the norms of the relevant random matrices in terms of vertex separators, where the cost of a separator depends on the number of vertices and also the number of induced edges.
Moreover, the estimates on spectral norms obtained via the trace power method can fluctuate significantly due to rare events (presence of some dense subgraphs), and we need to carefully condition on the absence of these events.
A more detailed overview of our approach is presented in \cref{sec:proof-overview}.

\subsection{Related work}
Several previous works prove SoS lower bounds in the dense setting, when the inputs can be viewed as independent Gaussian or Rademacher random variables. Examples include the planted clique lower bound of Barak \etal \cite{BHKKMP16}, CSP lower bounds of Kothari \etal \cite{KMOW17}, and the tensor PCA lower bounds \cite{hop17, PR20}.
The technical component of decomposing the moment matrix in the dense case, as a sum of PSD matrices, is developed into a general ``machinery" in a recent work by a subset of the authors \cite{PR20}. 
A different approach than the ones based on pseudocalibration, which also applies in the dense regime, was developed by Kunisky~\cite{Kunisky20}.
%

For the case of independent set in random sparse graphs, many works have considered the \emph{search} problem of finding a large independent set in a random sparse graph. Graphs from $G_{n, d/n}$ are known to have independent sets of size $(1+o_d(1)) \cdot \frac{2\ln d}{d} \cdot n$ with high probability,
and it is possible to find an independent set of size $(1~+~o_d(1))\cdot \frac{\ln d}{d}~\cdot~n$, 
either by greedily taking a maximal independent set in the dense case~\cite{GM75} or by using a local algorithm in the sparse case~\cite{W95:indset}.
This is conjectured to be a computational phase transition, with concrete 
lower bounds
against search beyond $\frac{\ln d}{d}\cdot n$ for local algorithms \cite{RV17} and
low-degree polynomials~\cite{wein2020optimal}. 
The game in the search problem is all about the constant 1 vs 2, whereas our work
shows that the integrality gap of SoS is significantly worse, on the order of $\sqrt{d}$.
Lower bounds against search work in the regime of constant $d$ (though in principle they could be extended to at least some $d = \omega(1)$ with additional technical work), while our techniques require
$d \geq \log(n)$.
For search problems, the \emph{overlap distribution} of two high-value solutions has emerged
as a heuristic indicator of computational hardness,
whereas for certification problems it is unclear how the overlap distribution plays a role.

Norm bounds for sparse graph matrices were also obtained using a different method of matrix deviation inequalities, by a subset of the authors~\cite{RT20}.

The work~\cite{BBKMW20} constructs a \textit{computationally quiet} planted distribution
that is a candidate for pseudocalibration. However, their distribution
is not quite suitable for our purposes. \footnote{\cite{BBKMW20} provide evidence that their distribution is hard
to distinguish from $G_{n,d/n}$ with probability $1-o(1)$ (it is not ``strongly detectable''). However,
their distribution \emph{is} distinguishable with probability $\Omega(1)$, via
a triangle count (it is ``weakly detectable''). In SoS pseudocalibration, this manifests as $\pE[1] = \Theta_d(1)$. We would like the low-degree distinguishing probability to be $o(1)$ i.e. $\pE[1] = 1+o_d(1)$ so that normalizing
by $\pE[1]$ does not affect the objective value.}
\footnote{Another issue is that their planted distribution introduces noise by adding a small number of edges inside the planted independent set.}

A recent paper by Pang \cite{Pang21} fixes a technical shortcoming of \cite{BHKKMP16}
by constructing a pseudoexpectation operator that satisfies  ``$\sum_{v \in V} x_v = k$'' as a polynomial constraint (whereas the shortcoming was $\pE[\sum_{v \in V} x_v] \geq (1-o(1))k$ like we have here).


\subsection{Organization of the paper}
In \cref{sec:preliminaries}, we introduce the technical preliminaries and notation required for our arguments.
A technical overview of our proof strategy is presented in 
\cref{sec:proof-overview}.
\cref{sec:pseudocalibration} describes our modification of the pseudo-calibration heuristic using additional truncations.
Norm bounds for sparse graph matrices are obtained in \cref{sec:sparse_graph_matrices_and_conditioning}.  
\cref{sec:psdness} then proves the PSD-ness of our SoS solutions broadly using the machinery developed in \cite{BHKKMP16} and \cite{PR20}.

\section{Technical Preliminaries}
\label{sec:preliminaries}
    
\subsection{The Sum-of-Squares hierarchy}

The Sum-of-Squares (SoS) hierarchy is a hierarchy of semidefinite programs parameterized by its degree $D$. We will work with two equivalent definitions of a degree-$D$ SoS
solution: a pseudoexpectation operator $\pE$ (\cref{def:pseudoexpectation}) and a moment matrix. 
For a degree-$D$ solution to be well defined, we
need $D$ to be at least the maximum degree of all constraint polynomials.
The degree-$D$ SoS algorithm checks feasibility of a polynomial system by 
checking whether or not a degree-$D$ pseudoexpectation operator exists. This can be done in time $n^{O(D)}$ via semidefinite programming 
(ignoring some issues of bit complexity~\cite{RW17:sos}). To
show an SoS lower bound, one must construct a pseudoexpectation 
operator that exhibits the desired integrality gap.

\subsubsection{Moment matrix}\label{sec:moment-mtx}

We define
the moment matrix associated with a degree-$D$
pseudoexpectation $\pE$.
\begin{definition}[Moment Matrix of $\pE$]
  The moment matrix $\Lda=\Lda(\pE)$ associated to a pseudoexpectation $\pE$ is a
  $\binom{[n]}{\leq D/2} \times \binom{[n]}{\leq D/2}$ matrix with rows and columns indexed
  by subsets of $I, J \subseteq [n]$ of size at most $D/2$ and defined as
  \[
  \Lda[I, J] \defeq \pE\left[ x^I \cdot x^J \right].
  \]
\end{definition}

To show that a candidate pseudoexpectation satisfies $\pE[f^2] \geq 0$ in~\cref{def:pseudoexpectation}, we will rely on the following standard fact.
\begin{fact}
  In the definition of pseudoexpectation, \cref{def:pseudoexpectation}, the condition $\pE[f^2] \geq 0$ for all $\deg(f) \leq D/2$ is equivalent to $\Lda \succeq 0$.
\end{fact}

\subsection{$p$-biased Fourier analysis}
Since we are interested in sparse Erd\"os-R\'enyi graphs in this work, we will resort to $p$-biased Fourier analysis \cite[Section 8.4]{ODonnell14}.
Formally, we view the input graph $G \sim G_{n, p}$ as a vector in
$\{0,1\}^{\binom{n}{2}}$ indexed by sets $\{i, j\}$ for $i, j \in [n], i \neq j$, where each entry is independently sampled from the $p$-biased Bernoulli distribution, $\Bernoulli(p)$.
Here, by convention $G_e = 1$ indicates the edge $e$ is present, which happens with probability $p$.
The Fourier basis we use for analysis on $G$
is the set of $p$-biased Fourier
characters (which are naturally indexed by graphs $H$ on $[n]$).
\begin{definition}
    $\chi$ denotes the $p$-biased 
    Fourier character,
    \[\chi(0) = \sqrt{\frac{p}{1-p}}, \qquad \chi(1) = -\sqrt{\frac{1-p}{p}}.\]
    For $H$ a subset or multi-subset of $\binom{[n]}{2}$, let $\chi_H(G) \defeq \prod_{e \in H}\chi(G_e)$.
\end{definition}

We will also need the function $1-G_e$ which indicates that an edge is not present. 
\begin{definition}
    For $H \subseteq \binom{[n]}{2}$, let 
    $\indset{H}(G) = \prod_{e \in H}(1-G_e)$.
\end{definition}

When $H$ is a clique, this is the \textit{independent set indicator} for the
vertices in $H$.

\begin{proposition}
    \label{lem:independent-set-indicator}
    For $e \in \{0,1\}$, $1 + \sqrt{\frac{p}{1-p}}\chi(e) = \frac{1}{1-p}(1-e)$.
    Therefore, for any $H \subseteq \binom{[n]}{2}$,
    \[ \sum_{T \subseteq H} \left(\frac{p}{1-p}\right)^{\abs{T}/2}\chi_T(G) = \frac{1}{(1-p)^{\abs{H}}}\cdot \indset{H}(G).\]
\end{proposition}

\subsection{Ribbons and graph matrices}
A degree-$D$ pseudoexpectation operator is a vector in $\R^{\binom{[n]}{\leq D}}$.
The matrices we consider will have rows and columns indexed by all subsets of $[n]$.
We express the moment matrix $\Lam$ in terms of the Fourier basis on $G_{n, p}$. A particular Fourier character in a particular matrix entry is 
identified by a combinatorial structure called a \textit{ribbon}.

\begin{definition}[Ribbon]
\label{def:ribbon}
    A ribbon is a tuple $R = (V(R), E(R), A_R, B_R)$,
    where $(V(R), E(R))$ is an undirected multigraph without self-loops,
    $V(R) \subseteq [n]$, and $A_R, B_R \subseteq V(R)$.
    Let $C_R \defeq V(R) \setminus (A_R \cup B_R)$.
\end{definition}

\begin{definition}[Matrix for a ribbon]
    For a ribbon $R$, the matrix $M_R$ has rows and columns 
    indexed by all subsets of $[n]$ and has a single nonzero entry,
    \[ M_R[I,J] = \begin{cases}
        \displaystyle\chi_{E(R)}(G) 
        & I = A_R, J = B_R\\
        0 & \text{Otherwise}
    \end{cases}\]
\end{definition}
\begin{definition}[Ribbon isomorphism]
Two ribbons $R, S$ are isomorphic, or have the same shape, if there is a bijection between $V(R)$ and $V(S)$ which is a multigraph isomorphism 
between $E(R), E(S)$ and 
is a bijection from $A_R$ to $A_S$ and $B_R$ to $B_S$.
Equivalently, letting $S_n$ permute the vertex labels of a ribbon, the two ribbons are in the same $S_n$-orbit.
\end{definition}

If we ignore the labels on the vertices of a ribbon, what remains is the
\textit{shape} of the ribbon. 

\begin{definition}[Shape]
\label{def:shape}
    A shape is an equivalence class of ribbons with the same shape. Each shape 
    has associated with it
    a representative $\alpha = (V(\alpha), E(\alpha), U_\alpha, V_\alpha)$, where $U_\alpha,V_\alpha \subseteq V(\alpha)$.
    Let $W_\alpha \defeq V(\alpha) \setminus (U_\alpha \cup V_\alpha).$
\end{definition}

\begin{definition}[Embedding]
    Given a shape $\alpha$ and an injective function $\varphi: V(\alpha) \to [n],$ we let $\varphi(\alpha)$ be the ribbon obtained by labeling $\alpha$ in the natural way.
\end{definition}

\begin{definition}[Graph matrix]
    For a shape $\alpha$, the graph matrix $M_\alpha$ is
    \[M_\alpha = \displaystyle\sum_{\text{injective }\phi: V(\alpha) \to [n]} M_{\phi(\alpha)}. \]
\end{definition}

Injectivity is an important property of graph matrices. On the one hand,
we have a finer partition of ribbons than allowing all assignments, and this
allows more control. On the other hand,
injectivity introduces technically challenging ``intersection terms'' into graph matrix 
multiplication.
A graph matrix is essentially a sum over all ribbons with shape $\alpha$ (this is not entirely accurate as each ribbon will
be repeated $\abs{\Aut(\alpha)}$ times).

\begin{definition}[Automorphism]
    For a shape $\alpha$, $\Aut(\al)$ is the group of bijections from $V(\al)$ to itself such that $U_\al$ and $V_\al$ are fixed as sets
    and the map is a multigraph automorphism on $E(\al)$.
    Equivalently, $\Aut(\al)$ is the stabilizer subgroup (of $S_n$) of any ribbon of shape $\al$.
\end{definition}
\begin{fact}
    \label{lem:injective-graph-matrices}
    \[M_\al = \sum_{\text{injective }\phi:V(\al) \to [n]} M_{\phi(\al)} = \abs{\Aut(\al)}\sum_{R\text{  ribbon of shape }\al} M_R\]
\end{fact}
\begin{figure}
    \centering
    \includegraphics[scale=0.5]{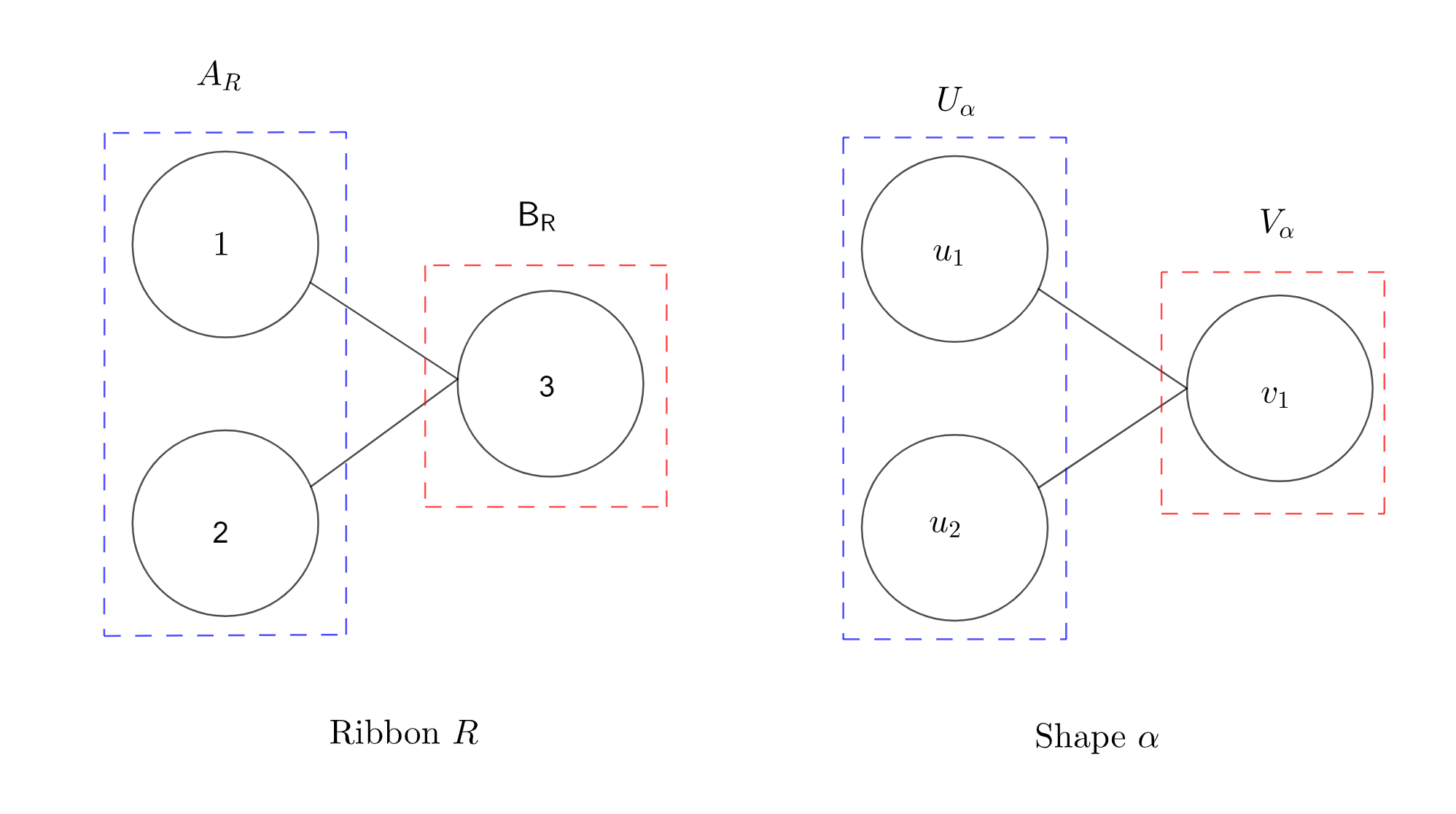}
    \caption{An example of a ribbon and a shape}
    \label{fig: ribbon_shape}
\end{figure}

\begin{example}[Ribbon]
    As an example, consider the ribbon in \cref{fig: ribbon_shape}. We have $A_R = \{1, 2\}, B_R = \{3\}, V(R) = \{1, 2, 3\}, E(R) = \{\{1, 2\}, \{2, 3\}\}$. The Fourier character is $\chi_{E_R} = \chi_{1, 3}\chi_{2, 3}$. And finally, $M_R$ is a matrix with rows and columns indexed by subsets of $[n]$, with exactly one nonzero entry $M_R(\{1, 2\}, \{3\}) = \chi_{\{1, 3\}}\chi_{\{2, 3\}}$. Succinctly, \[M_R =
  \begin{blockarray}{rl@{}c@{}r}
    & & \makebox[0pt]{column $\{3\}$} \\[-0.5ex]
    & & \,\downarrow \\[-0.5ex]
    \begin{block}{r(l@{}c@{}r)}
    & \makebox[3.1em]{\Large $0$\bigstrut[t]} & \vdots &\makebox[4.2em]{\Large $0$} \\[-0.2ex]
    \text{row }\{1, 2\} \to \mkern-9mu & \raisebox{0.5ex}{\makebox[3.2em][l]{\dotfill}} & \chi_{1, 3}\chi_{2, 3} & \raisebox{0.5ex}{\makebox[4.2em][r]{\dotfill}} \\[+0ex]
    & \makebox[3.1em]{\Large $0$} & \vdots &\makebox[4.2em]{\bigstrut\Large $0$} \\
    \end{block}
  \end{blockarray}\]
\end{example}

\begin{example}[Shape]
    In \cref{fig: ribbon_shape}, consider the shape $\al$ as shown. We have $U_{\al} = \{u_1, u_2\}, V_{\al} = \{v_1\}, W_{\al} = \emptyset, V(\al) = \{u_1, u_2, v_1\},$ and $E(\al) = \{\{u_1, v_1\}, \{u_2, v_1\}\}$. $M_{\al}$ is a matrix with rows and columns indexed by subsets of $[n]$. The nonzero entries will have rows and columns indexed by $\{a_1, a_2\}$ and $b_1$ respectively for all distinct $a_1, a_2, b_1 \in [n]$, with the corresponding entry being $M_{\al}(\{a_1, a_2\}, \{b_1\}) = \chi_{a_1, b_1}\chi_{a_2, b_1}$. Succinctly, \[M_{\al} =
  \begin{blockarray}{rl@{}c@{}r}
    & & \makebox[0pt]{column $\{b_1\}$} \\[-0.5ex]
    & & \,\downarrow \\[-0.5ex]
    \begin{block}{r(l@{}c@{}r)}
    &  & \vdots & \\[-0.2ex]
    \text{row }\{a_1, a_2\} \to \mkern-9mu & \raisebox{0.5ex}{\makebox[3.2em][l]{\dotfill}} & \chi_{a_1, b_1}\chi_{a_2, b_1} & \raisebox{0.5ex}{\makebox[4.2em][r]{\dotfill}} \\[+.5ex]
    &  & \vdots & \\
    \end{block}
  \end{blockarray}\]
\end{example}

\begin{definition}[Proper]
    A ribbon or shape is proper if it has no multi-edges. Otherwise, it is improper.
    Let $\mul_\al(e)$ be the multiplicity of edge $e$ in ribbon or shape $\al$.
\end{definition}

An improper ribbon or shape with an edge $e$ of multiplicity 2, e.g., 
has a squared Fourier character
$\chi_e^2$.
Since this is a function on $\{0,1\}$, by expressing it in the Fourier basis an improper ribbon or shape can be decomposed in a unique way 
into a linear combination of proper ones, which we call \emph{linearizations}.
\begin{definition}[Linearization]
    Given an improper ribbon or shape $\al$, a linearization $\beta$ is a proper ribbon or shape such that $\mul_\beta(e) \leq \mul_\al(e)$ for all $e \in E(\al)$.
\end{definition}

\begin{definition}[Isolated vertex]
    For a shape $\al$, an isolated vertex is a degree-0 vertex in $W_\al$.
    Let $I_\al$ denote the set of isolated vertices in $\al$.
    Similarly, for a ribbon $R$, the isolated vertices are denoted $I_R$.
\end{definition}
We stress that an isolated vertex never refers to degree-0 vertices inside $U_\al \cup V_\al$.

\begin{definition}[Trivial shape]
    A shape $\al$ is trivial if $V(\alpha) = U_\al = V_\al$ and
    $E(\al) = \emptyset$.
\end{definition}

$M_\al$ for a trivial $\al$ is the identity matrix restricted to the degree-$\abs{U_\al}$ block.

\begin{definition}[Transpose]
    Given a ribbon $R$ or shape $\alpha$, we define its transpose by swapping $A_R$ and $B_R$ (resp. $U_\alpha$ and $V_\alpha$).
    Observe that this transposes the matrix for the ribbon/shape.
\end{definition}
\section{An overview of the proof techniques}\label{sec:proof-overview}

Here, we will give a sketch of the proof techniques that we utilize in our SoS lower bound. Recall that we are given a graph $G \sim G_{n, p}$ where $d = pn$ is the average degree and our goal is to show that for any constant $\eps>0$, $\dsos \approx n^\delta$ for some $\delta>0$, degree $\dsos$ SoS thinks there exists an independent set of size $k \approx \frac{n}{d^{1/2+\eps}(\dsos\log n)^{c_0} }$ whereas the true independent set has size $\approx \frac{n \log d}{d}$ for some absolute constant $c_0$.

To prove the lower bound, we review the Planted Clique lower bound~\cite{BHKKMP16} and describe
the obstacles that need to be overcome in the sparse setting.

\subsection{Modified pseudocalibration}
Since SoS is a convex program, the goal of an SoS lower bound is to construct a dual object:
a set of \textit{pseudomoments} $\pE[x^S]$ for each small $S \subseteq V(G)$, which
are summarized in the \textit{moment matrix}.
The moment matrix must (i) obey the problem constraints (ii) be SoS-symmetric, and (iii) be positive semidefinite (PSD). 
Following the recipe of pseudocalibration introduced by~\cite{BHKKMP16}, we
can produce a candidate moment matrix which is guaranteed to satisfy the
first two conditions, while like all other SoS lower bounds, the hard work remains in verifying the PSDness of the moment matrix. Pseudocalibration has been successfully exploited in a multitude of SoS lower bound applications, e.g., \cite{BHKKMP16, KMOW17, MRX20, GJJPR20, PR20}. So, this is a natural starting point for us.

\paragraph{Failure of pseudocalibration}
The first obstacle we overcome is the lack of a planted distribution.
Pseudocalibration requires a planted and random
distribution which are hard to distinguish using the \textit{low-degree, 
likelihood ratio test} (i.e. $\widetilde{\E}[1]$ is bounded whp)~\cite{hop17, hop18}. In the case of sparse independent set, we have the following natural hypothesis testing problem with which one may hope to pseudocalibrate.

\begin{itemize}
    \item Null Hypothesis: Sample a graph $G \sim G_{n, p}$.
    \item Alternate Hypothesis: Sample a graph $G\sim G_{n, p}$. Then, sample a subset $S \subseteq [n]$ where each vertex is chosen with probability $\frac{k}{n}$. Then, plant an independent set in $S$, i.e. remove all the edges inside $S$.
\end{itemize}

In the case of sparse independent set,
the na\"ive planted distribution \textit{is} distinguishable from a random
instance via a simple low-degree test -- counting 4-cycles.
In all uses of pseudocalibration that we are aware of,
the two distributions being compared are
conjecturally hard to distinguish
by all polynomial-time algorithms.
We are still searching for a suitable planted 
distribution for sparse independent 
set, and we believe this is an interesting question on its own.

\paragraph{Fixing pseudocalibration via connected truncation}
To get around with this issue, we close our eyes and ''pretend'' the planted distribution is quiet, ignoring the obvious distinguisher, and make a ``connected truncation'' of the moment matrix
to remove terms which correspond to counting subgraphs in $G$.
What remains is that $\pE[x^S]$ is essentially independent of the
global statistics of $G$. It should be pointed out here that this is inherently distinct from the local truncation for weaker hierarchies (e.g. Sherali-Adams) where the moment matrix is an entirely local function \cite{DBLP:journals/corr/abs-1804-07842}. 
In contrast, our $\pE[x^S]$ may depend on parts of the graph
that are far away from $S$, in fact, even up to radius $n^\delta$, exceeding the diameter of the random graph!

At this point, the candidate moment matrix can be written as follows.

\begin{align*}
\Lambda &\defeq \displaystyle\sum_{\alpha \in \calS} \left(\frac{k}{n}\right)^{\abs{V(\alpha)}} \cdot \left(\frac{p}{1-p}\right)^{\frac{\abs{E(\alpha)}}{2}} M_\alpha.
\end{align*}

Here, $\calS$ ranges over all proper shapes $\al$ of appropriately bounded size such that \emph{all vertices of $\al$ are connected to $U_{\al} \cup V_{\al}$}. The latter property is the important distinction from standard pseudocalibration and will turn out to be quite essential for our analysis.

Using connected objects to take advantage of correlation decay is also a theme
in the cluster expansion from statistical physics~(see Chapter~5 of \cite{StatMechBook}). Although not formally
connected with connected truncation, the two methods share some similar characteristics.

\subsection{Approximate PSD decompositions, norm bounds and conditioning}
Continuing, to show the moment matrix is PSD, Planted Clique~\cite{BHKKMP16} performs an approximate factorization of the moment matrix in terms
of graph matrices.
A crucial part of this approach is to identify "dominant" and "non-dominant" terms in the approximate PSD decomposition. Then, the dominant terms are shown to be PSD and the non-dominant terms are bounded against the dominant terms.
In this approach, a crucial component in the latter step is to control the norms of graph matrices.

\paragraph{Tighter norm bounds for sparse graph matrices}
Existing norm bounds in the literature \cite{medarametla2016bounds, AMP20} for graph matrices have focused exclusively on the dense setting $G_{n, 1/2}$. Unfortunately, while these norm bounds apply for the sparse setting, they're too weak to be useful. Consider the case where we sample $G \sim G_{n, p}$ and try to bound the spectral norm of the centered adjacency matrix. Existing norm bounds give a bound of $\tilde{O}(\frac{\sqrt{n(1 - p)}}{\sqrt{p}})$ whereas the true norm is ${O}(\sqrt{n})$ regardless of $d$. This is even more pronounced when we use shapes with more vertices. So, our first step is to tighten the existing norm bounds in the literature for sparse graph matrices.

For a shape $\al$, a vertex separator $S$ is a subset of vertices such that there are no paths from $U_{\al}$ to $V_{\al}$ in $\al \setminus S$. It is known from previous works that in the dense case the spectral norm is controlled by the number of vertices in the minimum vertex separator between $U_\al$ and $V_\al$. Assuming $\al$ does not have isolated vertices and $U_\al\cap V_\al=\emptyset$ for simplicity, the norm of $M_\al$ is given by the following expression, up to polylog factors and the leading coefficient of at most $\abs{V(\al)}^{\abs{V(\al)}}$,\[ 
\norm{M_\al} \leq  \max_{\text{vertex separator } S} \sqrt{n}^{|V(\al)|-|V(S)|}
\]
However, it turns out this is no longer the controlling quantity if the underlying input matrix is sparse, and tightly determining this quantity arises as a natural task for our problem, and for future attack on SoS lower bounds for other problems in the sparse regime. 
To motivate the difference, we want to point out this is essentially due to the following simple observation. For $k\geq 2$,
\begin{align*}
    &|\E[X^k]| = 1\\
    &|\E[Y^k]| \leq \left(\sqrt{\frac{1-p}{p}}\right)^{k-2}\approx \left(\sqrt{\frac{n}{d}}\right)^{k-2} 
\end{align*}
for $X$ a uniform $\pm 1$ bit and $Y$ a $p$-biased random variable $\text{Ber}(p)$. This suggests that in the trace power method, there will be a preference among vertex separators of the same size if some contain more edges inside the separator (because vertices inside vertex separators are ''fixed'' in the dominant term in the trace calculation, and thus edges within the separator will contribute some large power of $\sqrt{\frac{n}{d}}$, creating a noticeable influence on the final trace). Finally, this leads us to the following characterization for sparse matrix norm bounds, up to polylog factors and the leading coefficient of at most $\abs{V(\al)}^{\abs{V(\al)}}$,
\[ 
\norm{M_\al} \leq \max_{\text{vertex separator } S} \sqrt{n}^{\abs{V(\al)}-\abs{V(S)} } \left(\sqrt{\frac{1-p}{p}}\right)^{\abs{E(S)}}
\]

The formal statement is given in \cref{sec:probabilistic-norm-bounds}. We prove this via an application of the trace method followed by a careful accounting of the large terms.

The key conceptual takeaway is that we need to redefine the weight of a vertex separator to also incorporate the edges within the separator, as opposed to only considering the number of vertices. We clearly distinguish these with the terms \textit{Dense Minimum Vertex Separator} (DMVS) and \textit{Sparse Minimum Vertex Separator} (SMVS). When $p = \frac{1}{2}$, these two bounds are the same up to lower order factors.

\paragraph{Approximate PSD decomposition}
We then perform an approximate PSD decomposition of the graph
matrices that make up $\Lam$. The general factoring strategy is
the same as~\cite{BHKKMP16}, though in the sparse regime we
must be very careful about what kind of combinatorial factors we
allow.
Each shape comes with a natural ``vertex decay'' coefficient arising from the fractional size of the independent set
and an ``edge decay'' coefficient arising from the sparsity of the graph. The vertex decay coefficients can be analyzed in a method
similar to Planted Clique (which only has vertex decay).
For the edge decay factors, we use novel charging arguments.
At this point, the techniques are strong enough to prove \cref{thm:informal-logd},
an SoS-degree $\Omega(\log n)$ lower bound for $d \geq n^\eps$. The remaining techniques are needed to push the SoS degree
up and the graph degree down.
 
\paragraph{Conditioning}
In our analysis, it turns out that to obtain strong SoS lower bounds in the sparse regime, a norm bound from the vanilla trace method is not quite sufficient. Sparse random matrices' spectral norms are fragile with respect to the influence of an unlikely event, exhibiting deviations away from the expectation with polynomially small probability (rather than exponentially small probability, like what is obtained from a good concentration bound). 
These ``bad events'' are small dense subgraphs present in a graph sampled from $G_{n,p}$.

To get around this, we condition on the high probability event that $G$
has no small dense subgraphs. For example, for $d = n^{1-\eps}$ whp every small subgraph $S$ has $O(\abs{S})$ edges (even up to size $n^{\delta}$).
For a \textbf{shape} which is dense (i.e. $v$ vertices and more than $O(v)$ edges)
we can show that its norm falls off extremely rapidly
under this conditioning.
This allows us to throw away dense shapes, which is 
critical for controlling combinatorial factors that would otherwise dominate
the moment matrix.

This type of conditioning is well-known: a long line of work showing tight norm bounds for the simple adjacency matrix appeals to a similar conditioning argument within the trace method \cite{BLM, Bor19, FM17, deshpande2019threshold}. 
We instantiate the conditioning in two ways. The first is through the following identity.
    
\begin{observation}
\label{obs:conditioning}
Given a set of edges $E \subseteq \binom{[n]}{2}$, if we know that not all of the edges of $E$ are in $E(G)$ then 
\[
\chi_E(G) = \sum_{E' \subseteq E: E' \neq E}{\left(\sqrt{\frac{p}{1-p}}\right)^{|E| - |E'|}\chi_{E'}(G)}
\]
    \end{observation}
This simple observation whose proof is deferred to \cref{sec: conditioning} can be applied recursively to replace a dense shape $\al$ by a sum of its sparse subshapes $\{\beta\}$. 
The second way we eliminate dense shapes is by using a bound on the Frobenius
norm which improves on the trace calculation for dense shapes.
After conditioning, we can restrict our attention to sparse shapes, which allows us to avoid several combinatorial factors which would otherwise overwhelm the ''charging'' argument.

Handling the subshapes $\{\beta\}$ requires some care. Destroying edges from a 
shape can cause its norm to either go up or down: the vertex separator gets
smaller (increasing the norm), but if we remove edges from inside the SMVS,
the norm goes down.
An important observation is we do not necessarily have to apply \cref{obs:conditioning} on the entire set of edges of a shape, but we can also just apply it on some of the edges. We will choose 
a set of edges $\res(\al)\subseteq E(\al)$ that ``protects the minimum vertex separator'' and only apply conditioning on edges outside $\res(\al)$. In this way
the norm of subshapes $\beta$ will be guaranteed to be less than $\al$. The fact that it's possible to reserve such edges is shown separately for the different kind of shapes we encounter in our analysis (see \cref{sec:reserved}).

Finally,
we are forced to include certain dense shapes that encode
indicator functions of independent sets. These shapes must be
factored out and tracked separately throughout the analysis.\footnote{One may ask if we could definitionally get rid of dense shapes, like
we did for disconnected shapes via the connected truncation, but these dense
shapes are absolutely necessary.}
After handling all of these items we have shown that $\Lam$ is PSD.

\section{Pseudocalibration with connected truncation}
\label{sec:pseudocalibration}

\subsection{Pseudo-calibration}\label{subsec: pseudocalibration}

Pseudo-calibration is a heuristic introduced by \cite{BHKKMP16} to construct candidate moment matrices to exhibit SoS lower bounds. Here, we formally describe the heuristic briefly.

To pseudocalibrate, a crucial starting point is to come up with a hard hypothesis testing problem.
Let $\nu$ denote the null distribution and $\mu$ denote the alternative distribution. Let $v$ denote the input and $x$ denote the variables for our SoS relaxation. 
The main idea of pseudo-calibration is that, for an input $v$ sampled from $\nu$ and any polynomial $f(x)$ of degree at most the SoS degree, and for any low-degree test $g(v)$, the correlation of $\pE[f(x)]$ with $g$ as functions of $v$ should match in the planted and random distributions. That is,
\[\EE_{v \sim \nu}[\pE[f(x)]g(v)] = \EE_{(x, v) \sim \mu}[f(x)g(v)]\]

Here, the notation $(x, v) \sim \mu$ means that in the planted distribution $\mu$, the input is $v$ and $x$ denotes the planted structure in that instance. For example, in independent set, $x$ would be the indicator vector of the planted independent set.

Let $\calF$ denote the Fourier basis of polynomials for the input $v$. By choosing $g$ from $\calF$ such that the degree is at most the truncation parameter $D_V = \truncparam$ (hence the term ``low-degree test''), the pseudo-calibration heuristic specifies all lower order Fourier coefficients for $\pE[f(x)]$ as a function of $v$. The heuristic suggests that the higher order coefficients are set to be $0$. In total the candidate pseudoexpectation operator can be written as
\[\pE [f(x)] = \sum_{\substack{g \in \calF:\\deg(g) \le D_V}} \EE_{v \sim \nu}[\pE[f(x)]g(v)] g(v) = \sum_{\substack{g \in \calF:\\deg(g) \le D_V}} \EE_{(x, v) \sim \mu}[f(x)g(v)] g(v)\]

The Fourier coefficients $\EE_{(x, v) \sim \mu}[f(x)g(v)]$ can be explicitly computed in many settings, which therefore gives an explicit pseudoexpectation operator $\pE$.

An advantage of pseudo-calibration is that this construction automatically satisfies some nice properties that the pseudoexpectation $\pE$ should satisfy. It's linear in $v$ by construction. For all polynomial equalities of the form $f(x) = 0$ that are satisfied in the planted distribution, it's true that $\pE[f(x)] = 0$. For other polynomial equalities of the form $f(x, v) = 0$ that are satisfied in the planted distribution, the equality $\pE[f(x, v)] = 0$ is approximately satisfied. In many cases, $\pE$ can be mildly adjusted to satisfy these exactly.

The condition $\pE[1] = 1$ is not automatically satisfied.
In all previous successful applications of pseudo-calibration, $\pE[1] = 1 \pm \littleoh(1)$. Once we have this, we simply set our final pseudoexpectation operator to be $\pE'$ defined as $\pE'[f(x)] = \pE[f(x)] / \pE[1]$.
We remark that the condition $\pE[1] = 1 \pm \littleoh(1)$ has been quite successful in predicting distinguishability of the
planted/null distributions~\cite{hop17, hop18}.

\subsection{The failure of "Just try pseudo-calibration"}
Despite the power of pseudocalibration in guiding the construction of a pseudomoment matrix in SoS lower bounds, it heavily relies upon a planted distribution that is hard to distinguish from the null distribution. Unfortunately, a ''quiet'' planted distribution remains on the search in our setting.


Towards this end, we will consider the following "na\"ive" planted distribution $\calD_{pl}$ that likely would have been many peoples' first guess:\begin{enumerate}[(1)]
    \item Sample a random graph $G\sim G_{n,p}$;
    \item Sample a subset $S \subseteq [n]$ by picking each vertex with probability $\frac{k}{n}$;
    \item Let $\Tilde{G}$ be $G$ with edges inside $S$ removed, and output $(S,\Tilde{G})$.
\end{enumerate}

\begin{proposition}
    For all $d = O(\sqrt{n})$ and $k = \Omega(n/d)$, the na\"ive planted distribution $\calD_{pl}$
    is distinguishable in polynomial time from $G_{n,p}$ with $\Omega(1)$ probability.
\end{proposition}
\begin{proof}
    The number of labeled 4-cycles in $G_{n,d/n}$ has expectation $\E[C_4] = \frac{d^4}{n^4}\cdot n(n-1)(n-2)(n-3)$ and variance $\Var{C_4} = O(d^4)$.
    In $\calD_{pl}$ the expected number of labeled 4-cycles is 
    \begin{align*}
    \E_{pl}[C_4] &= \frac{d^4}{n^4}\cdot n(n-1)(n-2)(n-3)\cdot \left(\left(1 - \frac{k}{n}\right)^4 +  
    4\frac{k}{n}\left(1-\frac{k}{n}\right)^3 + 2\left(\frac{k}{n}\right)^2\left(1-\frac{k}{n}\right)^2\right)\\
    &= \frac{d^4}{n^4}\cdot n(n-1)(n-2)(n-3)\cdot \left(1 - O(k^2/n^2)\right)\\
    &= \E[C_4] - O(d^4k^2/n^2) = \E[C_4] - O(d^2)
    \end{align*}
    Since this is less than $\E[C_4]$ by a factor on the order of $\sqrt{\Var{C_4}}$,
    counting 4-cycles succeeds with constant probability.
\end{proof}
\begin{remark}
    To beat distinguishers of this type, it may be possible to construct the 
    planted distribution from a sparse quasirandom graph (in the sense that all small subgraph counts match $G_{n,p}$ to leading order) 
    which has an independent set size of $\Omega(n/\sqrt{d})$.
    In the dense setting, the theory of quasirandom graphs states that if the planted distribution and $G_{n,1/2}$ match the subgraph count of $C_4$, this is sufficient for all subgraph counts to match;
    in the sparse setting this is no longer true and the situation is more complicated~\cite{CG02}.
\end{remark}

\begin{remark}
    Coja-Oghlan and Efthymiou \cite{COE15} show that a slight modification of this planted distribution (correct the expected number of edges to be $p\binom{n}{2}$) is indistinguishable from the random distribution
    provided $k$ is slightly smaller than the size of the maximum independent set in $G_{n,p}$.
    This is not useful for pseudocalibration because we are trying to plant an independent set with larger-than-expected size.
\end{remark}

The Fourier coefficients of the planted distribution are:
\begin{lemma}
    Let $x_T(S)$ be the indicator function for $T$ being in the planted solution i.e. $T \subseteq S$. Then, for all $T \subseteq [n]$ and $\al \subseteq \binom{[n]}{2}$,
    \begin{align*}
        \E_{(S,\tilde{G}) \sim \calD_{pl}}[x_T(S)\cdot \chi_\alpha(\tilde{G})] = 
    \left(\frac{k}{n}\right)^{|V(\alpha)\cup T|}\left(\frac{p}{1-p}\right)^{\frac{|E(\alpha)|}{2}}
\end{align*}
\end{lemma}

\begin{proof}
First observe that if any vertex of $V(\al) \cup T$ is outside $S$, then the expectation is $0$. This is because either $T$ is outside $S$, in which case $x_T = 0$, or a vertex of $\al$ is outside $S$, in which case the expectation of any edge incident on this vertex is $0$ so the entire expectation is $0$ using independence. Now, each vertex of $V(\al) \cup T$ is in $S$ independently with probability $\left(\frac{k}{n}\right)^{|V(\alpha)\cup T|}$. Conditioned on this event happening, the character is simply $\chi_{\al}(0) = \left(\frac{p}{1-p}\right)^{\frac{|E(\alpha)|}{2}}$. Putting them together gives the result.
\end{proof}

This planted distribution motivates the following incorrect definition of the pseudo-expectation operator.
\begin{definition}[Incorrect definition of $\pE$]
For $S \subseteq [n], \abs{S} \leq \dsos$,
\[ \pE[x^S] \defeq \sum_{\substack{\alpha \subseteq \binom{[n]}{2}:\\\abs{\alpha} \leq D_V}} \left(\frac{k}{n}\right)^{\abs{V(\al) \cup S}} \left(\frac{p}{1-p}\right)^{\frac{\abs{E(\al)}}{2}}\chi_\al(G)\]
\end{definition}
For this operator, $\pE[1] = 1 + \Omega(1)$. More generally, tail bounds
of graph matrix sums are not small, which ruins our analysis technique.

\subsection{Salvaging the doomed}

We will now discuss our novel truncation heuristic. The next paragraphs are for discussion and the technical proof resumes at \cref{def:moment-matrix}. Letting $\calS$ be the set of shapes $\alpha$ that contribute to the moment $\pE[x^I]$, the ``connected truncation'' restricts $\calS$ to shapes $\al$ such that all vertices in $\al$ have a path to $I$ (or in the shape view, to $U_\al \cup V_\al$).

Why might this be a good idea? Consider the planted distribution. 
The only tests we know of to distinguish the random/planted distributions 
are counting the number of edges or counting occurrences of small subgraphs.
These tests cannot be implemented using only connected Fourier characters; 
shapes with disconnected middle vertices are needed to count small 
subgraphs.
For example, suppose we fix a particular vertex $v\in V(G)$,
and we consider the set of functions on $G$ which are allowed to depend
on $v$ but are otherwise symmetric in the vertices of $G$.
A basis for these functions can be made by taking shapes $\alpha$ such 
that $U_\al$ has a single vertex, $V_\al = \emptyset$, then
fixing $U_\al = v$ (i.e. take the vector entry in row $v$).
Let $\symfourier_v$ denote this set of functions. 
\begin{proposition}
    Let $T(G)$ be the triangle counting function on $G$. 
    Let $\conn(\symfourier_v)$ be the subset of $\symfourier_v$ such that
    all vertices have a path to $v$. Then $T(G) \not \in\linspan(\conn(\symfourier_v))$.
\begin{proof}
    $T(G)$ has a unique representation in $\symfourier_v$ which requires disconnected
    Fourier characters.
    For example, one component of the function is $T_{\overline{v}}(G) = $ number of triangles not containing $v$.
    This is three vertices $x,y,z$ outside of $U_\al$ with $1_{(x,y) \in E(G)}1_{(x,z) \in E(G)}1_{(y,z) \in E(G)}$.
    The edge indicator function is implemented by~\cref{lem:independent-set-indicator}.
    But there are no edges between $x,y,z$ and $U_\al = v$ in these shapes.
    
    It's not even possible to implement $T_v(G) = $ number of triangles containing $G$,
    as a required shape is two vertices $x,y$ outside of $U_\al$ connected with an edge.
\end{proof}
\end{proposition}

Despite the truncation above, our pseudocalibration 
operator is not a ``local function" of the graph.
Our $\pE[x^S]$ can depend on vertices that are far away from $S$, but in an attenuated
way. The graph matrix for $\alpha$ is a sum of all ways of overlaying the vertices of $\al$
onto $G$. The edges do not need to be overlaid. If an edge ``misses'' in $G$,
then we can use this edge to get far away from $S$, but we take a decay factor of $\chi_e(0) = \sqrt{\frac{p}{1-p}} = \sqrt{\frac{d}{n-d}}$.

The ``local function'' property of $\pE[x^S]$ is also a connected truncation, but
it is a connected truncation in a different basis.
The basis is the 0/1 basis of $1_H(G)$ for $H \subseteq \binom{n}{2}$.
A reasonable definition of graph matrices in this basis is
\[M_\al = \sum_{\text{injective }\sigma : V(\al) \to [n]} 1_{\sigma(E(\al))}(G) \]
which sums all ways to embed $\al$ into $G$.
$\pE[x^S]$ is a local function if and only if in this basis it is a sum of shapes $\al$ satisfying the (same) condition that all vertices are connected to $S = U_\al \cup V_\al$.

For sparse graphs, the two bases are somewhat heuristically interchangeable since:
\[ \frac{1}{p}1_{e}(G) = 1 - \sqrt{\frac{1-p}{p}}\chi_e(G) \approx - \sqrt{\frac{1-p}{p}}\chi_e(G).\]

Comparing the 0/1 basis and the Fourier basis, the 0/1 basis expresses 
combinatorial properties such as subgraph counts more nicely, while spectral analysis is only feasible in the Fourier basis.
In the proof (see \cref{def:augmented-ribbon}), we will augment ribbons so that they may also contain 0/1 indicators, and this flexibility helps us overcome both the spectral and
combinatorial difficulties in the analysis.

Formally we define the candidate moment matrix as:
\begin{definition}[Moment matrix]
\label{def:moment-matrix}
\begin{align*}
\Lambda &\defeq \displaystyle\sum_{\alpha \in \calS} \left(\frac{k}{n}\right)^{\abs{V(\alpha)}} \cdot \left(\frac{p}{1-p}\right)^{\frac{\abs{E(\alpha)}}{2}} \frac{M_\alpha}{\abs{\Aut(\al)}}.
\end{align*}
where $\calS$ is the set of proper shapes such that\begin{enumerate}
    \item $\abs{U_\alpha}, \abs{V_\alpha} \leq \dsos$.
    \item $\abs{V(\alpha)} \leq D_V$ where $D_V = \truncparam$ for a sufficiently large constant $C$.
    \item Every vertex in $\alpha$ has a path to $U_\alpha \cup V_\alpha$.
\end{enumerate}
\end{definition}
We refer to $\left(\frac{k}{n}\right)^{\abs{V(\alpha)}}$ as the ``vertex
decay factor'' and $\left(\frac{p}{1-p}\right)^{\frac{\abs{E(\alpha)}}{2}}$ as
the ``edge decay factor''.
The candidate moment matrix is the principal submatrix indexed by $\binom{[n]}{\leq \dsos/2}$.
The non-PSDness properties of $\Lam$ are easy to verify:



\begin{lemma}
    $\Lam$ is SoS-symmetric.
\end{lemma}
\begin{proof}
    This is equivalent to the fact that the coefficient of $\alpha$ 
    does not depend on how $U_\al \cup V_\al$ is
    partitioned into $U_\al$ and $V_\al$ for $\abs{U_\al \cup V_\al} \leq \dsos$.
\end{proof}

\begin{lemma}
$\pE[1] = 1$.
\end{lemma}
\begin{proof}
    For $U_\al = V_\al = \emptyset$, the only shape in which all vertices
    are connected to $U_\al\cup V_\al$ is the empty shape. Therefore $\pE[1] = 1$.
\end{proof}

\begin{lemma}
$\pE$ satisfies the feasibility constraints.
\end{lemma}
\begin{proof}
We must show that $\pE[x^S] = 0$ whenever $S$ is not an independent set in $G$.
Observe that if ribbon $R$ contributes to $\pE[x^S]$, then if we modify the set
of edges inside $S$, the resulting ribbon still contributes to $\pE[x^S]$.
In fact, each edge also comes with a factor of $\sqrt{\frac{p}{1-p}}$.
By \cref{lem:independent-set-indicator}, we can group these ribbons into an indicator function $\frac{1}{(1-p)^{\binom{|S|}{2}}}{1}_{\overline{E(S)}}$.
That is, $\pE[x^S] = 0$ if $S$ has an edge.
\end{proof}

\begin{lemma}
    With probability at least $1-o_n(1)$, $\pE$ has objective value $\pE[\sum x_i] \geq (1-o(1))k$.
\end{lemma}
\begin{proof}
\begin{claim}\[
\E[\pE[x_i]]=\frac{k}{n}
\]
\end{claim}
\begin{proof}
The only shape that survives under expectation is the shape with one vertex, and it comes with coefficient $\frac{k}{n}$.
\end{proof}
\begin{claim}
\[ 
\Var{\pE[x_i]}\leq d^{-\Omega(\eps)}
\]
\cnote{update this proof to not use $d^\eps$ vertex decay}
\end{claim}
\begin{proof}
Let Count$(v,e)$ be the number of shapes with $|U_\al| = 1, V_\al = \emptyset$, $v$ vertices and $e$ edges,
\begin{align*}
    \Var{\pE[x_i]}&=\sum_{\al\neq \emptyset:\text{connected to $i$}} \left(\left(\frac{k}{n}\right)^{\abs{V(\al)}}\left( \frac{p}{1-p}\right)^{\abs{E(\al)}/2}\right)^2\\
    &= \sum_{v=2}^{D_V} \sum_{v-1\leq e\leq v^2} 
    \left(\frac{k}{n}\right)^{2v}\left(\frac{p}{1-p}\right)^{e}\cdot \text{Count}(v,e)\cdot n^{v}\\
    &\leq O(1)\sum_{v=2}^{D_V}\left(\frac{4k\sqrt{d}}{(1-p)n}\right)^{v}\\
    &\leq O\left(\frac{k\sqrt{d}}{n}\right)
\end{align*}
where the first inequality follows by observing the dominant term is tree-like and bounding the number of trees (up-to-isomorphism) with $v$ vertices by $2^{2v}$.
\end{proof}
Hence, \[ 
\Pr[\abs{\sum \pE[x_i]-k}\geq t]\leq \frac{\Var{\sum \pE[x_i]}}{t^2}=\frac{k\sqrt{d} }{t^2}\leq o_n(1)
\]
where the last inequality follows by picking $t=n^{\frac{1}{2}+\delta}=o(k)$ for $\delta>0$.
\end{proof}

What remains is to show $\Lam \psdgeq 0$. To do this it's helpful to renormalize the matrix entries by multiplying the degree-($k,l)$ block 
by a certain factor.
\begin{definition}[Shape coefficient]
For all shapes $\alpha$, let 
\[\lambda_{\alpha} \defeq \left(\frac{k}{n}\right)^{\abs{V(\alpha)} - \frac{\abs{U_\al} + \abs{V_\al}}{2}} \cdot \left(\frac{p}{1-p}\right)^{\frac{\abs{E(\alpha)}}{2}}.\]
\end{definition}
It suffices to show $\sum_{\alpha \in \calS} \lambda_\alpha \frac{M_\alpha}{\abs{\Aut(\al)}} \psdgeq 0$ because left and right multiplying by a rank-1 matrix and its transpose returns $\Lda$.

\begin{proposition}
    \label{lem:coefficient-factor}
    If $\alpha, \beta$ are composable shapes, then 
    $\lam_\alpha \lam_\beta = \lambda_{\alpha \circ \beta}.$
\end{proposition}

\section{Norm bounds and factoring}\label{sec:sparse_graph_matrices_and_conditioning}

As explained in \cref{sec:proof-overview}, existing norm bounds in the literature are not tight enough for sparse graph matrices. We first obtain tighter norm bounds for sparse graph matrices.

\subsection{Norm bounds}


In this section, we show the spectral norm bounds for $M_{\al}$ in terms of simple combinatorial factors involving $\al$. With only log factor loss the norm bounds hold with very high probability (all but probability $n^{-\Omega(\log n)}$).
This is too tight of a probabilistic bound since it allows for polynomially rare events such as small dense subgraphs. 
We will need to use conditioning to improve
the norm bound for shapes with a lot of edges.

\begin{definition}
    For a shape $\al$, define a vertex separator (or a separator) $S$ to be a subset of vertices such that there is no path from $U_{\al}$ to $V_{\al}$ in the graph $\al \setminus S$.
\end{definition}

Roughly, the norm bounds for a proper shape $\al$ are:
\[ \norm{M_\al} \leq \widetilde{O}\left(\max_{\text{vertex separator }S}\left\{ \sqrt{n}^{\abs{V(\al)} - \abs{S}} \sqrt{\frac{1-p}{p}}^{\abs{E(S)}}\right\}\right).\]

The maximizer of the above is called the \emph{sparse minimum vertex separator.}
The proof of this bound uses the trace method, which also underlies graph matrix
norm bounds in the dense case. 
We defer the proof
to \cref{sec: graph_matrix_norm_bounds}.

To get norm bounds for improper shapes, we linearize the shape and take the largest norm bound among its linearizations.

\begin{definition}
    For a linearization $\beta$ of a shape $\alpha$, $E_{phantom}(\beta)$ is the set (not multiset) of ``phantom edges'' of $\alpha$ which are not in $\beta$. 
\end{definition}

For an improper shape $\al$:
\[\norm{M_\al} \leq \widetilde{O}\left(\max_{\beta, S}\left\{
\sqrt{n}^{|V(\al)| - |S| + |I_\beta|} \left(\sqrt{\frac{1-p}{p}}\right)^{|E(\al)| - |E(\beta)| - 2|E_{phantom}(\beta) + |E(S)|}
\right\}\right) \]
where $\beta$ is a linearization of $\alpha$ and $S$ is a separator of $\beta$.

\subsection{Factoring ribbons and shapes}

Because on the fact that norm bounds depend on vertex separators, we will
need to do some combinatorics on vertex separators of shapes.
The essential ideas presented in this section have appeared in prior works such as \cite{BHKKMP16, AMP20, PR20}, but we redefine them for convenience and to set up the notation for the rest of the paper.

\begin{definition}[Composing ribbons]
    Two ribbons $R, S$ are composable if $B_R = A_S$.
    The composition $R \circ S$ is the (possibly improper) ribbon 
    $T = (V(R) \cup V(S), E(R) \sqcup E(S), A_R, B_S)$.
\end{definition}

\begin{fact}
    \label{fact:factor-ribbons}
    If $R, S$ are composable ribbons, then $M_{R \circ S} = M_R M_S$.
\end{fact}

\begin{definition}[Composing shapes]
    Two shapes $\alpha, \beta$ are composable if $\abs{V_\alpha} = \abs{U_\beta}$.
    Given a bijection $\varphi: V_\alpha \to U_\beta$, the composition $\alpha \circ_\varphi \beta$ is the (possibly improper) shape $\zeta$ whose
    multigraph is the result of gluing together the graphs for $\alpha, \beta$
    along $V_\alpha$ and $U_\beta$ using $\varphi$. Set $U_\zeta = U_\alpha$ and $V_\zeta = V_\beta$.
\end{definition}

If we write $\al \circ \beta$ then we will implicitly assume that $\al$ and $\beta$ are composable and the bijection $\phi$ is given.
We would like to say that the graph matrix $M_{\alpha\circ\beta}$ also factors as $M_{\al} M_{\beta}$, but this
is not quite true. There are \textit{intersection terms}.

\begin{definition}[Intersection pattern]\label{def:intersection-pattern}
    For composable shapes $\alpha_1, \alpha_2, \dots, \alpha_k$, let $\alpha = \alpha_1 \circ \alpha_2 \circ \cdots \circ \alpha_k$.
    An intersection pattern $P$
    is a partition of $V(\alpha)$ such that for all $i$ and $v, w \in V(\alpha_i)$, $v$ and $w$ are not in the same block of the partition.
    \footnote{The intersection pattern also specifies the bijections $\varphi$ for composing the shapes $\alpha_1, \dots, \al_k$.}
    We say that a vertex ``intersects'' if its block has size at least 2 and let $\Int(P)$ denote the set of intersecting vertices.
    
    Let $\calP_{\alpha_1, \alpha_2, \dots, \alpha_k}$ be the set of intersection
    patterns between $\alpha_1, \alpha_2, \dots, \alpha_k$.
\end{definition}

\begin{definition}[Intersection shape]
    For composable shapes $\alpha_1, \alpha_2, \dots, \alpha_k$ and an intersection pattern
    $P~\in~\calP_{\alpha_1, \alpha_2, \dots, \alpha_k}$, let $\alpha_P = \alpha_1 \circ \alpha_2 \circ \cdots \circ \al_k$ then
    identify all vertices in blocks of $P$, i.e. contract them into a single super vertex. Keep all edges (and hence $\alpha_P$ may be improper).
\end{definition}

Composable ribbons $R_1, \dots, R_k$ with shapes $\al_1, \dots, \al_k$ induce 
an intersection pattern ${P \in \calP_{\al_1, \dots, \al_k}}$ based on which
vertices are equal. When multiplying graph matrices, by casing on which
vertices are equal we have:

\begin{proposition}
    For composable shapes $\alpha_1, \alpha_2, \dots, \alpha_k$,
    \[M_{\alpha_1} \cdots M_{\alpha_k} =  \sum_{P \in \calP_{\al_1, \dots, \al_k}}
    M_{\alpha_P}.\]
\end{proposition}
Note that different $P$ may give the same $\alpha_P$,
and hence the total coefficient on $M_{\alpha_P}$ for a given shape $\al_P$ is more complicated.




\begin{definition}[Minimum vertex separators]
    For a shape $\al$, a vertex separator $S$ is a minimum vertex separator (MVS) if it has the smallest possible size.
    MVS $S$ is the leftmost minimum vertex separator (LMVS) if
    $\al \setminus S$ cuts all paths from $U_\al$ to any other MVS.
    The rightmost minimum vertex separator (RMVS) likewise cuts paths from $V_\al$.
\end{definition}
    The LMVS and RMVS can be easily shown to be uniquely defined~\cite{BHKKMP16}.

\begin{definition}[Left shape]
    A shape $\sigma$ is a left shape if it is proper,
    the unique minimum
    vertex separator is $V_\sigma$, there are no edges
    with both endpoints in $V_\sigma$, and every vertex is connected to $U_{\sig}$.
\end{definition}

\begin{definition}[Middle shape]
    A shape $\tau$ is a middle shape if 
    $U_\tau$ is the leftmost minimum vertex
    separator of $\tau$, and $V_\tau$ is the rightmost minimum vertex
    separator of $\tau$.
    If $\tau$ is proper, we say it is a proper middle shape.
\end{definition}

\begin{definition}[Right shape]
    $\sigma'$ is a right shape if $\sigma'^\T$ is a left shape.
\end{definition}

We also extend the definition of (L/R)MVS, left, middle, and right to ribbons.
Every proper shape admits a canonical decomposition into left, right, and middle parts.

\begin{proposition}
    \label{lem:shape-decomposition}
    Every proper shape $\alpha$ has
    a unique decomposition $\alpha = \sigma \circ \tau \circ \sigma'^\T$, where $\sigma$ is a left
    shape, $\tau$ is a middle shape, and $\sigma'^\T$ is a right shape.
\end{proposition}

The decomposition takes $\sigma$ to be the set of vertices reachable
from $U_\al$ via paths that do not pass through the LMVS, similarly for $\sigma'$ vis-{\`a}-vis the RMVS,
and then $\tau$ is the remainder.
This decomposition respects the connected truncation:

\begin{proposition}\label{prop:middle-shape-connectivity}
    If all vertices in $\al$ have a path to $U_\al \cup V_\al$, then
    decomposing $\al = \sigma \circ \tau \circ \sigma'^\T$, all vertices
    in $\tau$ have a path to both $U_\tau$ and $V_\tau$.
\begin{proof}
    It suffices to show that there is a path to $U_\tau \cup V_\tau$.
    In this case, say there is a path to vertex $u \in U_\tau$, then $u$
    must have a path to $V_\tau$. Otherwise, $U_\tau \setminus\{u\}$ would
    be a smaller vertex separator of $\tau$ than $U_\tau$, a contradiction
    to $\tau$ being a middle shape.
    
    Let $v \in V(\tau) \subseteq V(\al)$. By assumption there is a path from
    $v$ to $u \in U_\al \cup V_\al$; without loss of generality, $u \in U_\al$.
    Since $v$ is not in $\sigma$, which was constructed by taking all vertices
    reachable from $U_\al$ without passing through $U_\tau$, the path must pass through $U_\tau$.
\end{proof}
\end{proposition}

\section{PSD-ness}
\label{sec:psdness}

To prove PSD-ness, we perform an approximate PSD decomposition of the moment matrix
$\Lda$. This type of decomposition originates from the planted clique
problem~\cite{BHKKMP16}. Our usage of it will be similar at a high level,
but the accounting of graph matrices is more complicated.

\subsection{Overview of the approximate PSD decomposition}
\label{sec:informal-decomposition}
Recall that our goal is to show that the matrix $\sum_{\alpha \in \calS} \lda_\alpha M_\alpha$ is PSD.
Each shape $\alpha$ decomposes into left, middle, and right parts
$\alpha = \sigma \circ \tau \circ \sigma'^\T$.
We have a corresponding approximate decomposition of the graph matrix for $\al$,
up to intersection terms,
\[M_\alpha = M_{\sigma \circ \tau \circ \sigma'^\T} 
\approx M_\sigma M_\tau M_{\sigma'}^\T. \]

The dominant term in the PSD decomposition of $\Lda$ (call it $\calH$) collects together $\alpha$ such that 
its middle shape $\tau$ is trivial. Since the coefficients $\lambda_\alpha$
also factor, this term is automatically PSD.
\[ \sum_{\substack{\alpha : \text{ trivial}}} \lambda_\alpha M_\alpha 
 = \sum_{\sigma, \sigma'} \lambda_{\sigma\circ\sigma'^\T} M_{\sigma\circ\sigma'^\T} \approx \sum_{\sigma, \sigma'} \lambda_\sigma \lambda_{\sigma'}M_\sigma M_{\sigma'}^\T = \left(\sum_\sigma \lambda_\sigma M_\sigma\right)\left(\sum_\sigma \lambda_\sigma M_\sigma\right)^\T =: \calH.\]
 
There are two types of things to handle: $\alpha$ with nontrivial $\tau$,
and intersection terms. For $\alpha$ with a fixed middle shape $\tau$,
these should all be charged directly to $\calH$ (include $\tau^\T$ so the matrix is symmetric):
\[ \left(\sum_{\sigma} \lambda_\sigma M_\sigma \right)\lambda_\tau (M_\tau + M_{\tau}^\T) \left(\sum_{\sigma} \lambda_\sigma M_\sigma \right)^\T  \psdleq \frac{1}{c(\tau)}\left(\sum_{\sigma} \lambda_\sigma M_\sigma \right)\left(\sum_{\sigma} \lambda_\sigma M_\sigma \right)^\T\]
where we leave some space $c(\tau)$. Intuitively this is possible
because nontrivial $\tau$ have smaller coefficients $\lda_\tau$
due to vertex/edge decay, and the norm of $M_\tau$ is controlled
due to the factorization of $\alpha$ into left, middle, and right parts.
This check amounts to 
$\sum_{\tau\text{ nontrivial}}\lambda_\tau M_\tau~\psdleq~o(1)\Id$.

The intersection terms need to be handled in a recursive way so that their norms
can be kept under control:
if $\sigma, \tau, \sigma'^\T$
intersect to create a shape $\zeta$, then we need to factor
out the non-intersecting parts of $\sigma$ and $\sigma'$ from the intersecting
parts $\gamma, \gamma'$. Informally writing
\[\zeta = (\sigma - \gamma) \circ \tau_P \circ (\sigma' - \gamma')^\T,\]
where $\tau_P$ is the intersection of $\gamma, \tau, \gamma'^\T$, we now perform a further factorization
\[ M_\zeta \approx M_{\sigma - \gamma} M_{\tau_P} M_{\sigma' - \gamma'}^\T \]
which recursively creates more intersection terms.
The point is that the sum over $\sigma - \gamma, \sigma' - \gamma'$ is equivalent
to the sum over $\sigma, \sigma'$ (up to truncation error).
\begin{align*} \sum_{\sigma - \gamma} \sum_{\sigma' - \gamma'}
\lambda_{\sigma \circ \tau \circ \sigma'^\T} 
M_{\sigma - \gamma} M_{\tau_P} M_{\sigma' - \gamma'}^\T
&= \left(\sum_{\sigma - \gamma} 
\lambda_{\sigma - \gamma} 
M_{\sigma - \gamma}\right) 
\lambda_{\gamma \circ \tau \circ \gamma'^\T}M_{\tau_P} \left(\sum_{\sigma - \gamma} \lambda_{\sigma - \gamma} M_{\sigma - \gamma}^\T\right)
\\
&= \left(\sum_{\sigma} 
\lambda_{\sigma} 
M_{\sigma}\right) 
\lambda_{\gamma \circ \tau \circ \gamma'^\T}M_{\tau_P} \left(\sum_{\sigma} \lambda_{\sigma} M_{\sigma}^\T\right) + \text{truncation error}.
\end{align*}
In summary, we have the following informal decomposition of the moment matrix,
\[ \Lam = \left(\sum_{\sigma} \lambda_{\sigma} M_{\sigma}\right)
\left(\id + \sum_{\tau\text{ nontrivial}}\lambda_\tau M_\tau + \sum_{\substack{\text{intersection terms}\\\tau_P \in \calP_{\gamma, \tau, \gamma'}}} \lambda_{\gamma \circ \tau \circ \gamma'^\T} M_{\tau_P}\right)
\left(\sum_{\sigma} \lambda_{\sigma} M_{\sigma} \right)^\T + \text{truncation error}.\]
We then need to compare the intersection terms $M_{\tau_P}$ with $\id$.
The vertex decay factors are compared with $\tau$ via the ``intersection
tradeoff lemma''. For the edge decay factors, we give new charging arguments.
The number of ways to produce $\tau_P$ as an intersection pattern needs to be 
bounded combinatorially.

Finally, there is the issue of truncation.
The total norm of the truncation error should be small. 
This can be accomplished by taking a large truncation parameter.

The matrix
\[\Id + \sum_{\tau\text{ nontrivial}}\lambda_\tau M_\tau + \sum_{\substack{\text{intersection terms}\\\tau_P \in \calP_{\gamma, \tau, \gamma'}}} \lambda_{\gamma \circ \tau \circ \gamma'^\T} M_{\tau_P}\] 
actually does have a nullspace, which prevents the above strategy from
working perfectly.
This is because $\Lam$ needs to satisfy the independent set constraints ($\pE[x^S] = 0$ if $S$ has an edge).
For independent set it's easy to factor out these constraints.
Instead of $\Id$, the leading term is a diagonal projection matrix ${\pi \defeq \sum_{\tau:V(\tau) = U_\tau = V_\tau}\lam_\tau M_\tau}$
and the non-dominant middle shapes are $\tau$ with $|V(\tau)| > \frac{|U_\tau| +  |V_\tau|}{2}$.


\subsection{Informal sketch for bounding $\tau$ and $\tau_P$}
\label{sec:informal}

The most important part of the overview given in the previous section is showing that middle shapes $\tau$
and intersection terms $\tau_P$ can be charged to the identity matrix,
i.e. a type of ``graph matrix tail bound'' for middle shapes and intersection terms.
In this subsection, we describe the properties that $\tau$ and $\tau_P$ satisfy, their coefficients, and their norm bounds. Using this, we show that for each individual non-trivial $\tau$, $\lambda_{\tau}\norm{M_{\tau}} \ll 1$ and for each intersection pattern $P$, 
$\lambda_{\gamma \circ \tau \circ {\gamma'}^\T}\norm{M_{\tau_P}} \ll 1$.
We then explain how these arguments are used to prove \cref{thm:informal-logd}.
The combinatorial arguments in this section crucially rely on the connected truncation.

\subsubsection{Middle shapes}

\begin{proposition}\label{prop:informal-middle-shapes}
For each middle shape $\tau$ such that $|V(\tau)| > \frac{|U_\tau| + |V_\tau|}{2}$
and every vertex in $\tau$ is connected to $U_{\tau} \cup V_{\tau}$, $\lambda_{\tau}\norm{M_{\tau}}~\ll~1$.
\end{proposition}
\begin{proof}
The coefficient $\lambda_{\tau}$ is 
\[
\lambda_{\tau} = \left(\frac{k}{n}\right)^{|V(\tau) | - \frac{|U_{\tau}| + |V_{\tau}|}{2}}\left(-\sqrt{\frac{p}{1-p}}\right)^{|E(\tau)|}
\]
The norm bound on $M_{\tau}$ is 
\[
\norm{M_\tau} \leq \widetilde{O}\left(n^{\frac{|V(\tau)\setminus S|}{2}}\left(\sqrt{\frac{1-p}{p}}\right)^{|E(S)|}\right)
\]
where $S$ is the sparse minimum vertex separator of $\tau$. We now have that $|\lambda_{\tau}|\norm{M_{\tau}}$ is 
\[
\widetilde{O}\left(\left(\frac{k}{\sqrt{n}}\sqrt{\frac{p}{1-p}}\right)^{|V(\tau)| - \frac{|U_\tau| + |V_\tau|}{2}}\left(\sqrt{\frac{1-p}{np}}\right)^{|S| - \frac{|U_{\tau}| + |V_{\tau}|}{2}}\left(\sqrt{\frac{p}{1-p}}\right)^{|E(\tau) \setminus E(S)| - |V(\tau) \setminus S|}\right).
\]
We claim that each of the three terms is upper bounded by 1.
This will prove the claim, since the first term provides a decay for $|V(\tau)| > \frac{|U_\tau| + |V_\tau|}{2}$.
The base of the first term is less than 1 for $k \lessapprox n/\sqrt{d}$.
The exponent of the second term is nonnegative: since $\tau$ is a middle shape,
both $U_\tau$ and $V_\tau$ are minimum vertex separators of $\tau$, and since
$S$ is a vertex separator, it must have larger size.
The exponent of the third term is nonnegative by the following lemma.
\begin{lemma}
    For a middle shape $\tau$ such that every vertex is connected to $U_\tau \cup V_\tau$, and any vertex separator $S$ of $\tau$,
    \[\abs{E(\tau) \setminus E(S)} \geq \abs{V(\tau) \setminus S}.\]
\end{lemma}
\begin{proof}
    First, we claim that every vertex in $\tau$ is connected to $S$.
    Let $v \in V(\tau)$, and by the connected truncation,
    $v$ is connected to some $u \in U_\tau$.
    Since $\tau$ is a middle shape, there must be a path from $u$ to $V_\tau$
    (else $U_\tau \setminus \{u\}$ would be a smaller vertex separator
    than $U_\tau$).
    The path necessarily passes through $S$ since $S$ is a vertex separator,
    therefore we now have a path from $v$ to $S$.

    Now, we can assign an edge of $E(\tau) \setminus E(S)$ to
    each vertex of $V(\tau) \setminus S$ to prove the claim. To do this, run a breadth-first search from $S$,
    and assign an edge to the vertex that it explores.
\end{proof}
\end{proof}

\subsubsection{Intersection terms}

Recall that intersection terms are formed by the intersection of $\sigma, \tau, \sigma'^\T$.
Then we factor out the non-intersecting parts, leaving the \emph{middle intersection} $\tau_P$, which is an intersection of some portion $\gamma$ of $\sigma$, the middle shape $\tau$, and some portion $\gamma'$ of $\sigma'$,
which we now make formal.

\begin{definition}[Middle intersection]\label{def:middle-intersection}
    Let $\gamma, \gamma'$ be left shapes and $\tau$ be a shape such
    that $\gam \circ \tau \circ \gam'^\T$ are composable.
    We say that an intersection pattern $P \in \calP_{\gam, \tau, \gam'^\T}$
    is a \emph{middle intersection} if
    $U_{\gamma}$ is a minimum vertex separator in $\gamma$ of $U_\gam$ and $V_\gam \cup \Int(P)$. Similarly, $U_{\gam'}$ is a minimum vertex separator
    in $\gam'$ of $U_{\gam'}$ and $V_{\gam'} \cup \Int(P)$.
    Finally, we also require that $P$ has at least one intersection.
    
    Let $\calP^{mid}_{\gam,\tau,\gam'}$ denote the set of middle intersections.
\end{definition}
\begin{remark}
For middle intersections we use the notation $\tau_P$ to denote the resulting 
shape, as compared to $\al_P$ which is used for an arbitrary intersection pattern.
\end{remark}
\begin{remark}\label{rmk:middle-intersections}
In fact this definition also captures recursive intersection terms which
are created from later rounds of factorization.
We say that ${P \in \calP_{\gam_{k}, \dots, \gam_1,\tau,\gam_1'^\T, \dots, \gam_{k}'^\T}}$ is a middle intersection,
denoted ${P \in \calP^{mid}_{\gam_k, \dots\gam_1,\tau,\gam_1',\dots,\gam_k'}}$,
if for all $j =0, \dots, k-1$, letting $\tau_{j}$ be the shape
of intersections so far between $\gam_j,\dots,\gam_1,\tau,\gam_1'^\T, \dots, \gam_j'^\T$, the intersection $\gam_{j+1}, \tau_{j}, \gam_{j+1}'^\T$ is a middle intersection.
\end{remark}

We need the following structural property of middle intersections.
\begin{proposition}\label{prop:intersection-connectivity}
    For a middle intersection $P \in \calP^{mid}_{\gam, \tau, \gam'}$ such that every
    vertex of $\tau$ is connected to both $U_\tau$ and $V_\tau$,
    every vertex in $\tau_P$ is connected to both $U_{\tau_P}$ and $V_{\tau_P}$.
    \begin{proof}
        By assumption,
        all vertices in $\tau$ have a path
        to both $U_\tau$ and $V_\tau$.
        Since $U_\tau = V_\gam$, and all vertices in $\gam$ have a path to $U_\gam$
        by definition of a left shape, all vertices in $\tau$ have a path to $U_\gam = U_{\tau_P}$.
        Similarly, $V_\tau = U_{\gamma'}$ is connected to $V_{\gamma'} = V_{\tau_P}$.
        Thus we have shown that all vertices in $\tau$ are connected to both
        $U_{\tau_P}$ and $V_{\tau_P}$.

        Vertices in $\gam$ have a path to $U_{\tau_P}$ by definition; we must
        show that they also have a path to $V_{\tau_P}$. A similar argument will
        hold for vertices in $\gam'.$
        Since all vertices in $\gam$ have paths to $U_\gam$, it suffices to show
        that all $u \in U_\gam$ have a path to $V_{\tau_P}$.
        There must be a path from $u$ to some $v \in V_\gam \cup \Int(P)$, 
        else this would violate the definition of a middle intersection by taking $U_\gam \setminus\{u\}$.
        The vertex $v$ is in either $\tau$ or $\gam'$, both of which are connected
        to $V_{\tau_P}$, hence we are done.
    \end{proof}
\end{proposition}

Now we show that middle intersections have small norm. We focus on the first level of intersection terms; the general case of middle intersections in \cref{rmk:middle-intersections} follows by induction.
\begin{proposition}\label{prop:informal-intersection-terms}
    For left shapes $\gam, \gam'$, proper middle shape $\tau$ such that every
    vertex in $\tau$ is connected to $U_\tau \cup V_\tau$, and a middle
    intersection $P \in \calP^{mid}_{\gam, \tau, \gam'}$,
    ${\abs{\lam_{\gam \circ \tau \circ \gam'^\T}}\norm{M_{\tau_P}} \ll 1}$.
\end{proposition}
\begin{proof}
For an intersection pattern $P$, the coefficient is 
\[
\lambda_{\gamma \circ \tau \circ {\gamma'}^\T} = \left(\frac{k}{n}\right)^{|V(\tau_P)| + i_{P} - \frac{|U_{\tau_P}| + |V_{\tau_P}|}{2}}\left(-\sqrt{\frac{p}{1-p}}\right)^{|E(\tau_P)|}
\]
where $i_{P} \defeq \abs{V(\gamma \circ \tau \circ \gamma'^\T)} - \abs{V(\tau_P)}$ is the number of intersections in $P$.

The norm bound on $M_{\tau_P}$ is 
\[
\widetilde{O}\left(\max_{\beta,S}\left\{n^{\frac{|V(\beta)| + |I_{\beta}| - |S|}{2}}\left(\sqrt{\frac{1-p}{p}}\right)^{|E(S)| + |E(\tau_P)| - |E(\beta)| -2|E_{phantom}|}\right\}\right)
\]
where the maximization is taken over all linearizations $\beta$ of $\tau_P$ and all separators $S$ for $\beta$.

We now have that $\abs{\lam_{\gam \circ \tau \circ \gam'^\T}}\norm{M_{\tau_P}}$ is
\begin{align*}
\widetilde{O}\Bigg(\max_{\beta,S}\Bigg\{&
\left(\frac{k}{\sqrt{n}}\sqrt{\frac{p}{1-p}}\right)^{|V(\gam \circ \tau \circ \gam'^\T)| - \frac{|U_\beta| + |V_\beta|}{2}}
\left(\sqrt{\frac{1-p}{np}}\right)^{i_{P} - |I_{\beta}| + |S| - \frac{|U_{\beta}| + |V_{\beta}|}{2} }\cdot\\
&\left(\sqrt{\frac{p}{1-p}}\right)^{|E(\beta) \setminus E(S)| - |V(\beta) \setminus S| - |I_{\beta}| + 2|E_{phantom}|}\Bigg\}\Bigg).
\end{align*}
We claim that each of the three terms is upper bounded by 1, which proves
the proposition since the first term provides a decay (since an intersection is nontrivial).
The base of the first term (vertex decay) is less than 1 for $k \lessapprox n/\sqrt{d}$.
The second term has a nonnegative exponent by the \emph{intersection tradeoff lemma} from~\cite{BHKKMP16}. The version we cite is Lemma 9.32 in \cite{PR20},
with the simplification that $\tau$ is a proper middle shape, so $I_\tau = \emptyset$
and $\abs{S_{\tau,min}} = \abs{U_\tau} = \abs{V_\tau}$ (the full form is used for intersection terms with $\gam_1, \dots, \gam_k, k>1$).
\begin{lemma}[Intersection tradeoff lemma]
\label{lem:intersection-tradeoff-lemma}
    For all left shapes $\gam, \gam'$ and proper middle shapes $\tau$,
    let $P \in \calP_{\gam, \tau, \gam'^\T}$ be a middle intersection,
    then
    \[|V(\tau)| - \frac{|U_\tau| + |V_\tau|}{2} + |V(\gam)| - |U_{\gam}| + |V(\gam')| - |U_{\gam'}| \ge |V(\tau_P)| + |I_{\tau_P, min}| - |S_{\tau_P, min}|\]
    where $S_{\al, min}$ is defined to be a minimum vertex separator of $\al$ with all multi-edges deleted, and $I_{\al, min}$ is the set of isolated vertices in $\al$ with all multi-edges deleted.
\end{lemma}
The above inequality is equivalent to
\[ i_P + \frac{|U_\tau| + |V_\tau|}{2} - |U_\gam| - |U_{\gam'}| \geq \abs{I_{\tau_P, min}} - \abs{S_{\tau_P, min}}.\]
Since $|U_\tau| \leq |U_\gam| = |U_\beta|, |V_\tau| \leq |U_{\gam'}| = |V_\beta|, \abs{I_{\tau_P, min}} \geq |I_\beta|, |S| \geq \abs{S_{\tau_P, min}}$, this implies
\[ i_P - |I_\beta| + |S| - \frac{|U_\beta| + |V_\beta|}{2} \geq 0.\]
The third term has a nonnegative exponent, as the following lemma shows.

\begin{restatable}{lemma}{intersectionCharging}
\label{lem:intersection-charging}
For any linearization $\beta$ of $\tau_P$ and all separators $S$ for $\beta$, 
\[
|E(\beta)| + 2|E_{phantom}| - |E(S)| \geq |V(\tau_P)| + |I_{\beta}|- |S|
.\]
\end{restatable}

\begin{proof}
We give a proof using induction. An alternate proof that uses an explicit charging scheme is given
in \cref{app:intersection-charging}.

We start with $E(\beta)$ and add in pairs of phantom edges one by one.
The final graph will have $\abs{E(\beta)} +2\abs{E_{phantom}}$ edges.
The inductive claim is that at any time in the process,
for all connected components $C$ in the current graph, 
\begin{enumerate}
\item If $C$ is not connected to $U_{\tau_P} = U_\beta$ or $V_{\tau_P} = V_\beta$, $|E(C) \setminus E(S)| + 2|E_{phantom}(C)| \geq |C \setminus S| + |I_{\beta} \cap C| - 2$. 
For example, $C$ may consist of a single isolated vertex or a set of isolated vertices connected by phantom edges.
\item If $C$ is connected to $U_{\tau_P}$ or $V_{\tau_P}$ but not both, $|E(C) \setminus E(S)| + 2|E_{phantom}(C)| \geq |C \setminus S| + |I_{\beta} \cap C| - 1$.
\item If $C$ is connected to both $U_{\tau_P}$ and $V_{\tau_P}$, $|E(C) \setminus E(S)| + 2|E_{phantom}(C)| \geq |C \setminus S| + |I_{\beta} \cap C|$
\end{enumerate}
To see that the base case holds for $E(\beta)$ and $E_{phantom}(C) = \emptyset$, let $C$ be a component $C$ of $E(\beta)$.
If $C$ is a single isolated vertex, we are in case (1), and the inequality is a tight equality.
Otherwise, $C$ has no isolated vertices, and therefore $\abs{E(C) \setminus E(S)} \geq \abs{C\setminus S} - 1$ holds by connectivity, which is sufficient
for $C$ in either case (1) or (2).
If $C$ is in case (3), $C$ is connected to both $U_{\tau_P}$ and $V_{\tau_P}$.
Since $S$ is a separator, in this case $C$
must include at least one vertex in $S$ and the stronger inequality $\abs{E(C) \setminus E(S)} \geq \abs{C \setminus S}$ holds as required
for the base case.

Now consider adding another phantom edge $e$. If $e$ is entirely contained within
a component $C$, the right-hand side of the inequality is unchanged, therefore the inequality
for $C$ continues to hold.
If $e$ goes between two components $C_1$ and $C_2$, there are six cases
depending on which case $C_1$ and $C_2$ are in. Let $C$ denote the new joined component. We show the argument when $C_1, C_2$ are both case (1) or both case (2).
The remaining cases are similar.

\noindent \textbf{(Both are case (1))} $C$ is connected to neither $U_{\tau_P}$ nor $V_{\tau_P}$, so we must show the inequality for case (1) still holds. Adding the inequalities for $C_1$ and $C_2$,
\begin{align*}
& \abs{E(C_1) \setminus E(S)} + 2\abs{E_{phantom}(C_1)} + \abs{E(C_2) \setminus E(S)} + 2\abs{E_{phantom}(C_2)} \\
\geq & \abs{C_1 \setminus S} + \abs{I_\beta \cap C_1} + \abs{C_2 \setminus S} + \abs{I_\beta \cap C_2} - 4.
\end{align*}
Equivalently, using $2|E_{phantom}(C)| = 2|E_{phantom}(C_1)| + 2|E_{phantom}(C_2)| + 2$, 
\begin{align*}
    & \abs{E(C) \setminus E(S)} + 2\abs{E_{phantom}(C)} - 2
\geq \abs{C \setminus S} + \abs{I_\beta \cap C} - 4\\
\iff & \abs{E(C) \setminus E(S)} + 2\abs{E_{phantom}(C)}
\geq \abs{C \setminus S} + \abs{I_\beta \cap C} - 2.
\end{align*}

\noindent \textbf{(Both are case (2))} 
$C$ is connected to either one or both of $U_{\tau_P}$ and $V_{\tau_P}$; 
we show the tighter inequality required by both, case (3).
Adding the inequalities for $C_1$ and $C_2$,
\begin{align*}
& \abs{E(C_1) \setminus E(S)} + 2\abs{E_{phantom}(C_1)} + \abs{E(C_2) \setminus E(S)} + 2\abs{E_{phantom}(C_2)} \\
\geq & \abs{C_1 \setminus S} + \abs{I_\beta \cap C_1} + \abs{C_2 \setminus S} + \abs{I_\beta \cap C_2} - 2.
\end{align*}
Again using $2|E_{phantom}(C)| = 2|E_{phantom}(C_1)| + 2|E_{phantom}(C_2)| + 2$,
\begin{align*}
    & \abs{E(C) \setminus E(S)} + 2\abs{E_{phantom}(C)} - 2
\geq \abs{C \setminus S} + \abs{I_\beta \cap C} - 2\\
\iff & \abs{E(C) \setminus E(S)} + 2\abs{E_{phantom}(C)}
\geq \abs{C \setminus S} + \abs{I_\beta \cap C}.
\end{align*}

At the end of the process, the connectivity of the graph
is the same as the graph $\tau_P$ (though nonzero edge multiplicities may
be modified).
By \cref{prop:middle-shape-connectivity}, every vertex in $V(\tau_P)$ is connected to both $U_{\tau_P}$ and $V_{\tau_P}$.
Therefore only case (3) holds. Summing the inequalities proves the lemma.
\end{proof}
\end{proof}

\subsubsection{Proof of \cref{thm:informal-logd}}

The preceding lemmas are informal, but they actually show something more that justifies \cref{thm:informal-logd}.
For technical reasons, our formal proof that implements \cref{sec:informal-decomposition} will break for $d \geq n^{0.5}$,
and thus it doesn't cover \cref{thm:informal-logd},
though we show that the arguments above plus the argument of \cite{BHKKMP16}
are already sufficient to prove \cref{thm:informal-logd}.

In the proof of \cref{prop:informal-middle-shapes} the term
$\lambda_\tau \norm{M_\tau}$ was broken into three parts, each less than 1.
Using just the first two parts,
\[ \abs{\lam_\tau} \norm{M_\tau} \leq \widetilde{O}\left(\left(\frac{k}{\sqrt{n}}\sqrt{\frac{p}{1-p}}\right)^{|V(\tau)\setminus S|} \left(\frac{k}{n}\right)^{\abs{S}- \frac{|U_\tau| + |V_\tau|}{2}}\right).\]
Letting $k = \frac{n}{d^{1/2+\eps}}$,
\begin{align*} \abs{\lam_\tau} \norm{M_\tau} &\leq \widetilde{O}\left(\left(\frac{1}{d^\eps}\right)^{|V(\tau)| - |S|} \left(\frac{1}{d^{1/2+\eps}}\right)^{\abs{S}- \frac{|U_\tau| + |V_\tau|}{2}}\right)\\
&\leq \widetilde{O}\left(\left(\frac{1}{d^\eps}\right)^{|V(\tau)| - |S_\tau|} \left(\frac{1}{d^{1/2+\eps}}\right)^{\abs{S_\tau}- \frac{|U_\tau| + |V_\tau|}{2}}\right)
\end{align*}
where $S_\tau$ is the dense MVS of $\tau$.
Observe that \emph{this equals $\lam_\tau\norm{M_\tau}$ in the dense case $(p = 1/2$) for a random graph of size $d$.}
    That is, suppose we performed the pseudocalibration and norm bounds
    from~\cite{BHKKMP16} for a random graph $G \sim G_{d,1/2}$.
    Then we would get $\lambda_\tau = \left(\frac{1}{d^{1/2+\eps}}\right)^{\abs{V(\tau)} - \frac{\abs{U_\tau} + \abs{V_\tau}}{2}}$ and $\norm{M_\tau} \leq (\log d)^{O(1)}\sqrt{d}^{\abs{V(\tau)} - \abs{S_\tau}}$, and $\lam_\tau \norm{M_\tau}$ is the same as the above.
    The main point is: by the analysis from~\cite{BHKKMP16},
    we know that for shapes up to size 
    $\Omega(\log d)$, the sum of these norms is $o(1)$.
    Therefore\footnote{The log factor is $(\log n)^{O(1)}$ for our norm bound, vs $(\log d)^{O(1)}$ for $G \sim G_{d,1/2}$, but these are the same up to
    a constant for $d \geq n^{\eps}$.} our matrices sum to $o(1)$ for SoS degree $\Omega(\log d)$.
    
    The same phenomenon occurs for the intersection terms. Essentially we are neglecting the effect of the edges and considering only the vertex factors. By taking the first
    two out of three terms used in the proof of \cref{prop:informal-intersection-terms}
    we have the bound ($k = \frac{n}{d^{1/2+\eps}}$)
    \begin{align*}
    \abs{\lam_{\gam \circ \tau \circ \gam'^\T}} \norm{M_{\tau_P}}  &\leq \widetilde{O}\left(\left(\frac{1}{d^\eps}\right)^{|V(\tau_P)| +|I_\beta| - |S|} \left(\frac{1}{d^{1/2+\eps}}\right)^{|V(\gam \circ \tau \circ \gam'^\T)| - |V(\tau_P)| - |I_\beta| + |S| - \frac{|U_\beta| + |V_\beta|}{2}}\right)\\
    & \leq \widetilde{O}\left(\left(\frac{1}{d^\eps}\right)^{|V(\tau_P)| +|I_{\tau_P, min}| - |S_{\tau_P, min}|} \left(\frac{1}{d^{1/2+\eps}}\right)^{|V(\gam \circ \tau \circ \gam'^\T)| - |V(\tau_P)| - |I_{\tau_P, min}| + |S_{\tau_P, min}| - \frac{|U_{\tau_P}| + |V_{\tau_P}|}{2}}\right)
    \end{align*}
    where $S_{\tau_P, min}$ and $I_{\tau_P, min}$ are the dense MVS and isolated vertices respectively in the graph $\tau_P$ after deleting all multiedges.
    For $G \sim G_{d,1/2}$, the pseudocalibration coefficient and norm bounds
    are
    \begin{align*}
        \lam_{\gam \circ \tau \circ \gam'^\T} & = \left(\frac{1}{d^{1/2+\eps}}\right)^{\abs{V(\gam \circ \tau \circ \gam'^\T)} - \frac{|U_\gam| + |U_{\gam'}|}{2}}\\
        \norm{M_{\tau_P}} & \leq (\log d)^{O(1)} \sqrt{d}^{|V(\tau_P)| - |S_{\tau_P, min}| + |I_{\tau_P, min}|}.
    \end{align*}
Multiplying these together, one can see that $\lam_{\gam \circ \tau \circ \gam'^\T}\norm{M_{\tau_P}}$ is the same. Therefore, the analysis of \cite{BHKKMP16} also shows that the sum of all intersection
terms is negligible.

By passing to the ``dense proxy graph'' $G \sim G_{d,1/2}$, we can essentially use~\cite{BHKKMP16} to deduce a degree-$\Omega(\log d)$ lower bound in the sparse case
just using connected truncation with the sparse norm bounds.
%
The analysis of \cite{BHKKMP16} is limited to SoS degree $\Omega(\log d)$ and $d \geq n^\eps$.
In order to push
the SoS degree up and the graph degree down in the remaining sections, we need to utilize conditioning and handle the combinatorial terms more carefully.

\subsection{Decomposition in terms of ribbons}

We now work towards the proof of~\cref{thm:main}. Recall that our goal is to show PSD-ness of the matrix $\sum_{\alpha \in \calS}\lam_\al \frac{M_\alpha}{\abs{\Aut(\al)}}$.

\begin{definition}[Properly composable]
    Composable ribbons $R_1, \dots, R_k$ are properly composable if there
    are no intersections beyond the necessary ones $B_{R_i} = A_{R_{i+1}}$.
\end{definition}

\begin{definition}[Proper composition]
    Given composable ribbons $R_1, \dots, R_k$ which are not necessarily
    properly composable, with shapes $\alpha_1, \dots, \al_k$, 
    let $R_1 \odot R_2\odot \cdots\odot R_k$
    be the shape of $\al_1 \circ \cdots \circ\al_k$.
    
    That is, $R_1 \odot \cdots \odot R_k$ is the shape
    obtained by concatenating the ribbons but \textit{not} collapsing
    vertices which repeat between ribbons, as if they were properly composable.
    Compare this with $R_1 \circ \cdots \circ R_k$, in
    which repetitions collapse.
\end{definition}

\begin{definition}[$\calL$ and $\calM$]
    Let $\calL$ and $\calM$ be the set of left and middle shapes in ${\calS}$.
    Shapes in these sets will be denoted $\sigma, \gamma \in \calL$ and $\tau \in \calM$.
    We abuse notation and let $L,G \in \calL, T \in \calM$ denote ribbons of shapes
    in $\calL, \calM$ (following the convention of using Greek letters for shapes and Latin letters
    for ribbons).
\end{definition}

Formalizing the process described in \cref{sec:informal-decomposition},
we have the following lemma. It is easier to formally manipulate unsymmetrized objects, ribbons, for the PSD decomposition and switch to symmetrized objects,
shapes, only when we need to invoke norm bounds.
The details are carried out in \cref{sec:formal-decomposition}.

\begin{restatable}{lemma}{fullRibbonDecomposition}(Decomposition in terms of ribbons).
\label{lem:fullribbondecomposition}
\begin{align*}
    &\sum_{\al \in \calS} \lam_\al \frac{M_\al}{\abs{\Aut(\al)}} = \sum_{R \in {\calS}}{\lam_R M_R} = \\
    &\left(\sum_{L \in \calL}{{\lam_L}M_L}\right)\left(\sum_{j=0}^{2\dsos}{(-1)^{j}\sum_{G_j,\ldots,G_1,T,G'_1,\ldots,G'_j}{\lam_{G_j \odot \ldots \odot G_1 \odot T \odot {G'_1}^{\T} \odot \ldots \odot {G'_j}^{\T}} M_{G_j \circ \ldots \circ G_1 \circ T \circ {G'_1}^{\T} \circ \ldots \circ {G'_j}^{\T}}}}\right)\left(\sum_{L \in \calL}{{\lam_L}M_L}\right)^\T \\
    &+ {\text{truncation error}}_{\text{too many vertices}} + {\text{truncation error}}_{\text{too many edges in one part}}
\end{align*}
where 
\begin{enumerate}
    \item
    \begin{align*}
        &{\text{truncation error}}_{\text{too many vertices}} =         -\sum_{\substack{L,T,L':\\|V(L \circ T \circ L'^\T)| > D_V,\\L,T,L'^\T \text{ are properly composable}}} \lam_{L \circ T \circ L'^\T} M_{L \circ T \circ L'^\T}\\
        &+ \sum_{j=1}^{2\dsos}{(-1)^{j+1}\sum_{\substack{L,G_j,\ldots,G_1,T,G'_1,\ldots,G'_j,L': \\ 
        |V(L \odot G_j \odot \ldots \odot G_1)| > D_V \text{ or} \\ |V({G'_1}^{\T} \odot \ldots \odot {G'_j}^{\T} \odot {L'}^{\T})| > D_V}}
        {\lam_{L \odot G_j \odot \ldots \odot G_1 \odot T \odot {G'_1}^{\T} \odot \ldots \odot {G'_j}^{\T} \odot {L'}^{\T}} M_{L}M_{G_j \circ \ldots \circ G_1 \circ T \circ {G'_1}^{\T} \circ \ldots \circ {G'_j}^{\T}}M_{L'}^T}}
    \end{align*}
    \begin{align*}
        &{\text{truncation error}}_{\text{too many edges in one part}} = \sum_{\substack{L,T,L': \\
        |V(L \circ T \circ L'^\T)| \leq D_V,\\
        |E_{mid}(T)| - |V(T)| > C{\dsos} \\ L, T, {L'}^{\T} \text{ are properly composable}}}
        {\lam_{L \odot T \odot {L'}^{\T}} M_{L \circ T \circ {L'}^T}} \\
        &+\sum_{j=1}^{2\dsos}{(-1)^{j}\sum_{\substack{L,G_j,\ldots,G_1,T,G'_1,\ldots,G'_j,L': \\
        |V(L \odot G_j \odot \ldots \odot G_1)| \leq D_V \text{ and } |V({G'_1}^{\T} \odot \ldots \odot {G'_j}^{\T} \odot {L'}^{\T})| \leq D_V \\ 
        |E_{mid}(G_j)| - |V(G_j)| > C{\dsos} \text{ or } |E_{mid}(G'_j)| - |V(G'_j)| > C{\dsos} \\
        L, (G_j \circ \ldots \circ G_1 \circ T \circ {G'_1}^{\T} \circ \ldots \circ {G'_j}^{\T}), {L'}^{\T} \text{ are properly composable}}}
        {\lam_{L \odot G_j \odot \ldots \odot {G'_j}^{\T} \odot {L'}^{\T}} M_{L \circ G_j \circ \ldots \circ {G'_j}^{\T} \circ {L'}^T}}}
    \end{align*}
    where for a ribbon $R$, $E_{mid}(R)$ is the set of edges of $R$ which are not contained in $A_R$, not contained in $B_R$, and are not incident to any vertices in $A_R \cap B_R$.
    \item in all of these sums, the ribbons $L,G_j,\ldots,G_1,T,G'_1,\ldots,G'_j,L'$ satisfy the following conditions:
\begin{enumerate}
    \item $T\in \calM$, $L,L' \in \calL$, and each $G_i, G_i' \in \calL$.
    \item $L, G_j,\ldots,G_1,T,{G'_1}^{\T},\ldots,{G'_j}^{\T}, L'^\T$ are composable.
    \item The intersection pattern induced by $G_j, \dots, G_1, T, G_1', \dots, G_j'$ is a middle intersection pattern.
    
    \item $|V(T)| \leq D_V$, $|V(G_j \odot \ldots \odot G_1)| \leq D_V$, and $|V({G'_1}^{\T} \odot \ldots \odot {G'_j}^{\T})| \leq D_V$
    \item Except when noted otherwise (which only happens for ${\text{truncation error}}_{\text{too many edges in one part}}$), all of the ribbons $G_j,\ldots,G_1,T,G'_1,\ldots,G'_j$ (but not necessarily $L, L'$) satisfy the constraint that ${|E_{mid}(R)| - |V(R)| \leq C{\dsos}}$.
\end{enumerate}
\end{enumerate}
\end{restatable}

\subsection{Factor out $\Pi$ and edges incident to $U_\al \cap V_\al$}
    
    We need to augment our graph matrices to include missing edge indicators.
    There are two reasons. First, the dominant term in the decomposition is
    a projection matrix with these indicators, rather than the identity matrix. Second,
    we need to factor out dependence on $U_\tau \cap V_\tau$. These vertices do not essentially participate in the graph matrix, since the matrix is block diagonal with a block for each assignment to $U_\al \cap V_\al$.
    (When proving things with graph matrices, it is smart to start by assuming $U_\al \cap V_\al = \emptyset$, then handle the case $U_\al \cap V_\al \neq \emptyset$).
    
    For the first issue, define the matrix $\Pi$.
    \begin{definition}
         Let $\pi \in \R^{\binom{n}{\leq\dsos}\times \binom{n}{\leq\dsos} }$ be 
         the projector to the independent set constraints. $\pi$ is a diagonal matrix with entries \[ 
        \pi[S,S] = \1{S \text{ is an independent set in }G}\]
    \end{definition}
    
    \begin{definition}
        Let $\Pi \in \R^{\binom{n}{\leq\dsos}\times \binom{n}{\leq\dsos}}$ be
        a rescaling of $\pi$ by
        \[\Pi[S, S] = \left(\frac{1}{1-p}\right)^{\binom{\abs{S}}{2}}\1{S \text{ is an independent set in }G} \]
    \end{definition}
    The rescaling satisfies $\E_{G \sim G_{n,p}}[\Pi[S, S]] = 1$.
    Recall that in \cref{lem:independent-set-indicator}, we had
    the function $1 + \sqrt{\frac{p}{1-p}}\chi_{\{e\}} = \frac{1}{1-p}1_{e \notin E(G)}$ which is $\frac{1}{1-p}$ times the indicator function for $e$ being absent from the input graph $G$.
    Since the coefficients $\lam_R$ come with $\sqrt{\frac{p}{1-p}}$ for each edge,
    we can group the edges inside $A_R, B_R$ into missing edge indicators
    for all edges inside $A_R, B_R$ -- that is, independent set indicators on $A_R, B_R$ which form the matrix $\Pi$.
    For all the sums of ribbons we consider, after grouping:
    \[ \sum_R \lam_R M_R = \Pi^{1/2}\left(\sum_{R : \atop E(A_R) = E(B_R) = \emptyset} \left(\frac{1}{1-p}\right)^{\binom{|A_R|}{2}/2 + \binom{|B_R|}{2}/2 - \binom{|A_R \cap B_R|}{2}}\lam_R M_R\right)\Pi^{1/2}.\]
    
    For the second issue, consider the function $1 + \sqrt{\frac{p}{1-p}}\chi_{\{e\}} = \frac{1}{1-p}1_{e \notin E(G)}$ again. The magnitude of this function is clearly bounded by $\frac{1}{1-p}$. However, if we try to bound the magnitude of this function term by term, we instead get a bound of $2$ because $1$ always has magnitude $1$ and if $e \in E(G)$ then $\sqrt{\frac{p}{1-p}}\chi_{\{e\}} = 1$, so the best bound we can give on the magnitude of $\sqrt{\frac{p}{1-p}}\chi_{\{e\}}$ is also $1$.

    This can become a problem if $U_\tau \cap V_\tau$ is large, in which case there are many possible subsets of edges incident to $U_\tau \cap V_\tau$ even if the rest of the shape is small. If we try to bound things term by term, we will get a factor of $2^{|U_\tau \cap V_\tau|}$ which may be too large. To handle this, our strategy will be to group terms into missing edge indicators, which gives a factor of $\left(\frac{1}{1-p}\right)^{|U_\tau \cap V_\tau|}$ per vertex instead.
This almost works, however due to the connected truncation there are some edge cases of \cref{lem:independent-set-indicator} where we do not have the term where $T$ is empty. We handle this using the following lemma.
\begin{lemma}\label{lem:quasi-indicators}
For any set of potential edges $E$,
\[
\sum_{E' \subseteq E: E' \neq \emptyset}{\left(\sqrt{\frac{p}{1-p}}\right)^{|E'|}\chi_{E'}} = \sum_{e \in E}{\left(\sqrt{\frac{p}{1-p}}\right)\chi_{\{e\}}\sum_{E' \subseteq E \setminus \{e\}}{\frac{1}{|E|\binom{|E|-1}{|E'|}}\left(\frac{1}{1-p}\right)^{|E'|}1_{\forall e' \in E', e' \notin E(G)}}}
\]
\end{lemma}
\begin{proof}
Observe that if we give an ordering to the edges of $E$ then we can write
\[
\sum_{E' \subseteq E: E' \neq \emptyset}{\left(\sqrt{\frac{p}{1-p}}\right)^{|E'|}\chi_{E'}} = \sum_{e \in E}{\left(\sqrt{\frac{p}{1-p}}\right)\chi_{\{e\}}\left(\prod_{e' \in E: e' > e}{\left(\frac{1}{1-p}\right)1_{e' \notin E(G)}}\right)}
\]
Taking the average over all orderings of the edges of $E$ gives the result. To see this, observe that if we take a random ordering, the probability of having a term $\left(\sqrt{\frac{p}{1-p}}\right)\chi_{\{e\}}\left(\frac{1}{1-p}\right)^{|E'|}1_{\forall e' \in E', e' \notin E(G)}$ is $\frac{1}{|E|\binom{|E|-1}{|E'|}}$ as there must be exactly $|E'|$ elements after $e$ (so $e$ must be in the correct position) and these elements must be $E'$.
\end{proof}
\begin{remark}
We take the average over all orderings of the edges $E$ so that our expression will be symmetric with respect to permuting the indices in $U_\tau \cap V_\tau$.
\end{remark}
Based on this lemma, we make the following definition.
\begin{definition}[Quasi-indicator function]
Given a set of edges $E \subseteq \binom{[n]}{2}$, we define the quasi-missing edge indicator function $q_E$ to be 
\[
q_E = (1-p)^{|E|+1}\sum_{E' \subseteq E}{\frac{1}{(|E|+1)\binom{|E|}{|E'|}}\left(\frac{1}{1-p}\right)^{|E'|}1_{\forall e \in E', e \notin E(G)}}
\] 
\end{definition}
\begin{proposition}
    $\sum_{E' \subseteq E: E' \neq \emptyset}{\left(\sqrt{\frac{p}{1-p}}\right)^{|E'|}\chi_{E'}} = \left(\frac{1}{1-p}\right)^{|E|}\sum_{e \in E}{\left(\sqrt{\frac{p}{1-p}}\right)\chi_{\{e\}}q_{E \setminus \{e\}}}$
\end{proposition}
\begin{remark}
$q_E$ is a linear combination of terms where for each edge $e \in E$, either there is a missing edge indicator for $e$ or $e$ is not mentioned at all. Moreover, we add the factor of $(1-p)^{|E|+1}$ to $q_E$ so that the coefficients in this linear combination are non-negative and have sum at most $1$. This is the only fact that we will use about $q_E$ when we analyze the norms of the resulting graph matrices.
\end{remark}

\begin{definition}[Augmented ribbon]\label{def:augmented-ribbon}
    An augmented ribbon is a vertex set $V(R) \subseteq [n]$ with two subsets $A_R, B_R \subseteq V(R)$, as well as a multiset of edge-functions.
    In our augmentation, each edge-function is either a (single edge) Fourier character, a (single edge) missing edge indicator, or a (subset of edges) quasi-missing edge indicator.
\end{definition}
\begin{definition}[Matrix for an augmented ribbon]
    The matrix $M_R$ has rows and columns indexed by all subsets of $[n]$, with entries:
    \[M_R[I,J] = \begin{cases}
        \displaystyle\prod_{\text{edge-functions }f\text{ on }S \subseteq V(R)}f(G|_S) & I = A_R, J = B_R\\
        0 & \text{Otherwise}
    \end{cases}\]
\end{definition}
Two augmented ribbons are isomorphic if they are in the same $S_n$-orbit (i.e. they are equal after renumbering the vertices), and we define an augmented shape for each orbit.
Equivalently, augmented shapes are equivalence classes of augmented ribbons under \emph{type-preserving} multi-hypergraph isomorphism, where the type of a hyperedge is the associated function.
As before, we can specify the shape by a representative graph with edge-functions.
The graph matrix $M_\al$ for an augmented shape $\al$ is still defined as the sum of $M_R$ over
injective embeddings of $\al$ into $[n]$.

\begin{definition}[Permissible]\label{def:permissible-ribbon}
We say that an augmented ribbon $R$ is permissible if the following conditions are satisfied:
\begin{enumerate}
\item All vertices $v \in V(R)$ are reachable from $A_R \cup B_R$ using the Fourier character edges.
\item There is a missing edge indicator for all edges within $A_R$ and a missing edge indicator for all edges within $B_R$.
\item For any vertex $v \in V(R) \setminus (A_R \cap B_R)$ which is reachable from $(A_R \cup B_R) \setminus (A_R \cap B_R)$ without passing through $(A_R \cap B_R)$, there is a missing edge indicator for all edges between $(A_R \cap B_R)$ and $v$.

\item For each connected component $C$ of $R \setminus (A_R \cap B_R)$ which is disconnected from $A_R \cup B_R$, there is precisely one pair of vertices $u \in A_R \cap B_R$ and $w \in C$ such that $(u,w) \in E(R)$ and $R$ contains the quasi-missing edge indicator $q_{\left((A_R \cap B_R) \times C\right) \setminus \{u,w\}}$ for all other edges between $A_R \cap B_R$ and $C$.
\item There are no other (quasi-)missing edge indicators, and the Fourier character edges are disjoint from the (quasi-)missing edge indicators.
\end{enumerate}
A shape is permissible if any (equivalently, all) of its ribbons are permissible.
\end{definition}

To explain the conditions, the first condition is the connected truncation. The second condition follows because we factored out independent set indicators. The third condition observes that if we have a vertex $v \in W_\al$ that is connected to $(U_{\al} \cup V_{\al}) \setminus (U_{\al} \cap V_{\al})$, then we can factor out a missing edge indicator for all edges that go from $v$ to $U_{\al} \cap V_{\al}$ because such edges do not affect the connectivity properties of $\al$. The fourth condition picks out components of $W_\al$ that do affect the connectivity and handles them via the definition of quasi-missing edge indicators as explained above.

\begin{remark}
    For the remainder of \cref{sec:psdness}, ribbon means ``augmented ribbon'' and shape means ``augmented shape''. $E(\al)$ refers to the multiset of Fourier characters edges in $\al$.
    We will write ``permissible $R \in \calS$'' to mean a permissible ribbon such that the corresponding non-augmented ribbon that includes all edges involved in any edge function is in $\calS$ (and similarly for other sets of ribbons/shapes).
\end{remark}

\begin{lemma}\label{lem:permissible-has-pi}
    If $\al$ is permissible, then $\pi M_\al = M_\al \pi = M_\al$.
    \begin{proof}
        $\al$ has missing edge indicators for all edges inside $U_\al$ and $V_\al$,
        hence $M_\al$ is zero on any rows or columns that are not independent sets.
    \end{proof}
\end{lemma}

We count the extra factors of $1/(1-p)$ with the following definition.
\begin{definition}
    For a permissible shape or ribbon $R$, let:
\[m(R) = \left(\frac{1}{1-p}\right)^{\binom{|A_R|}{2}/2 + \binom{|B_R|}{2}/2  - \binom{|A_R \cap B_R|}{2} + |V(R) \setminus (A_R \cup B_R)| \cdot |A_R \cap B_R|}.\]
\end{definition}

\begin{lemma}
    \begin{align*}
    & \sum_{R \in {\calS}}{\lam_R M_R} = \\
    &\Pi^{1/2}\left(\sum_{\text{permissible}\atop L \in \calL}{
    m(L){\lam_L}M_L}\right)\cdot\\
    &\left(\sum_{j=0}^{2\dsos}{(-1)^{j}\sum_{\text{permissible} \atop G_j,\ldots,G_1,T,G'_1,\ldots,G'_j}
    m(T)\left(\prod_{i=1}^j m(G_i)m(G_i')\right)
    {\lam_{G_j \odot \ldots \odot G_1 \odot T \odot {G'_1}^{\T} \odot \ldots \odot {G'_j}^{\T}} M_{G_j \circ \ldots \circ G_1 \circ T \circ {G'_1}^{\T} \circ \ldots \circ {G'_j}^{\T}}}}\right)\cdot\\
    &\left(\sum_{\text{permissible}\atop L \in \calL}
    m(L)
    {{\lam_L}M_L}\right)^\T \Pi^{1/2}\\
    &+ {\text{truncation error}}
\end{align*}
\end{lemma}
\begin{proof}
    The factoring process is completely independent of edges inside $A_T, B_T, A_{G_i}, B_{G_i}$ or edges incident to $A_R \cap B_R$ for any of these shapes,
    as noted in \cref{lem:factoring-ignores-indicators}.
    All subsets of edges inside $A_R$ or $B_R$ are present, but not all subsets of edges incident to $A_R \cap B_R$ are present; because of the connected truncation, we do not have the empty set of edges between $A_R \cap B_R$ and components of $R \setminus (A_R \cap B_R)$ which are disconnected from $A_R \cup B_R$.
    Now group these Fourier characters into missing edge indicators
    and quasi-missing edge indicators as explained above.
\end{proof}

We can show the factor $m(R)$ is a negligible constant:
\begin{lemma}\label{lem:bound-m}
    For any ribbon $R$ with degree at most $\dsos$,
    \[ m(R) \leq \left(\frac{1}{(1-p)^{2\dsos}}\right)^{|V(R) \setminus (A_R \cap V_R)|}.\]
    \begin{proof}
        The claim follows from the following two inequalities:
        \[ \binom{\abs{A_R}}{2}/2 + \binom{\abs{B_R}}{2}/2 - \binom{\abs{A_R \cap B_R}}{2} \leq \dsos (\abs{A_R \setminus (A_R \cap B_R)} + \abs{B_R \setminus (A_R \cap B_R)}) \]
        \[\abs{A_R \cap B_R}\abs{V(R) \setminus (A_R \cup B_R)} \leq \dsos \abs{V(R) \setminus (A_R \cap B_R)}.\]
    \end{proof}
\end{lemma}
This will be handled later by the vertex decay.

\subsection{Conditioning I:  reduction to sparse shapes}
\label{sec: conditioning}




To improve norm bounds for dense shapes, we can replace them by sparse ones.
The replacement of dense ribbons is accomplished by the following lemmas.

\begin{lemma}
$1_{e \in E(G)}\chi_{e} = -\sqrt{p(1-p)} + (1-p)\chi_{e}$
\end{lemma}
\begin{proof}
Follows by explicit computation, as in \cref{lem:independent-set-indicator}.
\end{proof}
\begin{lemma}\label{cor:conditioning-factorization}
Given a set of edges $E$, if we know that not all of the edges of $E$ are in $E(G)$ then 
\[
\chi_E = -\sum_{E' \subseteq E: E' \neq E}{\left(-\sqrt{\frac{p}{1-p}}\right)^{|E| - |E'|}\chi_{E'}}
\]
\end{lemma}
\begin{proof}
Since not all of the edges of $E$ are in $E(G)$, $(1 - 1_{E \subseteq E(G)})\chi_{E} = \chi_E$. Now observe that 
\[
1_{E \subseteq E(G)}\chi_E = \prod_{e \in E}{(1_{e \in E(G)}\chi_{e})} = \prod_{e \in E}{(-\sqrt{p(1-p)} + (1-p)\chi_{e})}
\]
Thus,
\[
\chi_E = (1 - 1_{E \subseteq E(G)})\chi_E = (1 - (1-p)^{|E|})\chi_{E} - \sum_{E' \subseteq E: E' \neq E}{(-\sqrt{p(1-p)})^{|E| - |E'|}(1-p)^{|E'|}\chi_{E'}}
\]
Solving for $\chi_{E}$ gives
\[
\chi_E = -\sum_{E' \subseteq E: E' \neq E}{\left(-\sqrt{\frac{p}{1-p}}\right)^{|E| - |E'|}\chi_{E'}}
\]
\end{proof}

\begin{corollary}\label{cor:conditioning-recursive}
Given a set of edges $E$, if we know that at most $k < |E|$ of the edges of $E$ are in $E(G)$, then
\[\chi_E = (-1)^{|E| - k}\sum_{E' \subseteq E:|E'| \le k} \binom{|E| - 1 - |E'|}{|E| - 1 - k} \left(-\sqrt{\frac{p}{1-p}}\right)^{|E| - |E'|}\chi_{E'}.\]
\end{corollary}
\begin{proof}
    Apply \cref{cor:conditioning-factorization} repeatedly.
\end{proof}

\begin{corollary}\label{cor:conditioning-edge-bound}
Given a set of edges $E$, if we know that at most $k < |E|$ of the edges of $E$ are in $E(G)$, then
\[\chi_E = \sum_{E' \subseteq E:|E'| \le k} c_{E'}\chi_{E'}\]
where $\abs{c_{E'}} \leq (2|E|\sqrt{p})^{|E| - |E'|}$.
\end{corollary}

We can do the replacement if the graph of the shape does not occur in our random
sample of $G \sim G_{n,p}$.
We formalize a bound on the density of subgraphs of $G_{n,p}$. 
In expectation the number of occurrences of a small subgraph $H$ is
essentially $n^{\abs{V(H)}}p^{\abs{E(H)}}$ and hence $H$ is unlikely to appear in a random graph sample if it has too many edges
(specifically, morally all subgraphs are sparse, 
$|E(H)| \leq |V(H)| \log_{1/p}(n)$).
To translate from expectation to concentration, we use a simple first moment calculation to bound the densest-$k$-subgraph in $G_{n, p}$.

\begin{proposition} [Sparsity of small subgraphs of $G_{n,p}$]
\label{lem:density-bound} 
For $G \sim G_{n,p}, p\leq \frac{1}{2}$, constant $\eta > 0$, with probability at least $1-O(1/n^\eta)$, every subgraph $S$ of $G$ such that $|V(S)| \leq \sqrt[3]{\frac{n}{d}}$
satisfies:
\[|E(S)| \leq 3|V(S)|\log_{1/p}(n) + 3 \eta\log_{1/p}(n).
\]
\begin{proof}
    Let $e^*(v) \defeq 3v\log_{1/p}(n) + 3\eta\log_{1/p}(n)$.
    \begin{align*}
        \Pr[\exists \text{too dense subgraph}] &\leq \sum_{v = 2}^{n}\sum_{e = e^*(v)}^{\binom{v}{2}} \E[\text{number of subgraphs with $v$ vertices, $e$ edges}]\\
        &\leq \sum_{v = 2}^{n}\sum_{e = e^*(v)}^{\binom{v}{2}} \binom{n}{v}v^{2e} p^{e}\\
        &\leq \sum_{v = 2}^{n}\sum_{e = e^*(v)}^{\binom{v}{2}} \frac{n^{v}}{v!} (p^{1/3})^{e}\\
        &= \sum_{v = 2}^{n}\sum_{e = e^*(v)}^{\binom{v}{2}} \frac{n^{v}p^{e^*(v)/3}}{v!} (p^{1/3})^{e - e^*(v)}\\
        &= \sum_{v = 2}^{n}\sum_{e = e^*(v)}^{\binom{v}{2}} \frac{1}{n^{\eta} \cdot v!} (p^{1/3})^{e - e^*(v)}\\
        &\leq O(1/{n^{\eta}}).
    \end{align*}
\end{proof}
\end{proposition}

\begin{definition}[Sparse]
    A graph $G$ is $C'$-sparse if $\abs{E(G)} \leq C'\abs{V(G)}$.
    A shape $\al$ is $C'$-sparse if the underlying graph (not multigraph) is.
    We will generally say that a graph/shape is sparse if it is $C'$-sparse for some $C'$.
\end{definition}

\begin{corollary}\label{cor:2-sparse}
    For $d \leq n^{0.5}$, w.h.p. every subgraph of $G\sim G_{n,\frac{d}{n}}$ of size at most $n^{0.16}$
    is 7-sparse.
\end{corollary}

\begin{definition}[Forbidden subgraph]
    We call a graph $H$ forbidden if it does not occur as a subgraph of $G$.
\end{definition}

The assumptions in \cref{thm:main} ensure that every graph of size at most $D_V$ that is not sparse is forbidden per \cref{cor:2-sparse}.
This is the only class of forbidden graphs that we will need in this work.

\begin{remark}
    For $d = n^{1-\eps}$, subgraphs up to size $n^{\Omega(\eps)}$ will be $O(1/\eps)$-sparse.
    Using this bound and our techniques, \cref{thm:main} can be extended to show a $n^{\Omega(\eps)}$ SoS-degree lower bound for $d = n^{1-\eps}$.
\end{remark}
\begin{remark}
    For $d \ll n^\eps$, subgraphs of $G_{n,p}$ will actually be significantly sparser.
    For example, it is well-known that $o(\log n)$-radius neighborhoods in $G_{n,d/n}$ for constant $d$
    are trees with at most one extra cycle.
\end{remark}

\begin{definition}[Subshape/subribbon, supershape/superribbon]
    We call shape $\beta$ a sub-shape of shape $\al$ if $V(\al)=V(\beta)$, $U_\al = U_\beta, V_\al = V_\beta$, and $E(\beta)\subseteq E(\al)$.
    Furthermore, the (quasi-)missing edge indicators of $\beta$ and $\al$ must be equal.
    Supershapes, subribbons, and superribbons are defined similarly.
    We write $\al \subseteq \beta$ if $\al$ is a subshape of $\beta$.
\end{definition}

\begin{definition}
    For a sparse subribbon $R$ of a forbidden ribbon $U$, let $c(R, U)$ be the 
    coefficient on $R$ after applying conditioning \cref{cor:conditioning-recursive} to the non-reserved edges of $U$.
\end{definition}

\begin{definition}
Let $\lam'_R$ incorporate the constant factors into $\lam_R$,\footnote{
When $R$ is a left $L$ /middle $T$/intersecting $G$ ribbon, the sum over $U$ should be restricted to $U$ with the same type, and which actually appear in the decomposition. The stated sum is an upper bound.
}
\begin{align*}
\lam'_R &\defeq m(R) \displaystyle\left(\sum_{U\supseteq R} c(R,U) \lam_U\right)\lam_R,\\
\lam'_{R_1 \odot \cdots \odot R_k} &\defeq \prod_{i = 1}^k \lam'_{R_i}.
\end{align*}
\end{definition}

\begin{lemma}\label{lem:ribbon-decomposition-sparse}
    \begin{align*}
    & \sum_{R \in {\calS}}{\lam_R M_R} =\\ &\Pi^{1/2}\left(\sum_{\substack{\text{sparse}\\
    \text{permissible}\\ L \in \calL}}{
    {\lam_L'}M_L}\right)\cdot\\
    &\left(\sum_{j=0}^{2\dsos}{(-1)^{j}\sum_{\substack{\text{sparse}\\
    \text{permissible}\\ G_j,\ldots,G_1,T,G'_1,\ldots,G'_j}}
    {\lam'_{G_j \odot \ldots \odot G_1 \odot T \odot {G'_1}^{\T} \odot \ldots \odot {G'_j}^{\T}} M_{G_j \circ \ldots \circ G_1 \circ T \circ {G'_1}^{\T} \circ \ldots \circ {G'_j}^{\T}}}}\right)\cdot\\
    &\left(\sum_{\substack{\text{sparse}\\
    \text{permissible}\\ L \in \calL}}
    {{\lam'_L}M_L}\right)^\T \Pi^{1/2}\\
    &+ {\text{truncation error}}
\end{align*}
\end{lemma}

The proof follows mostly by definition of $\lam'_R$ but there is an important detail. When we condition a left $L$/middle $T$/intersection $G$ ribbon, 
the result may no longer be a left/middle/intersection ribbon.
For example, the MVS might change drastically.
It turns out that we can avoid this by ``reserving'' $O(|V(R)|)$ edges
in the ribbon and only allowing the removal of edges outside of the reserved set. 
The reserved edges guarantee that the subribbon will continue to be a left/middle/intersection ribbon, at the cost of increasing the sparsity.
We prove that $O(|V(R)|)$ edges suffice in \cref{sec:reserved}.

We can show that the conditioning negligibly affects the coefficients of the ribbons.
\begin{lemma}\label{lem:conditioning-is-negligible}
Under the conditions of \cref{thm:main}, for any ribbon $R \in \calS$,
\[ \sum_{U \supset R}\lam_U c(R, U) = o(1) \lam_R.\]
\begin{proof}
Let $U$ be a superribbon with $k$ more edges than $R$.
We have 
\begin{equation}\label{middle-conditioning:eq:1}
\lam_U = \left(\sqrt{\frac{p}{1-p}}\right)^{k}\lam_U \leq (2\sqrt{p})^k \lam_R.
\end{equation}
Using the
bound from \cref{cor:conditioning-edge-bound},
\begin{equation}\label{middle-conditioning:eq:2}
    |c(R,U)| \le (2|E(U)|\sqrt{p})^k \le (2|V(R)|^2\sqrt{p})^k.
\end{equation}
The number of ribbons $U$ that
are superribbons of $R$ with $k$ extra edges is at most 
\begin{equation}\label{middle-conditioning:eq:3}
    |V(R)|^{2k}
\end{equation}
because (see \cref{sec:reserved}) due to the connectivity of the reserved set, the vertex set of $U$ and $R$ is the same.
Combining \cref{middle-conditioning:eq:1}, \cref{middle-conditioning:eq:2}, \cref{middle-conditioning:eq:3}, 
the change in coefficient is at most
\[ \sum_{k = 1}^\infty \lam_R  (4|V(R)|^4p)^k \leq \sum_{k=1}^\infty \lam_R\left(\frac{4D_V^4d}{n}\right)^k \leq o(1)\lam_R.\]
The last inequality uses $d \leq n^{0.5}$ and $D_V \leq \widetilde{O}(\dsos) \leq \widetilde{O}(n^{0.1})$.
\end{proof}
\end{lemma}

\subsection{Shifting to shapes}

For a proper middle shape $\tau$ and left shapes $\gam_j, \dots, \gam_1, \gam_1', \dots, \gam_j'$, recall the notation $\calP^{mid}_{\gam_j, \dots, \gam_1, \tau, \gam_1', \dots, \gam_j'}$
for the set of middle intersection patterns (\cref{def:middle-intersection} and \cref{rmk:middle-intersections}).
We also let $\calP^{mid}_\tau$ contain a single term, the non-intersecting singleton partition.

We now analyze
\[
\sum_{\substack{\text{sparse}\\
    \text{permissible}\\ G_j,\ldots,G_1,T,G'_1,\ldots,G'_j}}
    {\lam'_{G_j \odot \ldots \odot G_1 \odot T \odot {G'_1}^{\T} \odot \ldots \odot {G'_j}^{\T}} M_{G_j \circ \ldots \circ G_1 \circ T \circ {G'_1}^{\T} \circ \ldots \circ {G'_j}^{\T}}}
\]
We first partition this sum based on the shapes $\gam_j,\ldots,\gam_1,\tau,\gam'_1,\ldots,\gam'_j$ of $G_j,\ldots,G_1,T,G'_1,\ldots,G'_j$. We then partition this sum further based on the intersection pattern $P \in \calP^{mid}_{\gam_j, \dots, \gam_1, \tau, \gam_1', \dots, \gam_j'}$.

\begin{definition}\label{def:np}
Define $N_P(\tau_P)$ to be the number of ways to choose ribbons $G_j,\ldots,G_1,T,G'_1,\ldots,G'_j$ of shapes $\gam_j, \dots, \gam_1, \tau, \gam_1', \dots, \gam_j'$ so that they have intersection pattern $P$ and $G_j \circ \ldots \circ G_1 \circ T \circ {G'_1}^{\T} \circ \ldots \circ {G'_j}^{\T} = T_P$ for a given ribbon $T_P$ of shape $\tau_P$.

By symmetry, this is independent of the choice of $T_P$.
\end{definition}


Recall that an intersection pattern specifies both intersections and how the shapes should be glued together via bijections between the $V$ of each
shape and the $U$ of the following shape (\cref{def:intersection-pattern}).

\begin{definition}
We say that two intersection patterns are equivalent if there is an element of
\[
Aut(\gam_j) \times \ldots \times Aut(\gam_j) \times Aut(\tau) \times Aut({\gam_1'}^\T) \times \ldots \times Aut({\gam_j'}^\T)
\]
which maps one intersection pattern to the other (where the gluing maps are permuted accordingly).
\end{definition}

This gives us the following equation:
\begin{align*}
    &\sum_{\substack{\text{sparse}\\
    \text{permissible}\\ G_j,\ldots,G_1,T,G'_1,\ldots,G'_j}}
    {\lam'_{G_j \odot \ldots \odot G_1 \odot T \odot {G'_1}^{\T} \odot \ldots \odot {G'_j}^{\T}} M_{G_j \circ \ldots \circ G_1 \circ T \circ {G'_1}^{\T} \circ \ldots \circ {G'_j}^{\T}}} \\
    &= \sum_{\substack{\text{sparse}\\
    \text{permissible}\\ \gam_j,\ldots,\gam_1,\tau,\gam'_1,\ldots,\gam'_j}} 
    \sum_{\substack{\text{nonequivalent} \\ P \in \calP^{mid}_{\gam_j, \dots, \gam_1, \tau, \gam'_1, \dots, \gam'_j}}} 
    N_{P}(\tau_P)\lam'_{\gam_j \circ \cdots \circ \gam_1 \circ \tau \circ \gam_1'^\T \circ \cdots \circ \gam_j'^\T}
    \frac{M_{\tau_P}}{|\Aut(\tau_P)|}
\end{align*}

Grouping into shapes, here is the full decomposition:
\begin{lemma}[Decomposition in terms of shapes]
\label{lem:full-shape-decomposition}
    \begin{align*}
    & \sum_{R \in {\calS}}{\lam_R M_R} =\\ &\Pi^{1/2}\left(\sum_{\substack{\text{sparse}\\
    \text{permissible}\\ \sigma \in \calL}}{
    {\lam_\sigma'}\frac{M_\sigma}{|\Aut(\sigma)|}}\right)\cdot\\
    &\left(\sum_{j=0}^{2\dsos}{(-1)^{j}\sum_{\substack{\text{sparse}\\
    \text{permissible}\\ \gam_j,\ldots,\gam_1,\tau, \gam'_1,\ldots,\gam'_j }}
    \sum_{\substack{\text{nonequivalent} \\P \in \calP^{mid}_{\gam_j, \dots, \gam_1, \tau, \gam'_1, \dots, \gam'_j}}} 
    N_{P}(\tau_P)\lam'_{\gam_j \circ \cdots \circ \gam_1 \circ \tau \circ \gam_1'^\T \circ \cdots \circ \gam_j'^\T}
    \frac{M_{\tau_P}}{|\Aut(\tau_P)|}
    }\right)\cdot\\
    &\left(\sum_{\substack{\text{sparse}\\
    \text{permissible}\\ \sigma \in \calL}}{
    {\lam_\sigma'}\frac{M_\sigma}{|\Aut(\sigma)|}}\right)^\T \Pi^{1/2}\\
    &+ \Pi^{1/2}{\text{truncation error}_\text{too many vertices}}\Pi^{1/2} + \Pi^{1/2}\text{truncation error}_\text{too many edges in one part}\Pi^{1/2}
\end{align*}
where
\begin{enumerate}
    \item
    \begin{align*}
        &{\text{truncation error}}_{\text{too many vertices}} =         -\sum_{\substack{\text{sparse, permissible}\\ \sigma,\tau,\sigma':\\|V(\sigma \circ \tau \circ \sigma'^\T)| > D_V}} \lam'_{\sigma \circ \tau \circ \sigma'^\T} \frac{M_{\sigma \circ \tau \circ \sigma'^\T}}{|\Aut(\sigma \circ \tau \circ \sigma'^\T)|}\\
        &+ \sum_{j=1}^{2\dsos}{(-1)^{j+1}\sum_{\substack{\text{sparse, permissible} \\ \sigma,\gam_j,\ldots,\gam_1,\tau,\gam'_1,\ldots,\gam'_j,\sigma': \\ 
        |V(\sigma \circ \gam_j \circ \ldots \circ \gam_1)| > D_V \text{ or} \\ |V({\gam'_1}^{\T} \circ \ldots \circ {\gam'_j}^{\T} \circ {\sigma'}^{\T})| > D_V}}
        \sum_{\substack{\text{nonequiv.}\\ P \in \calP^{mid}_{\gam_j, \dots, \gam_j'}}} N_{P}(\tau_{P})
        {\lam'_{\sigma \circ \gam_j \circ \ldots \circ {\gam'_j}^{\T} \circ {\sigma'}^{\T}} \frac{M_\sigma}{|\Aut(\sigma)|} \frac{M_{\tau_{P}}}{|\Aut(\tau_P)|} \frac{M_{\sigma'}^\T}{|\Aut(\sigma')|}}}
    \end{align*}
    \begin{align*}
        &{\text{truncation error}}_{\text{too many edges in one part}} = \sum_{\substack{\text{sparse, permissible}\\\sigma,\tau,\sigma': \\
        |V(\sigma \circ \tau \circ \sigma'^\T)| \leq D_V,\\
        |E_{mid}(\tau)| - |V(\tau)| > C{\dsos}}}
        {\lam'_{\sigma \circ \tau \circ {\sigma'}^{\T}} \frac{M_{\sigma \circ \tau \circ {\sigma'}^T}}{|\Aut(\sigma \circ \tau \circ \sigma'^\T)|}} \\
        &+\sum_{j=1}^{2\dsos}{(-1)^{j}\sum_{\substack{\text{sparse, permissible} \\
        \sigma,\gam_j,\ldots,\gam_1,\tau,\gam'_1,\ldots,\gam'_j,\sigma': \\ 
        |V(\sigma \circ \gam_j \circ \ldots \circ \gam_1)| \leq D_V \text{ and } |V({\gam'_1}^{\T} \circ \ldots \circ {\gam'_j}^{\T} \circ {\sigma'}^{\T})| \leq D_V \\
        |E_{mid}(\gam_j)| - |V(\gam_j)| > C{\dsos} \text{ or } |E_{mid}(\gam'_j)| - |V(\gam'_j)| > C{\dsos}}}\sum_{\substack{\text{nonequiv.}\\P \in \calP^{mid}_{\gam_j, \dots, \gam_j'}}}
        {\lam'_{\sigma \circ \gam_j \circ \ldots \circ {\gam'_j}^{\T} \circ {\sigma'}^{\T}} \frac{M_{\sigma \circ \tau_{P} \circ \sigma'^\T}}{|\Aut(\sigma \circ \tau_P \circ \sigma'^\T)|} }}
    \end{align*}
    \item in all of these sums, the shapes $\sigma,\gam_j,\ldots,\gam_1,\tau,\gam'_1,\ldots,\gam'_j,\sigma'$ satisfy the following conditions:
\begin{enumerate}
    \item $\tau\in \calM$, $\sigma,\sigma' \in \calL$, and each $\gam_i, \gam_i' \in \calL$.
    \item $\sigma, \gam_j,\ldots,\gam_1,\tau,{\gam'_1}^{\T},\ldots,{\gam'_j}^{\T}, \sigma'^\T$ are composable.
    \item $|V(\tau)| \leq D_V$, $|V(\gam_j \circ \ldots \circ \gam_1)| \leq D_V$, and $|V({\gam'_1}^{\T} \circ \ldots \circ {\gam'_j}^{\T})| \leq D_V$
    \item Except when noted otherwise
    (which only happens for ${\text{truncation error}}_{\text{too many edges in one part}}$), all of the shapes $\gam_j,\ldots,\gam_1,\tau,\gam'_1,\ldots,\gam'_j$ (but not necessarily $\sigma ,\sigma'$) satisfy ${|E_{mid}(\al)| - |V(\al)| \leq C{\dsos}}$.
\end{enumerate}
\end{enumerate}
\end{lemma}

To analyze these expressions, we need upper bounds on the number of possible intersection patterns, and on $N_{P}(\tau_P)$.
To gain slightly more control we furthermore partition the intersection patterns
based on how many intersections occur.

\begin{lemma}\label{lem:count-nonequivalent-intersection-patterns}
There are at most 
\[
(4D_{SoS})^{|V(\tau_P)| - \frac{|U_{\tau_P}| + |V_{\tau_P}|}{2} + k}(3D_V)^{k}
\]
non-equivalent intersection patterns $P \in \calP^{mid}_{\gam_j, \dots, \gam_j'}$ which have exactly $k$ intersections.
\end{lemma}
\begin{proof}
    We first show a na\"{i}ve bound. To specify an intersection pattern, it is sufficient to specify how the shapes $\gam_j,\ldots,\gam_1,\tau,\gam'_1,\ldots,\gam'_j$ of $G_j,\ldots,G_1,T,G'_1,\ldots,G'_j$ glue together and which vertices intersect with each other. There are $\left(\prod_{i=1}^{j}{|V_{\gam_i}|!}\right)\left(\prod_{i=1}^{j}{|V_{\gam'_i}|!}\right)$ ways to glue the shapes together. To specify the intersections, we can give the following data.
    \begin{enumerate}
        \item For each $i$ and each vertex $v \in V(\gam_i) \setminus V_{\gam_i}$, specify whether or not $v$ intersects with a vertex in $\gam_{i-1}, \ldots, \gam_1, \tau, {\gam'_1}^T, \ldots, {\gam'_i}^T$. Similarly, for each vertex in $v \in V({\gam'_i}^T) \setminus U_{{\gam'_i}^T}$, specify whether it intersects with a vertex in $\gam_{i-1}, \ldots, \gam_1, \tau, {\gam'_1}^T, \ldots, {\gam'_{i-1}}^T$.
        
        The number of possibilities for this is at most 
        \[
        2^{\sum_{i=1}^{j}{|V(\gam_i) \setminus V_{\gam_i}| + |V(\gam'_i) \setminus V_{\gam'_i}|}}
        \]
        \item For each vertex which intersects with another vertex, specify a vertex which it intersects with. There are at most $3D_V$ choices for this and this happens $k$ times where $k$ is the number of intersections.
    \end{enumerate}
    This gives a total bound of 
    \[
    2^{\sum_{i=1}^{j}{|V(\gam_i) \setminus V_{\gam_i}| + |V(\gam'_i) \setminus V_{\gam'_i}|}}\left(\prod_{i=1}^{j}{|V_{\gam_i}|!}\right)\left(\prod_{i=1}^{j}{|V_{\gam'_i}|!}\right)(3D_V)^{k}
    \]
    To improve this bound, we observe that there may be some redundancy in the choice of gluing maps. In particular, given a permissible shape $\alpha$, let $Dormant(\alpha)$ be the set of vertices in $U_{\alpha} \cap V_{\alpha}$ which are only involved in missing edge indicators and quasi-missing edge indicators (and are thus indistinguishable). Instead of specifying how the vertices in $Dormant(\alpha)$ match up with the vertices in the previous shape and how these vertices match up with the vertices in the next shape, we can make an arbitrary choice for one of these gluings and make the appropriate choice for the other gluing. This saves us a factor of $|Dormant(\alpha)|!$
    
    Thus, we can divide our na\"{i}ve bound by 
    \[
    \left(\prod_{i=1}^{j}{|Dormant(\gam_i)|!}\right)|Dormant(\tau)|!\left(\prod_{i=1}^{j}{|Dormant(\gam'_i)|!}\right)
    \]
    which gives a bound of 
    \[
    \frac{2^{\sum_{i=1}^{j}{|V(\gam_i) \setminus V_{\gam_i}| + |V(\gam'_i) \setminus V_{\gam'_i}|}}\left(\prod_{i=1}^{j}{|V_{\gam_i}|!}\right)\left(\prod_{i=1}^{j}{|V_{\gam'_i}|!}\right)(3D_V)^{k}}{\left(\prod_{i=1}^{j}{|Dormant(\gam_i)|!}\right)|Dormant(\tau)|!\left(\prod_{i=1}^{j}{|Dormant(\gam'_i)|!}\right)}
    \]
    
    Since each connected component of $\alpha \setminus (U_{\alpha} \cap V(\alpha))$ can only have an edge to one vertex in $(U_{\alpha} \cap V(\alpha))$ and the components connected to $(U_{\alpha} \cup V(\alpha)) \setminus (U_{\alpha} \cap V(\alpha))$ have no such edges, we have that $|(U_{\alpha} \cap V_{\alpha}) \setminus Dormant(\alpha)| \leq |V(\alpha) \setminus (U_{\alpha} \cup V_{\alpha})|$. Thus, our bound is at most 
    \[
    \frac{2^{\sum_{i=1}^{j}{|V(\gam_i) \setminus V_{\gam_i}| + |V(\gam'_i) \setminus V_{\gam'_i}|}}\left(\prod_{i=1}^{j}{|V_{\gam_i}|!}\right)\left(\prod_{i=1}^{j}{|V_{\gam'_i}|!}\right)(3D_V)^{k}}
    {\left(\prod_{i=1}^{j}{\frac{|U_{\gam_i} \cap V_{\gam_i}|!}{D_{SoS}^{|V(\gam_i) \setminus (U_{\gam_i} \cup V_{\gam_i})|}}}\right)\left(\frac{|U_{\tau} \cap V_{\tau}|!}{D_{SoS}^{|V(\tau) \setminus (U_{\tau} \cup V_{\tau})|}}\right)\left(\prod_{i=1}^{j}{\frac{|U_{\gam'_i} \cap V_{\gam'_i}|!}{D_{SoS}^{|V(\gam'_i) \setminus (U_{\gam'_i} \cup V_{\gam'_i})|}}}\right)}
    \]
    To cancel out the $\left(\prod_{i=1}^{j}{|V_{\gam_i}|!}\right)\left(\prod_{i=1}^{j}{|V_{\gam'_i}|!}\right)$ factors, we use the following proposition.
\begin{proposition}
    For any composable shapes $\alpha_1,\ldots,\alpha_j$, 
    \[
    \prod_{i=1}^{j-1}{|V_{\alpha_i}|!} \leq D_{SoS}^{\sum_{i=1}^{j}{\frac{|U_{\alpha_i} \setminus V_{\alpha_i}| + |V_{\alpha_i} \setminus U_{\alpha_i}|}{2}}}\prod_{i=1}^{j}{|U_{\alpha_i} \cap V_{\alpha_i}|!}
    \]
\end{proposition}
\begin{proof}
    Observe that for all $i \in [j-1]$, 
    $|V_{\alpha_i}|! \leq D_{SoS}^{|V_{\alpha_i} \setminus U_{\alpha_i}|}|U_{\alpha_i} \cap V_{\alpha_i}|!$ so 
    \[
    \prod_{i=1}^{j-1}{|V_{\alpha_i}|!} \leq D_{SoS}^{\sum_{i=1}^{j-1}{|V_{\alpha_i} \setminus U_{\alpha_i}|}}\prod_{i=1}^{j}{|U_{\alpha_i} \cap V_{\alpha_i}|!} 
    \]
    Similarly, for all $i \in [j-1]$, 
    $|V_{\alpha_{i}}|! = |U_{\alpha_{i+1}}|! \leq D_{SoS}^{|U_{\alpha_{i+1}} \setminus V_{\alpha_{i+1}}|}|U_{\alpha_{i+1}} \cap V_{\alpha_{i+1}}|!$ so 
    \[
    \prod_{i=1}^{j-1}{|V_{\alpha_i}|!} \leq D_{SoS}^{\sum_{i=1}^{j-1}{|U_{\alpha_{i+1}} \setminus V_{\alpha_{i+1}}|}}\prod_{i=1}^{j}{|U_{\alpha_i} \cap V_{\alpha_i}|!} 
    \]
    Multiplying these two inequalities together and taking the square root gives the result.
\end{proof}
We now observe that for any shape $\alpha$, 
\[
|V(\alpha) \setminus (U_{\alpha} \cup V_{\alpha})| + \frac{|U_{\alpha} \setminus V_{\alpha}| + |V_{\alpha} \setminus U_{\alpha}|}{2} = |V(\alpha)| - \frac{|U_{\alpha}| + |V_{\alpha}|}{2}
\]
We also observe that for any shape $\alpha$, $V(\alpha) \setminus |V_{\alpha}| \leq 2\left(|V(\alpha)| - \frac{|U_{\alpha}| + |V_{\alpha}|}{2}\right)$. Putting everything together, our upper bound is 
\[
(4D_{SoS})^{\left(\sum_{i=1}^{j}{|V(\gam_i)| - \frac{|U_{\gam_i}| + |V_{\gam_i}|}{2}}\right) + \left(|V(\tau)| - \frac{|U_{\tau}| + |V_{\tau}|}{2}\right) +  \left(\sum_{i=1}^{j}{|V(\gam'_i)| - \frac{|U_{\gam'_i}| + |V_{\gam'_i}|}{2}}\right)}(3D_V)^{k}
\]
Since 
\begin{align*}
&\left(\sum_{i=1}^{j}{|V(\gam_i)| - \frac{|U_{\gam_i}| + |V_{\gam_i}|}{2}}\right) + \left(|V(\tau)| - \frac{|U_{\tau}| + |V_{\tau}|}{2}\right) +  \left(\sum_{i=1}^{j}{|V(\gam'_i)| - \frac{|U_{\gam'_i}| + |V_{\gam'_i}|}{2}}\right) \\
&= |V(\tau_P)| - \frac{|U_{\tau_P}| + |V_{\tau_P}|}{2} + k
\end{align*} 
the result follows.
\end{proof}

\begin{lemma}\label{lem:bound-np}
    For any intersection pattern $P$ which has exactly $k$ intersections, 
    \[
    N_{P}(\tau_P) \leq \frac{D_{SoS}^{2k + |V(\tau_P)| - \frac{|U_{\tau_P}| + |V_{\tau_P}|}{2}}}{|U_{\tau_P} \cap V_{\tau_P}|!}|\Aut(\tau_P)|
    \]
\end{lemma}
\begin{proof}
    Observe that if we are given a ribbon of shape $\tau_P$, an element of $Aut(\tau_P)$, the intersection pattern $P$ allows us to recover the original ribbons. This implies that $N_{P}(\tau_P) \leq |Aut(\tau_P)|$. However, this bound is not quite good enough as we need a factor of $\frac{1}{|U_{\tau_P} \cap V_{\tau_P}|!}$
    
    To obtain this factor, we observe that there may be some redundancy in the information provided by an element of $Aut(\tau_P)$.  If we let $Dormant(P)$ be the set of vertices $u$ in $U_{\tau_P} \cap V_{\tau_P}$ such that 
    \begin{enumerate}
        \item $u$ is not involved in any intersections
        \item $u$ is only incident to missing edge indicators and quasi-missing edge indicators.
    \end{enumerate}
    then the vertices in $Dormant(P)$ are indistinguishable from each other, so permuting these vertices will not change the ribbons of shapes $\gam_j, \dots, \gam_1, \tau, {\gam_1'}^T, \dots, {\gam_j'}^T$ we end up with. Thus, there are at least $|Dormant(P)|!$ elements in $Aut(\tau_P)$ which will result in the same ribbons which implies that 
    $N_{P}(\tau_P) \leq \frac{|Aut(\tau_P)|}{|Dormant(P)|!}$
    
    We now observe that for each vertex $v \in U_{\tau_P} \cap V_{\tau_P}$, at least one of the following $3$ cases must hold:
    \begin{enumerate}
        \item $v$ is involved in an intersection.
        \item $v$ is in the intersection of $U$ and $V$ for all of the shapes 
        \item $v$ is dormant and is in the intersection of $U$ and $V$ for all of the shapes 
    \end{enumerate}
    This implies that 
    \begin{align*}
        |Dormant(P)| &\geq |U_{\tau_P} \cap V_{\tau_P}| - k -     \left(\sum_{i=1}^{j}{V(\gam_i) \setminus (U_{\gam_i} \cup V_{\gam_i})}\right)\\
        &-|V(\tau) \setminus (U_{\tau} \cup V_{\tau})| - \left(\sum_{i=1}^{j}{|V(\gam'_i) \setminus (U_{\gam'_i} \cup V_{\gam'_i})|}\right)
    \end{align*}
    Since 
    \[
    \left(\sum_{i=1}^{j}{V(\gam_i) \setminus (U_{\gam_i} \cup V_{\gam_i})}\right) + |V(\tau) \setminus (U_{\tau} \cup V_{\tau})| + \left(\sum_{i=1}^{j}{|V(\gam'_i) \setminus (U_{\gam'_i} \cup V_{\gam'_i})|}\right) \leq k + |V(\tau_P)| - \frac{|U_{\tau_P}| + |V_{\tau_P}|}{2}
    \]
    we have that $\frac{1}{|Dormant(P)|!} \leq \frac{D_{SoS}^{2k + |V(\tau_P)| - \frac{|U_{\tau_P}| + |V_{\tau_P}|}{2}}}{|U_{\tau_P} \cap V_{\tau_P}|!}$ and the result follows.
\end{proof}

\subsection{Counting shapes with the tail bound function $c(\al)$}
\label{subsection: count-everything}
We have collected the unsymmetrized ribbons into symmetrized shapes, and we
must upper bound the number of possible shapes $\sigma$, $\gam_i$ and $\tau$
that are summed over.
We will do this by introducing a ``tail bound function`` $c(\alpha)$, such that $\sum_{\text{nontrivial permissible shapes } \al} 1/c(\al) \ll 1.$
The function will be of the form $c(\al) = f(\al)^{|V(\al) \setminus (U_\al \cap V_\al)|} g(|E(\al)|)$.
That is, we need a vertex decay of $f(\al)$ per vertex not in $U_\al \cap V_\al$, as well as a decay from conditioning
when the shape $\al$ has too many edges.

\begin{lemma}\label{shapecountboundone}
For all $v \in \mathbb{N}$ and all $e,k \in \mathbb{N} \cup \{0\}$, there are at most
\[
2\dsos^{k+1}{2^v}v^{2e}
\]
permissible shapes $\alpha$ with degree at most $\dsos$ such that $|V(\alpha) \setminus (U_{\alpha} \cap V_{\alpha})| = v$, there are $e$ edges in $E(\alpha)$ which are not incident to a vertex in $U_{\alpha} \cap V_{\alpha}$,
and $\alpha \setminus (U_{\alpha} \cap V_{\alpha})$ has $k$ connected components which are disconnected from $U_{\alpha} \cup V_{\alpha}$.
\end{lemma}
\begin{proof}
We can choose such a shape $\alpha$ as follows:
\begin{enumerate}
\item First choose the number of vertices in $U_{\alpha} \cap V_{\alpha}$. There are at most $\dsos+1$ possibilities for this as $0 \leq |U_{\alpha} \cap V_{\alpha}| \leq \dsos$.
\item Indicate whether $U_{\alpha} \setminus V_{\alpha}$ and $V_{\alpha} \setminus U_{\alpha}$ are non-empty. There are $4$ choices for this.
\item Order the vertices in $V(\alpha) \setminus (U_{\alpha} \cap V_{\alpha})$ so that the vertices in $U_{\alpha} \setminus V_{\alpha}$ come first (if there are any), followed by the vertices in $V_{\alpha} \setminus U_{\alpha}$ (if there are any), followed by the vertices in each of the $k$ connected components of $\alpha \setminus (U_{\alpha} \cap V_{\alpha})$ which are disconnected from $U_{\alpha} \cup V_{\alpha}$ (where we put the vertex which has an edge to $U_{\alpha} \cap V_{\alpha}$ first), followed by any remaining vertices.
\item For each vertex except for the first vertex, indicate whether it is in the same group as the previous vertex or is the first vertex of the next group. There are at most $2^{v-1}$ choices for this.
\item For the first vertex $w_i$ in each connected component $W_i$ of $\alpha \setminus (U_{\alpha} \cap V_{\alpha})$ which is disconnected from $U_{\alpha} \cup V_{\alpha}$, choose which vertex in $U_{\alpha} \cap V_{\alpha}$ $w_i$ is adjacent to. There are at most $\dsos^{k}$ choices for this.
\item For each edge which is not incident to a vertex in $U_{\alpha} \cap V_{\alpha}$, choose its two endpoints. There are at most $v^{2e}$ choices for this.
\item The missing edge indicators and quasi-missing edge indicators are fixed because $\alpha$ is a permissible shape.
\end{enumerate}
\end{proof}
This bound works well when $v \leq C\cdot \dsos$. For larger shapes, we will need a different bound.
\begin{lemma}
For all $v \in \mathbb{N}$ and all $e,k,C \in \mathbb{N} \cup \{0\}$, there are at most 
\[
2(\dsos+1)^{k+1}{8^v}{4^e}((C+4){\dsos})^{(C+2)\dsos}v^{\max{\{0,e + k - v - C\dsos\}}}
\]
permissible shapes $\alpha$ with degree at most $\dsos$ such that $|V(\alpha) \setminus (U_{\alpha} \cap V_{\alpha})| = v$, there are $e$ edges in $E(\alpha)$ which are not incident to a vertex in $U_{\alpha} \cap V_{\alpha}$,
and $\alpha \setminus (U_{\alpha} \cap V_{\alpha})$ has $k$ connected components which are disconnected from $U_{\alpha} \cup V_{\alpha}$.
\end{lemma}
\begin{proof}
This can be proved in the same way as Lemma \ref{shapecountboundone} except that we handle the edges which are not incident to a vertex in $U_{\alpha} \cap V_{\alpha}$ in a different way. We order these based on a breadth first search where we start with the vertices in $U_{\alpha} \cup V_{\alpha}$ in the queue.
\begin{enumerate}
\item For each new vertex which we reach, indicate whether it will be reached again by one of the first $C\dsos$ edges which do not go to a new vertex. Also indicate whether it has edges incident to it which lead to a new vertex (in which case it should be added to the queue) or not.
\item For each edge, first indicate whether it starts from the same vertex as the previous edge or starts from the next vertex in the queue. Then indicate whether this edge reaches a new vertex or goes to a vertex which has already been reached. If it is one of the first $(C+2)\dsos$ edges which goes to a vertex which has already been reached, indicate which vertex this is. Note that there are at most $(C+4)\dsos$ choices for this as this vertex is either in $U_{\alpha} \cup V_{\alpha}$ or is one of the indicated destinations for these edges (of which there are at most $(C+2)\dsos$). If this is a later edge which goes to a vertex which has already been reached, specify this vertex (there are at most $v$ choices for this).
\end{enumerate}
We now observe that at least $v - k - 2\dsos$ edges which are not incident to a vertex in $U_{\alpha} \cap V_{\alpha}$ are needed to connect $U_{\alpha} \cup V_{\alpha}$ to the rest of $\alpha$, so there are at most $e + k - v + 2\dsos$ edges which do not go to a new vertex. After we take the first $(C+2)\dsos$ such edges, there are at most $e + k - v - C\dsos$ such edges left.
\end{proof}
\begin{definition}\label{def:c-function}
Let $C'$ be an upper bound on the sparsity of shapes in \cref{lem:full-shape-decomposition}.
Given a permissible shape $\alpha$ such that $|V(\alpha) \setminus (U_{\alpha} \cap V_{\alpha})| = v \geq 1$, there are $e$ edges in $E(\alpha)$ which are not incident to a vertex in $U_{\alpha} \cap V_{\alpha}$, and $\alpha \setminus (U_{\alpha} \cap V_{\alpha})$ has $k$ connected components which are disconnected from $U_{\alpha} \cup V_{\alpha}$, we define $c(\alpha)$ by:
\begin{enumerate}
\item If $v \leq C\dsos$ and $e \leq C'v$ then 
\[
c(\alpha) = 40(2C\dsos^{4C'+2})^{|V(\alpha)| - \frac{|U_{\alpha}| + |V_{\alpha}|}{2}}
\]
\item If $v \leq C\dsos$ and $e > C'v$ then 
\[
c(\alpha) = 40(2C\dsos^{4C'+2})^{|V(\alpha)| - \frac{|U_{\alpha}| + |V_{\alpha}|}{2}}(2{C^2}\dsos^2)^{(e - C'v)}
\]
\item If $v > C\dsos$ and $e \leq v - k + C\dsos$ then 
\[
c(\alpha) = 40(16C\dsos^{4})^{|V(\alpha)| - \frac{|U_{\alpha}| + |V_{\alpha}|}{2}}
\]
\item If $v > C\dsos$ and $e > v - k + C\dsos$ then 
\[
c(\alpha) = 40(16C\dsos^{4})^{|V(\alpha)| - \frac{|U_{\alpha}| + |V_{\alpha}|}{2}}(8D_V)^{e + k - v - C\dsos}
\]
\end{enumerate}
If $\alpha$ is a trivial shape then we define $c(\alpha) = 1$.
\end{definition}
\begin{corollary}\label{cor:sum-c-alpha}
\[\sum_{\substack{\alpha \in \calS: \\\text{non-trivial,}\\ 
\text{permissible}}}{\frac{1}{c(\alpha)}} \leq \frac{1}{10}.\]
\end{corollary}

\subsection{Statement of main lemmas}

Now that we have wrangled the moment matrix into the correct form, we will show that all the error terms are small.
Recall the expression for the moment matrix from \cref{lem:full-shape-decomposition}:
    \begin{align*}
    &\Pi^{1/2}\left(\sum_{\substack{\text{sparse}\\
    \text{permissible}\\ \sigma \in \calL}}{
    {\lam_\sigma'}\frac{M_\sigma}{|\Aut(\sigma)|}}\right)\cdot\\
    &\left(\sum_{j=0}^{2\dsos}{(-1)^{j}\sum_{\substack{\text{sparse}\\
    \text{permissible}\\ \gam_j,\ldots,\gam_1,\tau, \gam'_1,\ldots,\gam'_j }}
    \sum_{\substack{\text{nonequivalent} \\P \in \calP^{mid}_{\gam_j, \dots, \gam_1, \tau, \gam'_1, \dots, \gam'_j}}} 
    N_{P}(\tau_P)\lam'_{\gam_j \circ \cdots \circ \gam_1 \circ \tau \circ \gam_1'^\T \circ \cdots \circ \gam_j'^\T}
    \frac{M_{\tau_P}}{|\Aut(\tau_P)|}
    }\right)\cdot\\
    &\left(\sum_{\substack{\text{sparse}\\
    \text{permissible}\\ \sigma \in \calL}}{
    {\lam_\sigma'}\frac{M_\sigma}{|\Aut(\sigma)|}}\right)^\T \Pi^{1/2}\\
    &+ \Pi^{1/2}\cdot {\text{truncation error}}\cdot \Pi^{1/2}
\end{align*}

The following lemmas, for which the intuition was given in \cref{sec:informal}, are sufficient to prove \cref{thm:main}.

\begin{restatable}{lemma}{middleShapes}(Non-trivial Middle Shapes)
    \label{lem:nontrivial-middle-shapes}
    For all sparse permissible $\tau \in \calM$ such that ${|V(\tau)| > \frac{|U_\tau| + |V_\tau|}{2}}$ and $|E_{mid}(\tau)| - |V(\tau)| \leq C\dsos$,
    \[\lam'_\tau\frac{\norm{M_\tau}}{|\Aut(\tau)|} \leq \frac{1}{c(\tau)}.\]
\end{restatable}
    
\begin{restatable}{lemma}{intersectionTerms}(Intersection Terms)
    \label{lem:intersection-terms}
    For all $j \geq 1$ and sparse permissible $\gam_j, \dots, \tau, \dots, \gam_j'$ such that for each shape $|E_{mid}(\al)| - |V(\al)| \leq C\dsos$,
    \[ \sum_{\substack{\text{nonequivalent} \\ P \in \calP^{mid}_{\gam_j,\dots,\gam_j'}}} N_{P}(\tau_P) \lam'_{\gam_j \circ \cdots \circ \gam_j'^\T} \frac{\norm{M_{\tau_P}}}{|\Aut(\tau_P)|} \leq \frac{1}{c(\tau) \prod_{i=1}^j c(\gam_i)c(\gam_i')}.\]
\end{restatable}
\begin{restatable}{lemma}{truncationerror}(Truncation Error)
    \label{lem:truncation-error}
\[\text{truncation error} \psdleq n^{-\Omega(C\dsos)} \pi.\]
\end{restatable}

\begin{restatable}{lemma}{leftrightconditioning}(Sum of left shapes is well-conditioned)
    \label{lem:left-right-sep-conditioned}
    \[\left(\sum_{\substack{\text{sparse,}\\\text{permissible} \\ \sigma \in \calL}}\lam'_\sigma \frac{M_\sigma}{|\Aut(\sigma)|}\right)\left(\sum_{\substack{\text{sparse,}\\\text{permissible} \\ \sigma \in \calL}}\lam'_\sigma \frac{M_\sigma}{|\Aut(\sigma)|}\right)^\T
    \psdgeq n^{-O(\dsos)} \pi\]
\end{restatable}

\begin{proof}[Proof of~\cref{thm:main} assuming~\cref{lem:nontrivial-middle-shapes},~\cref{lem:intersection-terms},~\cref{lem:truncation-error}, \cref{lem:left-right-sep-conditioned}]
    For ease of exposition we omit the automorphism groups.
    In the approximate PSD decomposition, the term $j = 0$ can be broken up into
    the leading term $\pi$ and remaining nontrivial middle shapes,
    \begin{align*}
    \sum_{\substack{\text{sparse,}\\\text{permissible}\\\tau \in \calM}} \lam'_\tau M_\tau
    &= \sum_{\substack{\text{permissible}\\\tau \in \calM:\\U_\tau = V_\tau = V(\tau)}} \lam_\tau M_\tau && +  \sum_{\substack{\text{sparse,}\\\text{permissible}\\\tau \in \calM:\\|V(\tau)| > \frac{|U_\tau| + |V_\tau|}{2}}} \lam'_\tau M_\tau\\
    &= \pi && +  \sum_{\substack{\text{sparse,}\\\text{permissible}\\\tau \in \calM:\\|V(\tau)| > \frac{|U_\tau| + |V_\tau|}{2}}} \lam'_\tau M_\tau.
    \end{align*}
    Using \cref{lem:permissible-has-pi}, the middle shapes ($j= 0$) and intersection terms ($j \geq 1$) are
    \begin{align*}
    \pi\Bigg(\Id &+ \sum_{\substack{\text{sparse,}\\\text{permissible}\\\tau \in \calM:\\|V(\tau)| > \frac{|U_\tau| + |V_\tau|}{2}}} \lam'_\tau M_\tau\\
    & + \sum_{j=1}^{2\dsos}(-1)^j \sum_{\substack{\text{sparse},\\ \text{permissible}\\ \gam_j, \dots, \gam_j'}} \sum_{\substack{\text{nonequivalent}: \\ P \in \calP^{mid}_{\gam_j, \dots, \gam_j'}}}
    N_{P}(\tau_P)\lam'_{\gam_j \circ \cdots \circ \gam_j'^\T} M_{\tau_P}\Bigg).
    \end{align*}
    By \cref{lem:nontrivial-middle-shapes} and \cref{lem:intersection-terms},
    (summed with \cref{cor:sum-c-alpha})
    this is at least $\Omega(1)\pi$.
    
    Now plugging this into the PSD decomposition, we have:
    \begin{align*}
    \sum_{\alpha \in \calS}\lam_\al \frac{M_\alpha}{\abs{\Aut(\al)}} &\psdgeq \Pi^{1/2}\left(\sum_{\substack{\text{sparse}\\\text{permissible}\\\sigma \in \calL}} \lam'_\sigma M_\sigma\right)
    \Omega(1)\pi\left(\sum_{\substack{\text{sparse}\\\text{permissible}\\\sigma \in \calL}}\lam'_\sigma M_\sigma\right)^\T\Pi^{1/2} + \text{truncation error}
    \end{align*}
By \cref{lem:permissible-has-pi} again,
    \begin{align*}
    &= \Omega(1)\Pi^{1/2}\left(\sum_{\substack{\text{sparse}\\\text{permissible}\\\sigma \in \calL}} \lam'_\sigma M_\sigma\right)\left(\sum_{\substack{\text{sparse}\\\text{permissible}\\\sigma \in \calL}}\lam'_\sigma M_\sigma\right)^\T\Pi^{1/2} + \text{truncation error} 
    \end{align*}
By    \cref{lem:left-right-sep-conditioned}, \cref{lem:truncation-error}, and taking $C$ sufficiently large,
    \begin{align*}
    & \psdgeq n^{-O(\dsos)}\Pi  - n^{-\Omega(C\dsos)}\Pi \\
    & \psdgeq 0 
    \end{align*}
\end{proof}

\subsection{Conditioning II: Frobenius norm trick}

Shapes with a large number of excess edges should be handled using their Frobenius norm.
The Frobenius norm trick improves on the first moment method, \cref{lem:density-bound},
so it could be used to also handle $O(1)$-sparse subgraphs in this work. On the other hand,
the forbidden subgraph method can handle forbidden graphs beyond the first moment method, which
may lead to improvements for smaller $d$.

There are two parts to the Frobenius norm trick. First,
we compute the Frobenius norm of a graph matrix.

\begin{lemma}\label{lem:frobenius-trick-first-part}
   For all proper shapes $\al$, 
    \[\E[\tr(M_\al M_\al^\T)] \leq |\Aut(\al)| n^{|V(\al)| + |I_\al|} 
    .\]
\end{lemma}
\begin{proof}
Any term contributing to $\E \left[\tr(M_\al M_\al^\T)\right]$ is a labeling of $\al$ and $\al^\T$ s.t. every edge appears exactly twice. After choosing the labels for $\al$ in $n^{|V(\al)|}$ ways, the labeling of $\al^\T$ gives an isomorphism between $\al$ and $\al^\T$ (except for isolated vertices, which give an extra factor of $n$ each).
\end{proof}

For proper shapes $\al$, this gives a norm bound independent of the number of edges! MVS with a large number of induced edges can't hurt us.
We have $\lam_\al \norm{M_\al} \leq \left(\frac{k}{n}\right)^{|V(\al)| - \frac{|U_\al| + |V_\al|}{2}}\sqrt{p}^{|E(\al)|}\sqrt{n}^{|V(\al)|}$. 
The number of edges could be smaller than the number of vertices, but only by $\dsos$ (since everything must be connected to $U_\al \cup V_\al$). If there are at least $C\dsos$ excess edges, then the norm is small enough to sum.



The Frobenius norm is small for too-dense subgraphs of $G$. The second part of the Frobenius norm trick allows us to use these bounds with high probability instead of in expectation, which we can do using Markov's inequality once instead of using Markov's inequality on each one.
\begin{lemma}\label{lem:frobenius-trick-second-part}
For any random matrices $M_1,\ldots,M_k$,
\[
\E\left[\tr\left(\left(\sum_{i=1}^{k}{M_i}\right)\left(\sum_{i=1}^{k}{M_i}\right)^\T\right)\right] \leq \left(\sum_{i=1}^{k}{\sqrt{\E\left[\tr\left({M_i}{M_i}^\T\right)\right]}}\right)^2
\]
\end{lemma}
\begin{proof}
Observe that by Cauchy-Schwarz,
\begin{align*}
\E\left[\tr(M_i{M_j}^\T)\right] &= \E\left[\sum_{a,b}{(M_i)_{ab}(M_j)_{ab}}\right] \\
&\leq \sqrt{\E\left[\sum_{a,b}{(M_i)_{ab}^2}\right]}\sqrt{\E\left[\sum_{a,b}{(M_j)_{ab}^2}\right]} \\
&= \sqrt{\E\left[\tr\left({M_i}{M_i}^\T\right)\right]}\sqrt{\E\left[\tr\left({M_j}{M_j}^\T\right)\right]}
\end{align*}
Applying this inequality for all $i,j \in [k]$ gives the result.
\end{proof}

\subsection{Norm bounds}\label{sec:probabilistic-norm-bounds}

For the trace method, it is easier to use graph matrices with left/right sides
indexed by \emph{ordered} tuples.

\begin{theorem}
If $D_V \geq \lceil{D_{SoS}\ln(n)}\rceil$ then for all proper shapes $\alpha$ such that $|V(\alpha)| \leq 3D_V$ and all $\epsilon' > 0$, taking $M_{\alpha}$ to be the graph matrix where the rows and columns are indexed by ordered tuples rather than sets,
\[
\Pr\left(\norm{M_{\alpha}} > 10{\left(\frac{1}{\epsilon'}\right)}^{\frac{1}{2D_V}}\max_{\text{separator } S}{\left\{n^{\frac{|V(\alpha)| - |S|}{2}}\left(12D_V\right)^{|V(\alpha)| - \frac{|U_{\alpha}| + |V_{\alpha}|}{2} + |S| - \frac{|L_S| + |R_S|}{2}}\left(3\sqrt{\frac{1-p}{p}}\right)^{|E(S)|}\right\}}\right) < \epsilon'
\]
\end{theorem}
\begin{proof}
By \cref{thm:tracepowercalculations},
\begin{align*}
&\E\left[\tr\left(\left(M_{\alpha}M_{\alpha}^\T\right)^q\right)\right] \leq 4^{2q\left(|V(\alpha)| - \frac{|U_{\alpha}| + |V(\alpha)|}{2}\right)}\left(\sqrt{|V(\alpha)|}\right)^{2q|V(\alpha) \setminus (U_{\alpha} \cup V_{\alpha})|}n^{q|V(\alpha)|} \cdot \\
&\left(\max_{\text{separator } S}{n^{-\frac{|S|}{2}}\left(3\sqrt{\frac{1-p}{p}}\right)^{|E(S)|}\left(\sqrt{|V(\alpha)|}\right)^{|S| - |U_{\alpha} \cap V_{\alpha}|}(2q)^{|S|-\frac{|L_S| + |R_S|}{2} + \frac{c(\alpha,S)}{2}}}\right)^{2q-2}
\end{align*}
Since $|V(\alpha)|$ and $q$ will both be $O(D_V) = O(D_{SOS}log(n))$, for convenience we combine the factors of $|V(\alpha)|$ and $q$ using the following proposition:
\begin{proposition}
\[
\frac{1}{2}|V(\alpha) \setminus (U_{\alpha} \cup V_{\alpha})| + \frac{1}{2}(|S| - |U_{\alpha} \cap V_{\alpha}|) + \frac{c(\alpha,S)}{2} \leq |V(\alpha)| - \frac{|U_{\alpha}| + |V_{\alpha}|}{2}
\]
\end{proposition}
\begin{proof}
Observe that since $c(\alpha,S)$ counts the number of connected components of $\alpha \setminus S$ which are disconnected from $U_{\alpha} \cup V_{\alpha}$, we have that $|S| - |U_{\alpha} \cap V_{\alpha}| + c(\alpha,S) \leq |V(\alpha)| - |U_{\alpha} \cap V_{\alpha}|$. We now observe that $\frac{|V(\alpha) \setminus (U_{\alpha} \cup V_{\alpha})| + |V(\alpha)| - |U_{\alpha} \cap V_{\alpha}|}{2} = |V(\alpha)| - \frac{|U_{\alpha}| + |V_{\alpha}|}{2}$
\end{proof}
Since $|V(\alpha)| \leq 3D_V$, for all $q \leq \frac{3}{2}D_V$ we have that 
\begin{align*}
&\E\left[\tr\left(\left(M_{\alpha}M_{\alpha}^\T\right)^q\right)\right] \leq \\
&\max_{\text{separator } S}{\left\{n^{|S|}\left(n^{\frac{|V(\alpha)| - |S|}{2}}\left(12D_V\right)^{|V(\alpha)| - \frac{|U_{\alpha}| + |V_{\alpha}|}{2} + |S| - \frac{|L_S| + |R_S|}{2}}\left(3\sqrt{\frac{1-p}{p}}\right)^{|E(S)|}\right)^{2q} \right\}}
\end{align*}
Using the fact that for all $q \in \mathbb{N}$ and all $\epsilon' > 0$, 
\[
\Pr\left(\norm{M_{\alpha}} > \sqrt[2q]{\frac{\E\left[\tr\left(\left(M_{\alpha}M_{\alpha}^\T\right)^q\right)\right]}{\epsilon'}}\right) < \epsilon'
\]
and taking $q = D_V$, we have that 
\[
\Pr\left(\norm{M_{\alpha}} > 10{\left(\frac{1}{\epsilon'}\right)}^{\frac{1}{2D_V}}\max_{\text{separator } S}{\left\{n^{\frac{|V(\alpha)| - |S|}{2}}\left(12D_V\right)^{|V(\alpha)| - \frac{|U_{\alpha}| + |V_{\alpha}|}{2} + |S| - \frac{|L_S| + |R_S|}{2}}\left(3\sqrt{\frac{1-p}{p}}\right)^{|E(S)|}\right\}}\right) < \epsilon'
\]
\end{proof}
\begin{remark}
We will take $\epsilon' = \frac{\epsilon}{c(\alpha)}$where $\epsilon$ is the overall probability that our SoS bound fails. However, we did not want to hard code this choice into the theorem statement.
\end{remark}

For improper shapes, we linearize them. The coefficients in the linearization are: 
\begin{proposition}\label{prop:linearization}
For a shape $\al$, 
\[M_\al = \displaystyle\sum_{\text{linearizations }\beta\text{ of }\al} \left(\sqrt{\frac{1-p}{p}} \right)^{|E(\al)| - |E(\beta)| - 2|E_{phantom}(\beta)|}M_\beta. \]
\begin{proof}
Each Fourier character can be linearized as $\chi_e^k = \E[\chi_e^k] + \E[\chi_e^{k+1}]\chi_e$.
Then use \cref{prop:powerexpectedvalue}.
\end{proof}
\end{proposition}
\begin{corollary}\label{cor:epsilon-norm-bound}
If $D_V \geq \lceil{\dsos {ln(n)}}\rceil$ then for any shape $\alpha$ (including improper shapes, shapes with missing edge indicators, and shapes with quasi-missing edge indicators), taking $M_{\alpha}$ to be the graph matrix where the rows and columns are indexed by sets (i.e. the definition we use throughout the paper),
\begin{align*}
&\Pr\Bigg(\norm{M_{\alpha}} > 20{\left(\frac{2^{|E(\alpha)| - |E_{\text{no repetitions}}(\alpha)|}}{\epsilon'}\right)}^{\frac{1}{2D_V}}2^{|E(\alpha)| - |E_{\text{no repetitions}}(\alpha)|}|\sqrt{|U_{\alpha} |!|V_{\alpha} |!} \\
&\max_{\beta,S}{\left\{n^{\frac{|V(\alpha)| - |S| + |I_{\beta}|}{2}}\left(12D_V\right)^{|V(\alpha)| - \frac{|U_{\alpha}| + |V_{\alpha}|}{2} + |S| - \frac{|L_S| + |R_S|}{2}}\left(\sqrt{\frac{1-p}{p}}\right)^{|E(\alpha)| - |E(\beta)| - 2|E_{phantom}(\beta)|}\left(3\sqrt{\frac{1-p}{p}}\right)^{|E(S)|}\right\}}\Bigg) < \epsilon'
\end{align*}
where $E_{\text{no repetitions}}(\alpha)$ is the set (rather than the multi-set) of edges of $\alpha$, $\beta$ is a linearization of $\alpha$, $S$ is a separator of $\beta$, $I_{\beta}$ is the set of vertices of $\beta$ which are isolated, and $E_{phantom}(\beta)$ is the set of edges of $\alpha$ which are not in $\beta$. As a special case, when $\alpha$ is a proper shape with no isolated vertices (which may still contain missing edge indicators and quasi-missing edge indicators),
\begin{align*}
&\Pr\Bigg(\norm{M_{\alpha}} > 20{\left(\frac{1}{\epsilon'}\right)}^{\frac{1}{2D_V}}\sqrt{|U_\al|!|V_\al|!}\max_{S}{\left\{n^{\frac{|V(\alpha)| - |S|}{2}}\left(12D_V\right)^{|V(\alpha)| - \frac{|U_{\alpha}| + |V_{\alpha}|}{2} + |S| - \frac{|L_S| + |R_S|}{2}}\left(3\sqrt{\frac{1-p}{p}}\right)^{|E(S)|}\right\}}\Bigg) < \epsilon'
\end{align*}
\end{corollary}
\begin{proof}
We make the following observations:
\begin{enumerate}
    \item There are at most $2^{|E(\alpha)| - |E_{\text{no repetitions}}(\alpha)|}$ linearizations $\beta$ of $\alpha$, so to ensure that the bound on $M_{\alpha}$ fails with probability at most $\epsilon'$, it is sufficient to ensure that the bound for each $M_{\beta}$ fails with probability at most $\frac{\epsilon'}{2^{|E(\alpha)| - |E_{\text{no repetitions}}(\alpha)|}}$.
    \item The coefficient for each linearization $\beta$ of $\alpha$ is at most $\left(\sqrt{\frac{1-p}{p}}\right)^{|E(\alpha)| - |E(\beta)| - 2|E_{phantom}(\beta)|}$ per \cref{prop:linearization}.
    \item To handle quasi-missing edge indicators, we observe that each quasi-missing edge is a convex combination of missing edge indicators, so we can express $M_{\beta}$ as a convex combination of $M_{\beta'}$ where $\beta'$ only has missing edge indicators rather than quasi-missing edge indicators. We now use Lemma 5.14 of \cite{AMP20}, which says that when we use the trace power method on a matrix $M$ which is a convex combination of matrices $M'$, we can obtain a bound on $\norm{M}$ which is at most twice the maximum over $M'$ of the bound we obtain on $\norm{M'}$.
    \item Notice in the trace power method, the length-$q$ walk on the matrix indexed by sets follows the same block structure that each $\al, \al^T$ is specified by an embedding of vertices in $V(\al)$ to $[n]$. The only change is vertices in the boundary only agree as a set instead of ordered tuples, to address this, it suffices to additionally identify a permutation between the boundaries, and this requires $|U _\al|!$ when we go from $V_{\al^T_t}$ to $U_{\al_{t+1}}$, and it requires $|V_\al|!$ when go from $V_{\al_t}$ to $U_{\al_{t}^T}$. Taking the square root gives us the desired bound.
\end{enumerate}
\end{proof}


One more proposition is needed to handle the automorphism terms that arise
when switching between ribbons/shapes/shapes indexed by ordered tuples,
   \begin{proposition}\label{prop:midshape-aut}
    For a permissible shape $\al$ of degree at most $\dsos$,
    \[ 
   \frac{\sqrt{|U_\al|!|V_\al|!}}{|\Aut(\al)|}\leq \dsos^{|V(\al) \setminus (U_\al \cap V_\al)|/2 + |V(\al) \setminus (U_\al \cup V_\al)|}
    \]
    \end{proposition}
\begin{proof}
    Write 
    \[ 
   \frac{\sqrt{|U_\al|!|V_\al|!}}{|\Aut(\al)|}
    = \frac{\sqrt{|U_\al|!|V_\al|!}}{|U_\al \cap V_\al|!} \cdot \frac{|U_\al \cap V_\al|!}{|\Aut(\al)|}
    \]
    The first ratio is upper bounded by $\dsos^{|V(\al) \setminus (U_\al \cap V_\al)|/2}$.
    
    For the second ratio, in a permissible shape, $|\Aut(\al)| \geq (|U_\al \cap V_\al| - |V(\al) \setminus (U_\al \cup V_\al)|)!$.
    This is because in a permissible shape, the number of edges incident to $U_\al \cap V_\al$ is at most the
    number of connected components of $V(\al) \setminus (U_\al \cup V_\al)$, which is
    upper bounded by $|V(\al) \setminus (U_\al \cup V_\al)|$.
    The remaining vertices of $U \cap V$ that are not incident to an edge are completely interchangeable.
    Therefore this ratio is upper bounded by $\dsos^{|V(\al) \setminus (U_\al \cup V_\al)|}$.
\end{proof}

\subsection{Proof of main lemmas}

    \subsubsection{\cref{lem:nontrivial-middle-shapes}: nontrivial middle shapes}

\middleShapes*

\begin{proof}
    The remaining ``core'' shapes are $\tau \in \calM$ nontrivial, sparse, permissible
    with ${|E(\al)| \leq |V(\al)| + C\dsos}$.
    For these shapes, the norm bound in \cref{cor:epsilon-norm-bound}
    with $\eps' = \frac{1}{nc(\tau)}$ holds whp,
    via \cref{cor:sum-c-alpha}.
    \begin{align*}
        \lam'_\tau  \frac{\norm{M_\tau}}{|\Aut(\tau) |} \\
        \leq& m(\tau)\left(\lam_\tau + \sum_{U \supset \tau} \lam_U c(\tau, U)\right) \cdot\\
        & C \left(2^{|E(\tau)|}n c(\tau)\right)^{\frac{1}{2D_V}} 2^{|E(\tau)|} \max_{\text{separator } S}{\left\{n^{\frac{|V(\alpha)| - |S|}{2}}\left(12D_V\right)^{|V(\alpha)| - \frac{|U_{\alpha}| + |V_{\alpha}|}{2} + |S| - \frac{|L_S| + |R_S|}{2}}\left(3\sqrt{\frac{1-p}{p}}\right)^{|E(S)|}\right\}}\\ &
        \cdot \frac{\sqrt{|U_\tau|!|V_\tau|!}}{|\Aut(\tau)|}
    \end{align*}
    We have:
    \begin{align*}
        m(\tau) & \leq \left(\frac{1}{(1-p)^{2\dsos}}\right)^{|V(\tau) \setminus (U_\tau \cap V_\tau)|} 
        \leq 2^{|V(\tau) \setminus (U_\tau \cap V_\tau)|}
        && (\text{\cref{lem:bound-m}})\\
        \sum_{U \supset \tau} \lam_U c(\tau, U) &= o(1)\lam_\tau && (\text{\cref{lem:conditioning-is-negligible}})\\
        2^{|E(\al)|/2D_V}& \leq C && (\tau\text{ is sparse})\\
        n^{1/2D_V} & \leq 2\\
        c(\tau) &\leq C\dsos^{c|V(\tau) \setminus (U_\tau \cap V_\tau)|} && (\text{\cref{def:c-function}, $\tau$ is sparse})\\
        c(\tau)^{\frac{1}{2D_V}} &\leq \dsos^c && (\text{From previous line})\\
        3^{|E(S)|} &\leq C^{|V(\tau) \setminus (U_\tau \cap V_\tau)|} && (\tau\text{ is sparse})\\
        \frac{\sqrt{|U_\tau|!|V_\tau|!}}{|\Aut(\tau)|}&\leq  C\dsos^{c|V(\tau) \setminus (U_\tau \cap V_\tau)|} && (\text{\cref{prop:midshape-aut}}))\\
    \end{align*}
    Going through the charging argument for middle shapes in \cref{prop:informal-middle-shapes},
    \[ \lam_\tau n^{\frac{|V(\al)| - |S|}{2}} \left(\sqrt{\frac{1-p}{p}}\right)^{|E(S)|}\leq \left(\frac{k\sqrt{d}}{n}\right)^{|V(\tau)| - \frac{|U_\tau| + |V_\tau|}{2}} \cdot \left(\frac{1}{\sqrt{d}}\right)^{|S| - \frac{|U_\tau| + |V_\tau|}{2}}.\]
    Using $\dsos \leq \frac{d^{1/2}}{\log n}$, we have $\sqrt{d} \geq D_V$ and the last term cancels $D_V^{|S| - \frac{|L_S| + |R_S|}{2}}$
    from the norm bound after the following claim:
    \begin{claim}
        For a middle shape $\tau$ and any vertex separator $S$, $|L_S| \geq |U_\tau|$ and $|R_S| \geq |V_\tau|$.
        \begin{proof}
        $L_S$ and $R_S$ are both vertex separators of $\tau$, and since
        $\tau$ is a middle shape, $U_\tau, V_\tau$ are MVSs of $\tau$.
        \end{proof}
    \end{claim}
Putting it together, the vertex decay of $\frac{k\sqrt{d}}{n}$ on $|V(\tau)| - \frac{|U_\tau| + |V_\tau|}{2}$ needs to be $CD_V\dsos^c$ to
    overcome the combinatorial factors, as stated in \cref{thm:main}.
    \[ \lam'_\tau \frac{\norm{M_\tau}}{|\Aut(\tau)|} \leq \frac{1}{c(\tau)}.\]
\end{proof}
\subsubsection{\cref{lem:intersection-terms}: bounding intersection terms } 
\label{sec:intersection-terms}

We will need the following proposition.
\begin{proposition}\label{prop:edge-bound-in-composition}
    For any $C'$-sparse permissible composable shapes $\alpha_1,\ldots,\alpha_k$ and any intersection pattern $P$ on $\alpha_1,\ldots,\alpha_k$, letting $\tau_P$ be the resulting shape,
    \[
    |E(\tau_P)| \leq (2C'+1)\sum_{i=1}^{k}{\left(|V(\alpha_i)| - \frac{|U_{\alpha_i}| + |V_{\alpha_i}|}{2}\right)}
    \]
\end{proposition}
\begin{proof}
Observe that since $\alpha_i$ is permissible, 
the only edges incident to $U_{\alpha_i} \cap V_{\alpha_i}$
are one for each connected component of $V(\alpha_i) \setminus (U_{\al_i} \cup V_{\al_i})$.
Therefore, there are at most $|V(\al_i) \setminus (U_{\al_i} \cup V_{\al_i})|$ of these edges.
Since $\alpha_i$ is $C'$-sparse,
\[
|E(\alpha_i)| \leq |V(\al_i) \setminus (U_{\al_i} \cup V_{\al_i})| + C'|V(\alpha) \setminus (U_{\alpha_i} \cap V_{\alpha_i})| \leq (2C'+1)\left(|V(\alpha_i)| - \frac{|U_{\alpha_i}| + |V_{\alpha_i}|}{2}\right)
\]
Since $|E(\tau_P)| = \sum_{i=1}^{k}{|E(\alpha_i)|}$, summing this equation over all $i \in [k]$ gives the result.
\end{proof}

\intersectionTerms*

\begin{proof}
We index the intersecting shapes as $\al_i = \gam_k, \dots, \gam_1, \tau, \gam_1', \dots, \gam_k'$. For each $\al_i$,
by \cref{lem:bound-m} and \cref{lem:conditioning-is-negligible},
\[\lam'_{\al_i} \leq 2^{|V(\al_i) \setminus (U_{\al_i} \cap V_{\al_i})|}\lam_{\al_i}. \]
Let $p_\el(\al_i)$ be the number of nonequivalent intersection patterns $P \in \calP^{mid}_{\gam_j, \dots, \gam_j'}$ which have exactly $\el$ intersections.

Apply the norm bound in \cref{cor:epsilon-norm-bound} with $\eps' = \frac{1}{n^C \prod_{i=1}^{2k+1} c(\al_i) p_k(\alpha_i)}$, 
which holds whp via \cref{cor:sum-c-alpha}.

\begin{align*}
    &\sum_{\text{non equiv. P}} N_p(\tau_P) \cdot\lam_{\gam_k \circ \cdots \circ \gam_1 \circ \tau \circ \gam_1^\T \circ \cdots \circ \gam_k'^\T} \frac{\norm{M_{\tau_P}}}{|\Aut(\tau_P)|} \\
    \leq &\lam_{\gam_k \circ \cdots \circ \gam_1 \circ \tau \circ \gam_1^\T \circ \cdots \circ \gam_k'^\T} \cdot 20{\left(2^{|E(\tau_P)|}n^C\prod_{i=1}^{2k+1}c(\al_i) p_\el(\al_i)\right)}^{\frac{1}{2D_V}}2^{|E(\tau_P)| } \cdot \frac{\sqrt{|U_{\tau_P}|! |V_{\tau_P}|!}}{|\Aut(\tau_P) |} 
    \cdot \\
    & \max_{\beta,S}{\left\{n^{\frac{|V(\tau_P)| - |S| + |I_{\beta}|}{2}}\left(12D_V\right)^{|V(\tau_P)| - \frac{|U_{\tau_P}| + |V_{\tau_P}|}{2} + |S| - \frac{|L_S| + |R_S|}{2}}\left(\sqrt{\frac{1-p}{p}}\right)^{|E(\tau_P)| - |E(\beta)| - 2|E_{phantom}(\beta)|}\left(3\sqrt{\frac{1-p}{p}}\right)^{|E(S)|}\right\}.}
\end{align*}

Bounds:
\begin{align*}
    12^{|E(\tau_P)|} & \leq C^{\sum_{i = 1}^{2k+1} |V(\al_i)| - \frac{|U_{\al_i}| + |V_{\al_i}|}{2}} && (\text{\cref{prop:edge-bound-in-composition}})\\
    n^{C/2D_V} & \leq 2\\
    c(\al_i) & \leq C(\dsos^c)^{|V(\al_i) \setminus (U_{\al_i} \cap V_{\al_i})|} && (\text{\cref{def:c-function}, $\al_i$ is sparse})\\
    c(\al_i)^{1/2D_V} & \leq \dsos^c && (\text{From previous line})\\
    p_\el(\al_i) & \leq (D_{SoS}^c)^{|V(\tau_P)| - \frac{|U_{\tau_P}| + |V_{\tau_P}|}{2}}D_V^{\el} && (\text{\cref{lem:count-nonequivalent-intersection-patterns}})\\
    \sum_{\text{non equiv. P}} N_P(\tau_P) & \leq \frac{(\dsos^c)^{\sum_{i = 1}^{2k+1} |V(\al_i)| - \frac{|U_{\al_i}| + |V_{\al_i}|}{2}}}{|U_{\tau_P} \cap V_{\tau_P}|!} |\Aut(\tau_P)| && (\text{\cref{lem:bound-np}})\\
    \frac{\sqrt{|U_{\tau_P}|!|V_{\tau_P}|!}}{|\Aut(\tau_P) |} &\leq C(\dsos^c)^{|V(\tau_P) \setminus (U_{\tau_P} \cap V_{\tau_P})|} && (\text{\cref{prop:midshape-aut}})
\end{align*}

Following the charging argument for intersection terms in \cref{prop:informal-intersection-terms},
\begin{align*}
&\lam_{\gam_k \circ \cdots \circ \gam_1 \circ \tau \circ \gam_1^\T \circ \cdots \circ \gam_k'^\T} n^{\frac{|V(\tau_P)| - |S| + |I_\beta|}{2}}\left(\sqrt{\frac{1-p}{p}}\right)^{|E(S)| + |E(\al)| - |E(\beta)| - 2|E_{phantom}(\beta)|}\\
\leq &\left(\frac{k\sqrt{d}}{n}\right)^{\sum_{i=1}^{2k+1}|V(\al_i)| - \frac{|U_{\al_i}| + |V_{\al_i}|}{2}} \cdot \left(\frac{1}{\sqrt{d}}\right)^{i_P - |I_\beta| + |S| - \frac{|U_\beta| + |V_\beta|}{2}}.
\end{align*}
In \cref{prop:informal-intersection-terms} we used the intersection tradeoff lemma, \cref{lem:intersection-tradeoff-lemma}, to prove that the exponent of the second term
is nonnegative. In that proof, use the following claim to get a lower bound
of $|S| - \frac{|L_S| + |R_S|}{2}$ on the exponent.
\begin{claim}
    $\frac{|L_S| + |R_S|}{2} \geq |S_{\tau_P,min}|.$ 
\end{claim}
\begin{proof}
    Both $L_S$ and $R_S$ are vertex separators of $\beta$, hence they are
    larger than $S_{\tau_P, min}$, which is the smallest separator among all linearizations of $\tau_P$.
\end{proof}
Since $\sqrt{d} \geq D_V$, the second term cancels $D_V^{|S| - \frac{|L_S| + |R_S|}{2}}$
in the norm bound.

The remaining factors are mostly easily seen to be under control by the vertex decay.
One exception is the term $D_V^{\el}$ in $p_\el(\al_i)$, the count of nonequivalent intersection patterns.
By combining this with $D_V^{|V(\tau_P)| - \frac{|U_{\tau_P}| + |V_{\tau_P}|}{2}}$ from the norm bound we get at most $D_V^{\sum_{i=1}^{2k+1}|V(\al_i)| - \frac{|U_{\al_i}| + |U_{\al_i}|}{2}}$,
which is controlled by one factor of $D_V$ in the vertex decay.
In total, a vertex decay of $CD_V\dsos^c$ per vertex is sufficient.

Since each additional level of intersections is non-trivial,
the total number of vertex decay factors is at least $k$, which is
sufficient to handle $C\dsos^c$ per shape in addition to per vertex.
\end{proof}

\subsubsection{\cref{lem:truncation-error}: truncation error for shapes with too many vertices}

The goal of the next two subsections is to upper bound the truncation error.
\truncationerror*

In this section we upper bound the first part of the truncation error,
where shapes have many extra vertices.
The counting is similar to the previous section while we use simpler (and looser) bounds for matrix norms.
These will be handled by the fact that large shapes
provide lots of vertex decay factors.

We recall
\begin{align*}
        &{\text{truncation error}}_{\text{too many vertices}} =         -\sum_{\substack{\text{sparse, permissible}\\ \sigma,\tau,\sigma':\\|V(\sigma \circ \tau \circ \sigma'^\T)| > D_V}} \lam'_{\sigma \circ \tau \circ \sigma'^\T} \frac{M_{\sigma \circ \tau \circ \sigma'^\T}}{|\Aut(\sigma \circ \tau \circ \sigma'^\T)|}\\
        &+ \sum_{j=1}^{2\dsos}{(-1)^{j+1}\sum_{\substack{\text{sparse, permissible} \\ \sigma,\gam_j,\ldots,\gam_1,\tau,\gam'_1,\ldots,\gam'_j,\sigma': \\ 
        |V(\sigma \circ \gam_j \circ \ldots \circ \gam_1)| > D_V \text{ or} \\ |V({\gam'_1}^{\T} \circ \ldots \circ {\gam'_j}^{\T} \circ {\sigma'}^{\T})| > D_V}}
        \sum_{\substack{\text{nonequiv.}\\ P \in \calP^{mid}_{\gam_j, \dots, \gam_j'}}} N_{P}(\tau_{P})
        {\lam'_{\sigma \circ \gam_j \circ \ldots \circ {\gam'_j}^{\T} \circ {\sigma'}^{\T}} \frac{M_\sigma}{|\Aut(\sigma)|} \frac{M_{\tau_{P}}}{|\Aut(\tau_P)|} \frac{M_{\sigma'}^\T}{|\Aut(\sigma')|}}}
\end{align*}

To bound this term, we use a tail bound on $c(\alpha)$ plus the main norm bound of this section.

\begin{proposition}\label{prop:sum-c-alpha-tail}
    \[\sum_{\substack{\al \in \calS:\\ \text{non-trivial}, \\\text{permissible,} \\ |V(\al)| - \frac{|U_\al| + |V_\al|}{2} \geq k}}\frac{1}{c(\al)} \leq \frac{1}{2^{k}}.\]
\end{proposition}
\begin{proof}
    In \cref{def:c-function} we may introduce a slack of $2$
    for each power of $|V(\al)| - \frac{|U_\al| + |V_\al|}{2} .$
\end{proof}

\begin{lemma}\label{lem:large-shape}
    For all sparse permissible $\sigma, \sigma' \in \calL$, $\tau \in \calM$,
    such that $|E_{mid}(\al)| - |V(\al)| \leq C \dsos$ for each of $\sigma, \tau, \sigma'$,
    \[ \lam'_{\sigma \circ \tau \circ \sigma'^\T} \frac{\norm{M_{\sigma \circ \tau \circ \sigma'^\T}}}{|\Aut(\sigma \circ \tau \circ \sigma'^\T)|} \leq \frac{n^{2\dsos}}{c(\sigma)c(\tau)c(\sigma')}\]
\end{lemma}

\begin{proof}
Let $\al = \sigma \circ \tau \circ \sigma'^\T$.
As in previous sections, the vertex decay is enough to replace $\lam'_{\al}$ by $\lam_{\al}$.
Since all shapes obey the edge bound, the factor of $c(\sigma)c(\tau)c(\sigma')$ (which is the same as $c(\al)$ up to $\dsos^C$ per vertex)
is handled by vertex decay as in previous sections.
Unpacking the LHS, the powers of $n, d$ and $D_V$ are:
\begin{align*}
&\left(\frac{k}{n}\right)^{V(\al)-\frac{|U_\al|+|V_\al|}{2}} \left(\sqrt{\frac{p}{1-p}}\right)^{|E(\al)|}\max_{S}D_V^{|V(S)| - \frac{|L_S| + |R_S|}{2}} \sqrt{n}^{|V(\al)|-|V(S)|}\left(\sqrt{\frac{1-p}{p}}\right)^{E(S)}
\end{align*}
We show this is bounded by $n^{|U_\al \cup V_\al| + \frac{|U_\al| + |V_\al|}{2}} \cdot \left(\frac{k\sqrt{d}}{n}\right)^{|V(\al)| - \frac{|U_\al| + |V_\al|}{2}}$ for any separator $S$.
As this is bounded by $n^{2\dsos}$ times vertex decay, this will complete the proof.

\begin{align*}
= &\left(\frac{k\sqrt{d}}{n}\right)^{V(\al)-\frac{|U_\al|+|V_\al|}{2}} \left(\sqrt{\frac{p}{1-p}}\right)^{|E(\al)| - |E(S)|}\frac{D_V^{|V(S)| - \frac{|L_S| + |R_S|}{2}}}{\sqrt{d}^{|V(S)| - \frac{|U_\al| + |V_\al|}{2}}} \left(\frac{1}{\sqrt{p}}\right)^{|V(\al)|-|V(S)|}
\end{align*}
The $\sqrt{d}$ cancels out the factor of $D_V$, except for a factor of at
most $\sqrt{d}^{\frac{|U_\al| + |V_\al|}{2}}$. This is at most $n^{\frac{|U_\al| + |V_\al|}{2}}$.

The factors of $\sqrt{p}$ follow a similar strategy as in charging middle shapes. Consider a BFS process from $S$. For each vertex connected to $S$,
we can charge it to the edge leading to it.
For vertices not connected to $S$, they must be connected to $U_\al \cup V_\al$ by the connected truncation and we can consider a BFS from $U_\al\cup V_\al$.
This cancels all but $|U_\al \cup V_\al|$ factors of $1/\sqrt{p}$.
This is at most $n^{|U_\al \cup V_\al|}$.


\end{proof}

We show how to use \cref{lem:large-shape} to bound truncation error$_{\text{too many vertices}}$.
In the first term, for example,
at least one of the three parts $\sigma, \tau, \sigma'$ must have size at least $D_V / 3$, and therefore we can upper bound
that sum
by (let $\al = \sigma \circ \tau \circ \sigma'^\T$ for shorthand)
\begin{align*}
&\sum_{\substack{\text{sparse, permissible}\\ \sigma, \tau, \sigma':\\ |V(\sigma)| \geq D_V / 3}} \lam'_{\al}\frac{\norm{M_{\al}}}{|\Aut(\al)|} + \sum_{\substack{\text{sp., per.}\\ \sigma, \tau, \sigma':\\ |V(\tau)| \geq D_V / 3}} \lam'_{\al}\frac{\norm{M_{\al}}}{|\Aut(\al)|} + \sum_{\substack{\text{sp., per.}\\ \sigma, \tau, \sigma':\\ |V(\sigma')| \geq D_V / 3}} \lam'_{\al}\frac{\norm{M_{\al}}}{|\Aut(\al)|}\\
\leq & \sum_{\substack{\text{sp., per.}\\ \sigma, \tau, \sigma':\\ |V(\sigma)| \geq D_V / 3}} \frac{n^{2\dsos}}{c(\sigma)c(\tau)c(\sigma')} + \sum_{\substack{\text{sp., per.}\\ \sigma, \tau, \sigma':\\ |V(\tau)| \geq D_V / 3}} \frac{n^{2\dsos}}{c(\sigma)c(\tau)c(\sigma')} + \sum_{\substack{\text{sp., per.}\\ \sigma, \tau, \sigma':\\ |V(\sigma')| \geq D_V / 3}} \frac{n^{ 2\dsos}}{c(\sigma)c(\tau)c(\sigma')}&& (\text{\cref{lem:large-shape}})\\
\leq & \frac{3n^{2\dsos}}{2^{D_V/3}} && (\text{\cref{prop:sum-c-alpha-tail}})\\
= & \frac{3n^{2\dsos}}{2^{\frac{1}{3}C \dsos\log n}} \leq n^{-\Omega(C\dsos)}.
\end{align*}
There is one detail that we have skipped. Some terms have $\sigma$ or $\sigma'$
with $|E_{mid}(\sigma)| - |V(\sigma)| > C\dsos$. For these terms,
the Frobenius norm trick already shows
that they are at most $\frac{1}{c(\sigma) c(\tau) c(\sigma')}.$

\begin{lemma}\label{cor:frobenius-norm-trick}
   For any permissible proper shape $\al$ such that $|E(\al)| - |V(\al)| \geq C \dsos$,
   \[ \lambda'_\al  \frac{\sqrt{\E[\tr(M_\al M_\al^\T)]}}{|\Aut(\al)|} \leq \frac{1}{c(\al)}.\]
\end{lemma}
\begin{proof}
\begin{align*}
    \lam_\al  \frac{\sqrt{\E[\tr(M_\al M_\al^\T)]}}{|\Aut(\al)|} 
    & \leq \left(\frac{k}{n}\right)^{|V(\al)| - \frac{|U_\al| + |V_\al|}{2}} \sqrt{p}^{|E(\al)|}\sqrt{n}^{|V(\al)|}\\
    &= \left(\frac{k\sqrt{d}}{n}\right)^{|V(\al)| - \frac{|U_\al| + |V_\al|}{2}} \sqrt{d}^{\frac{|U_\al| + |V_\al|}{2}}\sqrt{p}^{|E(\al)| - |V(\al)|}\\
    &\leq \frac{1}{c(\al)}
\end{align*}
\end{proof}

The terms with $j \geq 1$ are bounded by $n^{-\Omega(C\dsos)}$ by
using \cref{lem:large-shape} once on $\sigma$ alone and once on $\sigma'$ alone,
using \cref{lem:intersection-terms} on $\tau_P$, and using the Frobenius norm trick
on $\sigma$ or $\sigma'$ with too many edges.
This completes the proof.

\subsubsection{\cref{lem:truncation-error}: truncation error for shapes with too many edges in one part}

\begin{lemma}
   $\norm{\text{truncation error}_{\text{too many edges in one part}}} \leq n^{-\Omega(C\dsos)}.$
\end{lemma}

Recall the quantity of our interest for this section,
 \begin{align*}
       &{\text{truncation error}}_{\text{too many edges in one part}} = \sum_{\substack{\text{sparse, permissible}\\\sigma,\tau,\sigma': \\
        |V(\sigma \circ \tau \circ \sigma'^\T)| \leq D_V,\\
        |E_{mid}(\tau)| - |V(\tau)| > C{\dsos}}}
        {\lam'_{\sigma \circ \tau \circ {\sigma'}^{\T}} \frac{M_{\sigma \circ \tau \circ {\sigma'}^T}}{|\Aut(\sigma \circ \tau \circ \sigma'^\T)|}} \\
        &+\sum_{j=1}^{2\dsos}{(-1)^{j}\sum_{\substack{\text{sparse, permissible} \\
        \sigma,\gam_j,\ldots,\gam_1,\tau,\gam'_1,\ldots,\gam'_j,\sigma': \\ 
        |V(\sigma \circ \gam_j \circ \ldots \circ \gam_1)| \leq D_V \text{ and } |V({\gam'_1}^{\T} \circ \ldots \circ {\gam'_j}^{\T} \circ {\sigma'}^{\T})| \leq D_V \\
        |E_{mid}(\gam_j)| - |V(\gam_j)| > C{\dsos} \text{ or } |E_{mid}(\gam'_j)| - |V(\gam'_j)| > C{\dsos}}}\sum_{\substack{\text{nonequiv.}\\P \in \calP^{mid}_{\gam_j, \dots, \gam_j'}}}
        {\lam'_{\sigma \circ \gam_j \circ \ldots \circ {\gam'_j}^{\T} \circ {\sigma'}^{\T}} \frac{M_{\sigma \circ \tau_{P} \circ \sigma'^\T}}{|\Aut(\sigma \circ \tau_P \circ \sigma'^\T)|} }}
    \end{align*}
    
    We will utilize the Frobenius norm trick to bound these norms.
    \begin{proposition} \label{prop:improper-frobenius-trick-first-part}
    For all shapes $\al$ (including improper shapes and shapes with (quasi-)missing edge indicators), 
    \[\E[\tr(M_\al M_\al^\T)] \leq 4^{|E(\al)|}|\Aut(\al)| 
    \max_{\beta:\text{ linearization of }\al}\left(\frac{1-p}{p}\right)^{|E(\al)| - |E(\beta)| - 2|E_{phantom}|} n^{|V(\al)| + |I_{\beta}|}.\]
\end{proposition}
\begin{proof}
    There are at most $2^{|E(\al)|}$ linearizations of $\al$.
    Using \cref{prop:linearization}, the coefficient on each one is at most $\left(\sqrt{\frac{1-p}{p}}\right)^{|E(\al)| - |E(\beta)| - 2|E_{phantom}|}$.
    Applying \cref{lem:frobenius-trick-first-part} on each linearization gives the bound.
\end{proof}

\begin{lemma} \label{lem:excess-edges-stay}
    If some part $\gamma_i$, $\tau$ or ${\gamma'_i}^\T$ has $C\dsos$ excess edges then $\tau_P$ has at least $(C-2)\dsos$ excess unique edges (i.e. $|E_{unique}(\tau_P)| - |V(\tau_P)| \geq (C-2)\dsos$).
\end{lemma}
\begin{proof}
    Observe that for the part with $C\dsos$ excess edges, we can assign a cycle to each excess edge $e$ so that the cycle contains $e$ and does not contain any other excess edge $e'$. 
    
    Even after intersections, we will still have these cycles. This implies that we can delete the $C\dsos$ excess edges and we will still have that every vertex of $\tau_P$ is connected to either $U_{\tau_P}$ or $V_{\tau_P}$. In turn, this implies that $|E(\tau_P)| - |V_{\tau_P}| \geq C\dsos - |U_{\tau_P}| - |V_{\tau_P}| \geq (C-2)\dsos$, as needed.
\end{proof}

\begin{lemma}\label{lem:Frobeniusnormbound}
For all $j\geq 0$ and sparse permissible $\sigma, \gam_j,\ldots,\gam_1,\tau,\gam'_1,\ldots,\gam'_j, \sigma'$, $ |V({\gam'_1}^{\T} \circ \ldots \circ {\gam'_j}^{\T} \circ {\sigma'}^{\T})| \leq D_V $, and $|V({\gam'_1}^{\T} \circ \ldots \circ {\gam'_j}^{\T} \circ {\sigma'}^{\T})| \leq D_V$, while $|E_{mid}(\gam_j)| - |V(\gam_j)| > C{\dsos}$ or  $|E_{mid}(\gam'_j)| - |V(\gam'_j)| > C{\dsos}$, \begin{align*}
  \E\bigg[\sum_{\text{non-equiv P}} N_P(\tau_P) {\lam'_{\sigma \circ \gam_j \circ \ldots \circ {\gam'_j}^{\T} \circ {\sigma'}^{\T}} \frac{\norm{M_{\sigma \circ \tau_{P} \circ \sigma'^\T}}_F}{|\Aut(\sigma \circ \tau_P \circ \sigma'^\T)|} }  \bigg]\leq 
  \frac{n^{-\Omega(CD_{SoS})}}{c(\sigma)c(\tau)c(\sigma')\prod_{i = 1}^j c(\gam_i)c(\gam_i')}
\end{align*}
\end{lemma}
\begin{proof}
Let  $\tau_P$ be any shape that may arise from intersection of  $\gam_j,\ldots,\gam_1,\tau,\gam'_1,\ldots,\gam'_j$. By \cref{lem:excess-edges-stay}, we have $|E_{\text{unique}}(\sigma\circ \tau\circ\sigma')|-|V(\sigma\circ \tau\circ\sigma')|\geq C\dsos$ for some $C>0$. To apply the Frobenius norm trick on possibly improper shapes, we will again first linearize the shape and then apply Frobenius norm bound,
\begin{align*}
    &\E\bigg[\sum_{\text{non-equiv }P} N_P(\tau_P) {\lam'_{\sigma \circ \gam_j \circ \ldots \circ {\gam'_j}^{\T} \circ {\sigma'}^{\T}} \frac{\norm{M_{\sigma \circ \tau_{P} \circ \sigma'^\T}}_F}{|\Aut(\sigma \circ \tau_P \circ \sigma'^\T)|} } \bigg]\\ 
  &\leq  \sum_{\text{non-equiv }P} N_P(\tau_P)\lambda'_{\sigma\circ\gam_j\circ\cdots\circ\gam_j'^\T\circ \sigma'^\T} \frac{\E[\|M_{\sigma\circ\tau_P\circ\sigma'^\T}\|_F]}{|\Aut(\sigma\circ\tau_P\circ\sigma'^\T) |}
 \end{align*}
 As in previous sections, the difference between $\lam'$ and $\lam$,
 the count of non-equivalent $P$, and $N_P(\tau_P)$ are under control from vertex decay, so we drop them. Continuing,
 \begin{align*}
 &\leq \left(\frac{k\sqrt{d}}{n}\right)^{|V(\sigma\circ\gam_j\circ\cdots\circ\gam_j'^\T\circ\sigma')|-\frac{|U_\sigma|+|U_{\sigma'}|}{2}} \sqrt{p}^{|E_{\text{unique}}(\sigma\circ\tau_P\circ\sigma'^\T)|-|V(\sigma\circ\tau_P\circ\sigma'^\T)|}\\
 & \qquad \cdot\sqrt{d}^{\frac{|U_\sigma|+|U_{\sigma'}|}{2}} \sqrt{n}^{|I_{\beta}| - 2|E_{phantom}|} && (\text{\cref{prop:improper-frobenius-trick-first-part}})  \\
 &\leq \left(\frac{k\sqrt{d}}{n}\right)^{|V(\sigma\circ\gam_j\circ\cdots\circ\gam_j'^\T\circ\sigma'^\T)|-\frac{|U_\sigma|+|U_{\sigma'}|}{2}} \sqrt{p}^{|E_{\text{unique}}(\sigma\circ\tau_P\circ\sigma'^\T)|-|V(\sigma\circ\tau_P\circ\sigma'^\T)|} \sqrt{d}^{\dsos} \\
 &\leq  \frac{n^{-\Omega(CD_{SoS})}}{c(\sigma)c(\tau)c(\sigma')\prod_{i = 1}^j c(\gam_i)c(\gam_i')}
\end{align*}
where $n^{-\Omega(CD_{SoS})}$ comes from $\sqrt{p}^{|E_{\text{unique}}(\sigma\circ\tau_P\circ\sigma'^\T)|-|V(\sigma\circ\tau_P\circ\sigma'^\T))|}$. In the second-to-last inequality we observe that any isolated vertex is incident to at least one phantom edge, hence the isolated vertices are under control.
We can then restrict our attention to the norm bound factor from the non-isolated vertices in the linearization of $\tau_P$, and the last inequality follows by the decay from excess edges.
\end{proof}
Now we use \cref{lem:frobenius-trick-second-part} to bound all their Frobenius norms together.
\[\E[\norm{\text{truncation error}_{\text{too many edges in one part}}}]
\leq \sum_{\sigma,\gam_j,\dots, \gam_j', \sigma'} \frac{n^{-\Omega(C\dsos)}}{c(\sigma)c(\tau)c(\sigma')\prod_{i = 1}^j c(\gam_i)c(\gam_i')}.\]
By a single application of Markov's inequality, the bound holds.

\subsubsection{\cref{lem:left-right-sep-conditioned}: sum of left shapes is well-conditioned}
\label{sec:well-conditioned}

\leftrightconditioning* 
\begin{proof}
Observe that
\begin{align*}
&\left(\sum_{\substack{\text{sparse,}\\\text{permissible} \\ \sigma \in \calL}}\lam'_\sigma \frac{M_\sigma}{|Aut(\sigma)|}\right)\left(\sum_{\substack{\text{sparse,}\\\text{permissible} \\ \sigma' \in \calL}}\lam'_{\sigma'} \frac{M_{\sigma'}}{|Aut(\sigma')|}\right)^\T \\
= &\sum_{j=0}^{\dsos}{\left(\sum_{\substack{\text{sparse,}\\\text{permissible} \\ \sigma \in \calL: |V_{\sigma}| = j}}\lam'_\sigma \frac{M_\sigma}{|Aut(\sigma)|}\right)\left(\sum_{\substack{\text{sparse,}\\\text{permissible} \\ \sigma' \in \calL: |V_{\sigma'}| = j}}\lam'_{\sigma'} \frac{M_{\sigma'}}{|Aut(\sigma')|}\right)^\T}
\end{align*}
Thus, it is sufficient to show that 
\[
\sum_{j=0}^{\dsos}{\left(\sum_{\substack{\text{sparse,}\\\text{permissible} \\ \sigma \in \calL: |V_{\sigma}| = j}} w_j\lam'_\sigma \frac{M_\sigma}{|Aut(\sigma)|}\right)\left(\sum_{\substack{\text{sparse,}\\\text{permissible} \\ \sigma' \in \calL: |V_{\sigma'}| = j}}w_j \lam'_{\sigma'} \frac{M_{\sigma'}}{|Aut(\sigma')|}\right)^\T} \succeq n^{-O(\dsos)}\pi
\]
for some weights $w_0,\ldots,w_{\dsos} \in [0,1]$.
To show this, we will show that if we choose the weights $w_0,\ldots,w_{\dsos} \in [0,1]$ carefully then all of the terms where $\sigma$ is non-trivial or $\sigma'$ is non-trivial can be charged to $\pi$.
\begin{definition}
Define ${\pi}_j$ to be the matrix such that $M_{AB} = 1$ if $A = B$ and $A$ is an independent set of size $j$ in $G$ and $M_{AB} = 0$ otherwise.
\end{definition}
\begin{proposition}
$\pi = \sum_{j=0}^{\dsos}{\pi_j}$
\end{proposition}
\begin{definition}
Define $\lambda'_{\pi_j} \approx \left(\frac{k}{n}\right)^{\frac{j}{2}}$ to be the coefficient $\lambda'_{\sigma}$ of the shape $\sigma$ where $V(\sigma) = U_{\sigma} = V_{\sigma}$, $|U_{\sigma}| = j$, and there is a missing edge indicator for every possible edge in $V(\sigma) = U_{\sigma} = V_{\sigma}$.
\end{definition}
We will choose our weights so that the following condition is satisfied. For all $C'$-sparse $\sigma$ and $\sigma'$ such that $\sigma$ is non-trivial or $\sigma'$ is non-trivial, letting $i = |U_{\sigma}|$ and letting $i' = |U_{\sigma'}|$,
\[
-{w_j}\lam'_{\sigma}\lam'_{\sigma'} \frac{{M_\sigma}M_{\sigma'}^{\T} + {M_{\sigma'}}M_{\sigma}^{\T}}{|Aut(\sigma)| \cdot |Aut(\sigma')|} \preceq \frac{1}{c(\sigma)c(\sigma')\dsos}\left(w_i{\lambda'}_{\pi_i}^2{\pi_i} + w_{i'}{\lambda'}_{\pi_{i'}}^2{\pi_{i'}}\right)
\]
To see why this condition is sufficient, observe that since $\sum_{\alpha}{\frac{1}{c(\alpha)}} \leq \frac{1}{4}$, 
\begin{enumerate}
    \item For all non-trivial $\sigma$ such that $|V_{\sigma}| = i$ and  $|V_{\sigma}| = j$,
    \[
    \sum_{\sigma': |V_{\sigma'}| = j}{\frac{1}{c(\sigma)c(\sigma')\dsos}w_i{\lambda'}_{\pi_i}^2{\pi_i}} \preceq \frac{2}{c(\sigma)\dsos}w_i{\lambda'}_{\pi_i}^2{\pi_i}
    \] 
    Summing this over all non-trivial $\sigma$ such that $|V_{\sigma}| = i$ and $|V_{\sigma}| = j$, we have that 
    \[
    \sum_{\sigma: |U_{\sigma}| = i, |V_{\sigma}| = j, \sigma \text{ is non-trivial}}{\sum_{\sigma': |V_{\sigma'}| = j }{\frac{1}{c(\sigma)c(\sigma')\dsos}w_i{\lambda'}_{\pi_i}^2{\pi_i}}} \preceq \frac{1}{2\dsos}w_i{\lambda'}_{\pi_i}^2{\pi_i}
    \]
    \item When $\sigma$ is trivial, 
    $\sum_{\sigma': |V_{\sigma'}| = j, \sigma' \text{ is non-trivial}}{\frac{1}{c(\sigma')\dsos}w_j{\lambda'}_{\pi_j}^2{\pi_j}} \preceq \frac{1}{4\dsos}w_j{\lambda'}_{\pi_j}^2{\pi_j}$
\end{enumerate}
Thus, we have that 
\[
\sum_{j=0}^{\dsos}{\sum_{\sigma,\sigma': |V_{\sigma}| = |V_{\sigma'}| = j, \sigma \text{ or } \sigma' \text{ is non-trivial}}{\frac{1}{c(\sigma)c(\sigma')\dsos}w_{|U_{\sigma}|}{\lambda'}_{\pi_{|U_{\sigma}|}}^2{\pi_{|U_{\sigma}|}}}} \preceq \sum_{i=0}^{\dsos}{\frac{\dsos+1}{2\dsos}w_i{\lambda'}_{\pi_i}^2{\pi_i}}
\]
The $w_{i'}{\lambda'}_{\pi_{i'}}^2{\pi_{i'}}$ terms can be bounded in the same way. This implies that
\[
\sum_{j=0}^{\dsos}{\left(\sum_{\substack{\text{sparse,}\\\text{permissible} \\ \sigma \in \calL: |V_{\sigma}| = j}}\lam'_\sigma \frac{M_\sigma}{|Aut(\sigma)|}\right)\left(\sum_{\substack{\text{sparse,}\\\text{permissible} \\ \sigma' \in \calL: |V_{\sigma'}| = j}}\lam'_{\sigma'} \frac{M_{\sigma'}}{|Aut(\sigma')|}\right)^\T} \succeq \sum_{i=0}^{\dsos}{\frac{1}{4}w_i{\lambda'}_{\pi_i}^2{\pi_i}} \succeq \frac{1}{4}\left(\min_{i \in [0,\dsos]}{\{w_i{\lambda'}_{\pi_i}^2\}}\right)\pi
\]
We now choose the weights so that for all $\sigma$ and $\sigma'$ such that $\sigma$ is non-trivial or $\sigma'$ is non-trivial, letting $i = |U_{\sigma}|$ and letting $i' = |U_{\sigma'}|$,
\[
-{w_j}\lam'_{\sigma}\lam'_{\sigma'} \frac{{M_\sigma}M_{\sigma'}^{\T} + {M_{\sigma'}}M_{\sigma}^{\T}}{|Aut(\sigma)| \cdot |Aut(\sigma')|} \preceq \frac{1}{c(\sigma)c(\sigma')\dsos}\left(w_i{\lambda'}_{\pi_i}^2{\pi_i} + w_{i'}{\lambda'}_{\pi_{i'}}^2{\pi_{i'}}\right)
\]
Observe that for all $a,b > 0$
\[
\left(a\frac{M_\sigma}{|Aut(\sigma)|} - b\frac{M_{\sigma'}}{|Aut(\sigma')|}\right)\left(a\frac{M_{\sigma}^T}{|Aut(\sigma)|} - b\frac{M_{\sigma'}^T}{|Aut(\sigma')|}\right) \succeq 0
\]
and 
\[
a^2\frac{M_{\sigma}M_{\sigma}^T}{|Aut(\sigma)|^2} + b^2\frac{M_{\sigma'}M_{\sigma'}^T}{|Aut(\sigma)|^2} \preceq a^2\norm{\frac{M_{\sigma}}{|Aut(\sigma)|}}^2{\pi_i} + b^2\norm{\frac{M_{\sigma'}}{|Aut(\sigma')|}}^2{\pi_{i'}}
\]
Choosing $a = \frac{\sqrt{\frac{w_i}{c(\sigma)c(\sigma')\dsos}}{\lambda'}_{\pi_i}|Aut(\sigma)|}{\norm{M_{\sigma}}}$ and $b = \frac{\sqrt{\frac{w_{i'}}{c(\sigma)c(\sigma')\dsos}}{\lambda'}_{\pi_{i'}}|Aut(\sigma')|}{\norm{M_{\sigma'}}}$, we require that 
\[
ab = \frac{\sqrt{{w_i}w_{i'}}{\lambda'}_{\pi_i}{\lambda'}_{\pi_{i'}}|Aut(\sigma)| \cdot |Aut(\sigma')|}{c(\sigma)c(\sigma')\dsos\norm{M_{\sigma}}\norm{M_{\sigma'}}} \geq {w_j}\lam'_{\sigma}\lam'_{\sigma'}
\]
Rearranging, we need to choose the weights so that 
\[
w_j \leq \frac{\sqrt{{w_i}w_{i'}}{\lambda'}_{\pi_i}{\lambda'}_{\pi_{i'}}|Aut(\sigma)| \cdot |Aut(\sigma')|}{c(\sigma)c(\sigma')\dsos\lam'_{\sigma}\lam'_{\sigma'}\norm{M_{\sigma}}\norm{M_{\sigma'}}}
\]
We have that 
$\frac{\lambda'_{\sigma}\norm{M_{\sigma}}}{\lambda'_{\pi_i}|Aut(\sigma)|}$ is $\tilde{O}\left(\frac{k^3}{n}\right)^{\frac{|U_{\sigma}| - |V_{\sigma}|}{2}}$.
To see this, we make the following observations:
\begin{enumerate}
    \item $\lam'_{\sigma} \approx \left(\frac{k}{n}\right)^{|V(\sigma)| - \frac{|V_{\sigma}|}{2}}\left(\sqrt{\frac{p}{1-p}}\right)^{|E(\sigma)|}$
    \item ${\lambda'}_{\pi_i} \approx \left(\frac{k}{n}\right)^{\frac{|U_{\sigma}|}{2}}$
    \item $\frac{\norm{M_{\sigma}}}{|Aut(\sigma)|}$ is $\tilde{O}\left(n^{\frac{|V(\sigma) \setminus S|}{2}}\left(\sqrt{\frac{1-p}{p}}\right)^{|E(S)|}\right)$ where $S$ is the sparse minimum vertex separator of $\sigma$.
    \item $|E(\sigma) \setminus E(S)| \geq |V(\sigma)| - |S| + |V_{\sigma}| - |U_{\sigma}|$ as there are at most $|U_{\sigma}| - |V_{\sigma}|$ connected components of $\sigma$ which do not contain a vertex in $S$.
    \item Since $\sigma$ is a left shape and $S$ is a vertex separator, $|S| \geq |V_{\sigma}|$.
    \item Since $\sqrt{\frac{np}{1-p}} > 1$ and $\frac{k}{n}\sqrt{\frac{np}{1-p}} \leq 1$, 
    \begin{align*}
        &\left(\frac{k}{n}\right)^{|V(\sigma)| - \frac{|U_{\sigma}| + |V_{\sigma}| }{2}}n^{\frac{|V(\sigma) \setminus S|}{2}}\left(\sqrt{\frac{p}{1-p}}\right)^{|V(\sigma)| - |S| + |V_{\sigma}| - |U_{\sigma}|} \\
        &\leq \left(\sqrt{\frac{np}{1-p}}\right)^{|V_{\sigma}| - |S|}\left(\frac{k}{n}\sqrt{\frac{np}{1-p}}\right)^{|V(\sigma)| - |U_{\sigma}|}
        \left(\frac{k}{n}\right)^{\frac{|U_{\sigma}| - |V_{\sigma}|}{2}}n^{\frac{|U_{\sigma}| - |V_{\sigma}|}{2}} \\
        &\leq k^{\frac{|U_{\sigma}| - |V_{\sigma}|}{2}}
    \end{align*}
    Thus, if we take $w_j = \tilde{O}\left(k^{D_{SoS} - j}\right)$ then this equation will be satisfied for all $C'$-sparse $\sigma$ and $\sigma'$. Note that as long as $E_{mid}(\sigma) \leq V(\sigma) + CD_{SoS}$, the factors of $\frac{k}{n}\sqrt{\frac{np}{1-p}}$ are enough to handle $c(\sigma)$. 
    By analyzing their Frobenius norm, it can be shown that with high probability, the terms where $E_{mid}(\sigma) > V(\sigma) + CD_{SoS}$ or $E_{mid}(\sigma') > V(\sigma') + CD_{SoS}$ have norm $n^{-\Omega(CD_{SoS})}$ and are thus negligible.
\end{enumerate}
\end{proof}




 \anote{
 Here is the rough calculation. When $\sigma = \sigma'$ is this worst case shape, we'll want that $ab = w_i\left(\frac{k}{n}\right)^{2j-i}$ as the coefficient of $M_{\sigma}M_{\sigma}^T$ is $w_i\left(\frac{k}{n}\right)^{i + 2(j-i)} = w_i\left(\frac{k}{n}\right)^{2j-i}$. We'll want that ${a^2}n^{j-i} << w_j\left(\frac{k}{n}\right)^{j}$ and ${b^2}n^{j-i} << w_j\left(\frac{k}{n}\right)^{j}$ as $\norm{M_{\sigma}}^2 \approx n^{j-i}$ and the coefficient
 for $Id_{j}$ is $w_j\left(\frac{k}{n}\right)^{j}$.

 For this to work, we need that 
 \[
 (ab)^{2}n^{2(j-i)} = w_i^2\left(\frac{k}{n}\right)^{4j-2i}n^{2(j-i)} << w_j^{2}\left(\frac{k}{n}\right)^{2j}
 \]
 which implies that $w_i << \left(\frac{1}{k}\right)^{j-i}{w_j}$
 }

\section{Open Problems}
\label{sec:open-problems}
Several other problems on sparse graphs are conjectured to be hard for SoS 
and it is our hope that the techniques here can help prove that these problems are hard for SoS.
These problems include MaxCut, $k$-Coloring, and Densest-$k$-Subgraph.
For MaxCut in particular, since there are no constraints other than booleanity of
the variables it may be possible to truncate away dense shapes, which we could not
do here due to the presence of independent set indicator functions.

Another direction for further research is to handle random graphs which are not \Erdos-R{\'e}nyi.
Since the techniques here depend on graph matrix norms,
one would hope that they generalize to distributions such as $d$-regular graphs for which low-degree
polynomials are still concentrated. However, in the non-iid setting, it is not clear what the analogue of graph
matrices should be used due to the lack of a Fourier basis that is friendly to work with.

The polynomial constraint ``$\sum_{v \in V} x_v = k$'' is not satisfied exactly by our
pseudoexpectation operator. It's possible that techniques from \cite{Pang21}
can be used to fix this.

The parameters in this paper can likely be improved. One direction is to remove the final factor of $\log n$ from our bound. This would allow us to prove an SoS lower bound for the ``ultrasparse regime'' $d = O(1)$ rather than $d \geq \log^2 n$. This setting is interesting as there is a nontrivial algorithm that finds an independent set of half optimal size~\cite{GS17,RV17}. Furthermore, this algorithm is local in a sense that we don't define here. It would be extremely interesting if this algorithm could be converted into a rounding algorithm for constant-degree SoS.

Another direction is to improve the dependence on $\dsos$. While our bound has a $\frac{1}{poly(\dsos)}$ dependence on $\dsos$, we conjecture that the dependence should actually be $(1-p)^{O(\dsos)}$. If so, this would provide strong evidence for the prevailing wisdom in parameterized complexity and proof complexity that a maximum independent set of size $k$ requires $n^{\Omega(k)}$ time to find/certify (corresponding to SoS degree $\Omega(k)$).

\bibliographystyle{alphaurl}
\bibliography{macros,madhur}

\appendix

\section{Trace Power Calculations}\label{sec: graph_matrix_norm_bounds}
Let $\Omega$ be the distribution which is $\sqrt{\frac{1-p}{p}}$ with probability $p$ and $-\sqrt{\frac{p}{1-p}}$ with probability $1-p$.
\begin{proposition}\label{prop:powerexpectedvalue}
For all $k \in \mathbb{N}$ and all $p \leq \frac{1}{2}$, $\left|\E_{\Omega}[x^k]\right| \leq \left(\sqrt{\frac{1-p}{p}}\right)^{k-2}$ and $\left|\E_{\Omega}\left[{x^k}1_{x = -\sqrt{\frac{p}{1-p}}}\right]\right| \leq \left(\sqrt{\frac{1-p}{p}}\right)^{k-2}$.
\end{proposition}
\begin{proof}
For the first part, observe that for $k = 1$, $E_{\Omega}[x^k] = 0$ and for $k \geq 2$,
\[
E_{\Omega}[x^k] = p\left(\sqrt{\frac{1-p}{p}}\right)^{k} + (1-p)\left(-\sqrt{\frac{p}{1-p}}\right)^{k} = 
(1-p)\left(\sqrt{\frac{1-p}{p}}\right)^{k-2} + p\left(-\sqrt{\frac{p}{1-p}}\right)^{k-2}
\]
and since $\sqrt{\frac{p}{1-p}} \leq \sqrt{\frac{1-p}{p}}$ and $1-p \geq p$, $0 \leq (1-p)\left(\sqrt{\frac{1-p}{p}}\right)^{k-2} + p\left(-\sqrt{\frac{p}{1-p}}\right)^{k-2} \leq \left(\sqrt{\frac{1-p}{p}}\right)^{k-2}$

For the second part, observe that for all $k \geq 1$
\[
\left|\E_{\Omega}\left[{x^k}1_{x = -\sqrt{\frac{p}{1-p}}}\right]\right| = (1-p)\left(\sqrt{\frac{p}{1-p}}\right)^{k} \leq \left(\sqrt{\frac{1-p}{p}}\right)^{k-2}
\]
\end{proof}




Before we embark upon the trace method calculation, it is necessary to note that throughout our PSD analysis, our graphical matrix for shape has rows/columns indexed by sets (as it is a sum of ribbons which have rows/columns indexed by set); however, for the trace method calculation, it is easier to work with graph matrices indexed by ordered tuples and then identify the change in trace method encoding that allows us to shift to set-indexed matrices.
\begin{definition}[Ordered-ribbon]
Given a ribbon $R$, let $\mathcal{O}$ be an ordering for vertices in $U_R\cup V_R$, we can define graphical matrix for an ordered-ribbon to have rows and columns indexed by subset of $[n]$  with a single non-zero entry\[ 
M^{(\text{ordered})}_{R,\mathcal{O} }[(I,\mathcal{O_I}), (J,\mathcal{O_J}) ] = \chi_{E(R)}(G)
\]
for $(I,\mathcal{O_I})=(A_R, \mathcal{O_{A_R}}), (J,\mathcal{O_J})=(B_R, \mathcal{O_{B_R}})$.
\end{definition}
It is immediate from the above definition that for $S,T\subseteq [n]$, ribbon can be obtained from ordered-ribbons by summing over ordering of the labeled vertices in $U_R\cup V_R$, \[ 
M_R[S,T] = \sum_{\mathcal{O}:\text{ordering for vertices in $U_R\cup V_R$} } M^{(\text{ordered})}_{R,\mathcal{O}} [(S,\mathcal{O_S}), (T,\mathcal{O_T}) ]
\]
\begin{definition}[Graphical matrix with (unsymmetrized) ordered tuples]
We analogously define graphical matrix indexed by ordered tuples by summing over ribbons with ordered tuples, for a shape $\al$,\[ 
M^{\text{(ordered)}}_\al = \sum_{\substack{R:\text{ribbon of shape }\al,\\ \mathcal{O}:\text{ordering}}} M^{\text{(ordered)}}_{(R,\mathcal{O}) }
\]
\end{definition}
\begin{definition}[Graphical matrix with sym-ordered tuples]
\[
M_\al^{(sym-ordered)}[(I,\mathcal{O_I}), (J, \mathcal{O_J})] = \sum_{\mathcal{O}:\text{ordering for $U_\al\cup V_\al$}} M^{\text{(ordered)}}_{R,\mathcal{O}}
\]
\end{definition}
\begin{remark}
The graphical matrix for shape indexed by sets is a principle submatrix of the the graphical matrix with sym-ordered tuples.
\end{remark}

We now proceed onto the analysis for graph matrices indexed by (unsymmetrized) ordered tuples.
\begin{definition}
Given a shape $\alpha$ and a separator $S$ between $U_{\alpha}$ and $V_{\alpha}$, we make the following definitions:
\begin{enumerate}
\item We define $E(S)$ to be the set of edges with both endpoints in $S$.
\item We define $L_S$ to be the set of vertices in $S$ which are reachable from $U_{\alpha}$ without passing through any other vertices in $S$.
\item We define $R_S$ to be the set of vertices in $S$ which are reachable from $V_{\alpha}$ without passing through any other vertices in $S$.
\item Define $c(S)$ to be the number of connected components of $\alpha \setminus S$ which are not reachable from $U_{\alpha}$ or $V_{\alpha}$.
\end{enumerate}
\end{definition}
\begin{theorem}\label{thm:tracepowercalculations}
For any proper shape $\alpha$ (including shapes with missing edge indicators) and any $q \in \mathbb{N}$, letting $M_{\alpha}$ be the graph matrix where the rows and columns are indexed by ordered tuples rather than sets,
\begin{align*}
&\E\left[\tr\left(\left(M_{\alpha}M_{\alpha}^\T\right)^q\right)\right] \leq 4^{2q\left(|V(\alpha)| - \frac{|U_{\alpha}| + |V(\alpha)|}{2}\right)}\left(\sqrt{|V(\alpha)|}\right)^{2q|V(\alpha) \setminus (U_{\alpha} \cup V_{\alpha})|}n^{q|V(\alpha)|} \cdot \\
&\left(\max_{\text{separator } S}{n^{-\frac{|S|}{2}}\left(3\sqrt{\frac{1-p}{p}}\right)^{|E(S)|}\left(\sqrt{|V(\alpha)|}\right)^{|S| - |U_{\alpha} \cap V_{\alpha}|}(2q)^{|S|-\frac{|L_S| + |R_S|}{2} + \frac{c(\alpha,S)}{2}}}\right)^{2q-2}
\end{align*}
\end{theorem}

\begin{proof}
We use the idea from Appendix B of \cite{AMP20}. Given a constraint graph on $V(\alpha,2q)$, for each vertex in $V(\alpha,2q)$ which is not a copy of a vertex in $U_{\alpha} \cap V_{\alpha}$, we note whether it has a constraint edge to the left and whether it has a constraint edge to the right. Note that there are at most $4^{2q\left(|V(\alpha)| - \frac{|U_{\alpha}| + |V(\alpha)|}{2}\right)}$ possibilities for this (note that vertices in $(U_{\alpha} \cup V_{\alpha}) \setminus (U_{\alpha} \cap V_{\alpha})$ only have $q$ copies in $V(\alpha,2q)$). For each copy of $\alpha$, we then construct a separator $S$ as follows:
\begin{enumerate}
\item We include any vertex which has both a constraint edge to the left and a constraint edge to the right in $S$.
\item We include any vertex in $U_{\alpha}$ which has a constraint edge to the right in $S$.
\item We include any vertex in $V_{\alpha}$ which has a constraint edge to the left in $S$.
\item We include all vertices in $U_{\alpha} \cap V_{\alpha}$.
\end{enumerate}
We choose separators for the copies of $\alpha^{T}$ in the same way. The theorem follows from the following lemma.
\begin{lemma}
Given any separators $S_3, S_5, \ldots, S_{2q-1}$ for the copies of $\alpha$ (except the first one) and any separators $S_2,S_4, \ldots,S_{2q-2}$ for the copies of $\alpha^{\T}$, there are at most 
\[
\left(\sqrt{|V(\alpha)|}\right)^{2q|V(\alpha) \setminus (U_{\alpha} \cup V_{\alpha})|}n^{q|V(\alpha)|}\left(\prod_{j=2}^{2q-1}{n^{-\frac{|S_j|}{2}}3^{|E(S_j)|}\left(\sqrt{|V(\alpha)|}\right)^{|S_j| - |U_{\alpha} \cap V_{\alpha}|}(2q)^{|S_j|-\frac{|L_{S_j}| + |R_{S_j}|}{2} + \frac{c(S_j)}{2}}}\right)
\]
terms in $\E\left[\tr\left(\left(M_{\alpha}M_{\alpha}^\T\right)^q\right)\right]$ which have these separators and have nonzero expected value and for each of these terms the expected value is at most $\prod_{j=2}^{2q-1}{\left(\sqrt{\frac{1-p}{p}}\right)^{|E(S_j)|}}$.
\end{lemma}
\begin{proof}
To see that the expected value is at most $\prod_{j=2}^{2q-1}{\left(\sqrt{\frac{1-p}{p}}\right)^{|E(S_j)|}}$, observe that by Proposition \ref{prop:powerexpectedvalue}, the expected value of a term is at most 
\[
\prod_{e}\left(\sqrt{\frac{1-p}{p}}\right)^{(\text{multiplicity of } e) - 2}
\]
and this holds even if $\alpha$ contains missing edge indicators. We can view this as giving each edge which is making a middle appearance (i.e. it has appeared before and will appear again) a factor of $\left(\sqrt{\frac{1-p}{p}}\right)$ and giving edges which are appearing for the first or last time a factor of $1$ (if an edge appears exactly once, it would instead get a factor of $\sqrt{\frac{p}{1-p}}$, but this only helps us). Since edges must be appearing for the first or last time unless both endpoints are in the separator, the result follows. 

To count the number of terms, for each copy of $\alpha$, we can go through the copies of $\alpha$ and $\alpha^{\T}$ and identify their vertices one by one. We give each vertex in $U_{\alpha} \cap V_{\alpha}$ a label in $[n]$ and we only need to do this once. For the other vertices:
\begin{enumerate}
\item If the vertex $v$ has no constraint edge to the left and is thus appearing for the first time, we give $v$ a label in $[n]$.
\item If the vertex $v$ is adjacent to a vertex $u$ which has been identified and has no constraint edge to the right, since $u$ is appearing for the last time, there are $deg(u) \leq |V(\alpha)|$ edges incident to $u$ which have not appeared for the last time before this copy of $\alpha$ or $\alpha^{\T}$. To identify $v$, it is sufficient to identify which of these edges is $(u,v)$.
\item Otherwise, we can pay a factor of $2q|V(\alpha)|$ to identify the vertex as it has appeared before.
\end{enumerate}
For each edge in $E(S)$, we note whether the edge is appearing for the first time, appearing for the last time, or is making a middle appearance. For edges which are not in $E(S)$, one of its endpoints will not have a constraint edge to the left or a constraint edge to the right and this determines whether this is the first or last time the edge appears.

We now consider the factors required by each vertex in the $\frac{j+1}{2}$-th copy of $\alpha$ (which is not a copy of a vertex in $U_{\alpha} \cap V_{\alpha}$) for $j = 3,5,\ldots,2q-1$. Similar reasoning applies to copies of $\alpha^{T}$.
\begin{enumerate}
\item Vertices in $U_{\alpha}$ are already identified.
\item For vertices which are not in $U_{\alpha}$ but are reachable from $U_{\alpha}$ without passing through $S_j$ (including vertices in $L_{S_j}$), when we want to identify these vertices, they will have a neighbor which is already identified and is appearing for the last time. Thus, each of these vertices requires a factor of at most $|V(\alpha)|$.
\item For vertices which are reachable from $V_{\alpha}$ without passing through $S_j$ (and which are not in $S_j$ themselves), these vertices have no constraint edge to the left, so this is the first time these vertices appear. Thus, each of these vertices requires a factor of $n$.
\item For vertices in $S_j \setminus L_{S_j}$, these vertices have appeared before but are not adjacent to a vertex which has been identified and has no constraint edge to the right, so we pay a factor of $2q|V(\alpha)|$ for these vertices
\item For a component of $V(\alpha) \setminus S_j$ which is not reachable from $U_{\alpha}$ or $V_{\alpha}$, all of the vertices in this component will either have no constraint edges to the left or no constraint edges to the right. If the vertices in this component have no constraint edges to the left, they are appearing for the first time and require a factor of $n$.  If the vertices in this component have no constraint edges to the right, we can pay a factor of $2q|V(\alpha)|$ for the first vertex in this component and a factor of $|V(\alpha)|$ for all of the other vertices in this component.

\end{enumerate}
We can also try identifying vertices by going from the right to the left. Following similar logic, this gives us the following factors.
\begin{enumerate}
\item Vertices in $V_{\alpha}$ are already identified.
\item For vertices which are not in $V_{\alpha}$ but are reachable from $V_{\alpha}$ without passing through $S_j$ (including vertices in $R_{S_j}$), we need a factor of at most $|V(\alpha)|$.
\item For vertices which are reachable from $U_{\alpha}$ without passing through $S_j$ (and which are not in $S_j$ themselves), we need a factor of $n$.
\item For vertices in $S_j \setminus R_{S_j}$, we need a factor of $2q|V(\alpha)|$.
\item For a component of $V(\alpha) \setminus S_j$ which is not reachable from $U_{\alpha}$ or $V_{\alpha}$, all of the vertices in this component will either have no constraint edges to the left or no constraint edges to the right. If the vertices in this component have no constraint edges to the right, they are appearing for the first time (when going from the right to the left) and require a factor of $n$.  If the vertices in this component have no constraint edges to the left, we can pay a factor of $2q|V(\alpha)|$ for the first vertex in this component and a factor of $|V(\alpha)|$ for all of the other vertices in this component.

\end{enumerate}
We can choose whichever direction gives a smaller bounds. Since we don't know which direction this will be, we take the square root of the product of the two bounds. We now consider how many factors of $n$, $|V(\alpha)|$, and $2q$ this gives us. For $S_2,\ldots,S_{2q-1}$, 
\begin{enumerate}
\item Factors of $n$: Each vertex in $S_j$ does not give us a factor of $n$. Each vertex outside of $S_j$ gives a factor of $n$ for exactly one of the two directions.
\item Factors of $|V(\alpha)|$: Vertices which are copies of vertices in $U_{\alpha} \cap V_{\alpha}$ do not need any factors of $|V(\alpha)|$. For the other vertices, we have the following cases. For each vertex $v \in S_j$ ,
\begin{enumerate}
    \item If $v \notin U_{\alpha} \cup V_{\alpha}$ then $v$ requires a factor of $|V(\alpha)|$ in both directions.
    \item If $v \in U_{\alpha} \setminus V_{\alpha}$ or $v \in V_{\alpha} \setminus U_{\alpha}$ then $v$ requires a factor of $|V(\alpha)|$ in one of the two directions.
\end{enumerate}
For each vertex $v \notin S_j$, 
\begin{enumerate}
    \item If $v \notin U_{\alpha} \cup V_{\alpha}$ then $v$ requires a factor of $|V(\alpha)|$ in one of the two  directions.
    \item If $v \in U_{\alpha} \setminus V_{\alpha}$ or $v \in V_{\alpha} \setminus U_{\alpha}$ then $v$ does not require a factor of $|V(\alpha)|$ as in one direction it will already be specified and in the other direction it will be a new vertex and will thus require a factor of $n$ rather than $|V(\alpha)|$.
\end{enumerate}
\item Factors of $2q$: Vertices in $S_j \setminus (L_{S_j} \cup R_{S_j})$ require a factor of $q$ in both directions. Vertices in $L_{S_j} \setminus R_{S_j}$ and $L_{S_j} \setminus R_{S_j}$ require a factor of $q$ for exactly one of the two directions. Vertices in $L_{S_j} \cap R_{S_j}$ do not require any factors of $q$. Each component of $V(\alpha) \setminus S_j$ which is not reachable from $U_{\alpha}$ or $V_{\alpha}$ requires a factor of $q$ for one of its vertices for exactly one of the two directions.

\end{enumerate}

For the first copy of $\alpha$, we need a factor of $n$ for each vertex as it is the first time all of the vertices appear (this also accounts for the vertices in $U_{\alpha} \cap V_{\alpha}$). For the final copy of $\alpha^T$, the vertices which are copies of vertices in $U_{\alpha} \cup V_{\alpha}$ have already been identified and we need a factor of at most $|V(\alpha)|$ to specify each vertex which is not a copy of a vertex in $U_{\alpha} \cup V_{\alpha}$ as each such vertex will be adjacent to a vertex which has been identified and has no constraint edge to the right.

Putting everything together gives the result.
\end{proof}
\end{proof}
\section{Reserving edges to preserve left/middle shapes}\label{sec:reserved}

Given a shape $\al$, and a subset of edges $\res(\al)$ that we are to reserve, we will only apply conditioning on edges $E(\al)\setminus \res(\al)$.

\begin{proposition}\label{prop:middle-reservation}
    Given any middle shape $\tau$ such that all vertices have a path to $U_\tau \cup V_\tau$, there is a set of edges $\res(\tau)$ of size at most $\abs{V(\tau)\setminus (U_\tau\cap V_\tau)}$ such that every vertex
    has a path to $U_\tau \cup V_\tau$ in $\res(\tau)$, and removing any subset of the edges $E(\tau)\setminus \res(\tau)$ does not change the 
    size of the minimum vertex separator of the shape.
\begin{proof}
    Select $\abs{U_\tau}$ vertex-disjoint paths between $U_\tau$ and $V_\tau$.
    This is enough to ensure that the minimum vertex separator will not decrease.
    To ensure that every vertex has a path to $U_\tau \cup V_\tau$, 
    also add in edges connecting to these paths until all vertices are connected.
\end{proof}
\end{proposition}

\begin{proposition}
For a middle intersection ${\gam, \tau, \gam'}$ we can reserve at most $|V(G) \setminus (U_\gam \cap V_\gam)|$ edges in $\gam$ so that the reserved ribbons are still a middle intersection.
\end{proposition}
\begin{proof}
    Select $|U_\gam|$ vertex-disjoint paths between $U_\gam$ and $\Int(\gam) \cup V_\gam$.
\end{proof}

    \begin{proposition}\label{lem: conditioning_left}
For a left shape $\gamma$, there is a set $\res(\gam) \subseteq E(\gamma)$ of edges such that $|\res(\gam)| \leq 2|V(\gamma)|$ and reserving $\res(\gam)$ is sufficient to guarantee $\gamma$ remains a left shape under removal of any subset of edges in $E(\gamma)\setminus \res(\gam)$.
\end{proposition}
\begin{proof}
Observe that if $\gamma$ is a left part, for each vertex $v \in V_{\gamma}$, if we add another copy of $v$ (duplicating all of the edges incident to $v$) to $V_{\gamma}$ (and $U_{\gamma}$ if $v \in U_{\gamma} \cap V_{\gamma}$), then there are $|V_{\gamma}| + 1$ vertex disjoint paths 
from $U_{\gamma}$ to $V_{\gamma} \cup \{v_{duplicated}\}$ where $v_{duplicated}$ is the copy of $v$. Moreover, this condition is sufficient to guarantee that $\gamma$ is a left part.

To see this, assume that there is a separator $S$ of size $|V_{\gamma}|$ between $U_{\gamma}$ and $V_{\gamma} \cup \{v_{duplicated}\}$. Note that $S \cap V(\gamma)$ is a separator between $U_{\gamma}$ and $V_{\gamma}$ and $V_{\gamma}$ is the only such separator of size at most $|V_{\gamma}|$, so we must have that $S = V_{\gamma}$. However, this is impossible as there is a path from $U_{\gamma}$ to $v_{duplicated}$ which does not pass through $S$.

To see the moreover part, assume there is a separator $S$  between $U_{\gamma}$ and $V_{\gamma}$ of size at most $|V_{\gamma}|$ which is not equal to $V_{\gamma}$. Since $S = V_{\gamma}$, there is a vertex $v \in V(\gamma) \setminus S$. This condition says that there are $|V_{\gamma}| + 1$ paths from $U_{\gamma}$ to $V_{\gamma}$ which are vertex disjoint except that two of the paths end at $v$. At least one of these paths does not pass through $S$, so $S$ is not a separator between $U_{\gamma}$ and $V_{\gamma}$. Thus, there is no such separator $S$ and $\gamma$ is a left part.

Now consider the following graph $G$ which is designed so that there is an integer-valued flow of value $k$ from $s$ to $t$ in $G$ corresponding to $k$ vertex disjoint paths from $U_{\gamma}$ to $V_{\gamma}$.
\begin{enumerate}
\item For each vertex $u \in V(\gamma)$, create two copies $u_{in}$ and $u_{out}$ of $u$. Add a directed edge with capacity $1$ from $u_{in}$ to $u_{out}$.
\item Add directed edges with capacity $1$ from $s$ to $u_{in}$ for each $u \in U_{\gamma}$. Similarly, add directed edges with capacity $1$ from $v_{out}$ to $t$ for each $v \in V_{\gamma}$.
\item For each edge $(u,v) \in E(\gamma)$, create a directed edge with capacity $1$ from $u_{out}$ to $v_{in}$ and a directed edge with capacity $1$ from $v_{out}$ to $u_{in}$.
\end{enumerate}
Choose $|V_{\gamma}|$ vertex disjoint paths from $U_{\gamma}$ to $V_{\gamma}$ and take the corresponding paths in $G$. Now observe that for each $v \in V_{\gamma}$, if we increased the capacity of the edges from $v_{in}$ to $v_{out}$ and from $v_{out}$ to $t$ to $2$, this would increase the max flow value to $|V_{\gamma}| + 1$. This implies that if we run an iteration of Ford-Fulkerson on the residual graph, we must be able to reach $v_{in}$ for each $v \in V_{\gamma}$. If we use DFS or BFS to explore the residual graph, we need at most $|V(\alpha)|$ of the non $u_{in}$ to $u_{out}$ edges to reach all of these destinations.

Adding these edges to the $|V_{\gamma}|$ vertex disjoint paths from $U_{\gamma}$ to $V_{\gamma}$ gives a set of edges of size at most $2|V(\gamma)|$ which can be reserved to guarantee that $\gamma$ remains a left part.
\end{proof}

\section{Second proof of \cref{lem:intersection-charging}}
\label{app:intersection-charging}

\intersectionCharging*

We prove this by exhibiting an explicit charging argument.

\begin{proof}
    Let $\tau_P^{phant}$ be the multigraph formed from $\beta$ plus two edges
    for each phantom edge in $E_{phantom}$. We need to assign each vertex of $V\left(\tau_P^{phant}\right) \setminus V(S)$ an edge of $E\left(\tau_P^{phant}\right) \setminus E(S)$,
    and we need to assign isolated vertices in $I_\beta$ two edges.
    Note that the connectivity of $\tau_P^{phant}$ is exactly the same as $\tau_P$,
    just the nonzero edge multiplicities are modified.
    
    For vertices that are in the same connected component of $\tau_P^{phant}$ as a vertex in $S$, 
    consider running a breadth-first search from $S$, and assign each edge
    to the vertex it explores.
    Vertices in $I_\beta$ must be explored via a double edge, in order for them
    to become isolated during linearization, so assign both.
    
    For vertices that are not in the same component of $\tau_P^{phant}$ as $S$,
    the edge assignment is more complicated. Let $C$ be a component.
    Using \cref{prop:intersection-connectivity}, component $C$ must intersect both $U_{\beta}$ and
    $V_{\beta}$.
    For the isolated vertices in $I_\beta \cap C$, order them by their distance from $U_{\beta} \cup V_{\beta}$.
    Charge them to the two edges along the shortest path to $U_\beta \cup V_\beta$.
    For the remaining non-isolated vertices, we claim that there is at least one additional double edge in $C$ that has not yet been charged.
    As noted just above using \cref{prop:intersection-connectivity}, there is a path $P$ between $U_\beta$ and $V_\beta$ in this component.
    However, since $C$ does not intersect $S$, which is by assumption a separator in $\beta$, the path cannot be entirely in $\beta$ i.e. it must contain at least
    one double edge.
    If $P$ does not pass through an isolated vertex, this double edge evidently has
    not yet been charged. 
    If $P$ passes through an isolated vertex,
    because all edges incident to isolated vertices are double edges,
    there must be more double edges than isolated vertices in $P$.
    Since each isolated vertex only charges one incident double edge, 
    they can't all be charged.
    In either case, $P$ contains an uncharged double edge.

    Now contract the edges in $\tau_P^{phant}$ that were charged for isolated vertices, and
    order the non-isolated vertices by their distance from this double edge.
    The double edge can be used to charge its two endpoints (which are
    not isolated), and the other vertices can be charged using the next
    edge in the shortest path to the double edge. This completes the charging, and the proof of
    the lemma.
\end{proof}
\section{Formalizing the PSD Decomposition}
\label{sec:formal-decomposition}

Here we prove \cref{lem:fullribbondecomposition} (which we restate here for convenience) by formally going through the approximate PSD decomposition described in \cref{sec:informal-decomposition}.

\fullRibbonDecomposition*

\begin{proof}
We abuse notation and write the sum over ribbons in $\calS$ as
\[ \sum_{\alpha \in \calS}\lam_\al \frac{M_\alpha}{\abs{\Aut(\al)}} = \sum_{R \in \calS}\lam_R M_R.\]
(The automorphism group disappears as it was only there to ensure each
ribbon is represented once anyway.)




From~\cref{lem:shape-decomposition}, every ribbon $R \in {\calS}$ decomposes into
$R = L \circ T \circ L'^\T$ where $L,L'~\in~\Lribbon$
and $T~\in~\Mribbon$.
We would like that
\[ \sum_{R \in {\calS}}\lam_R M_R = \sum_{\substack{L,L' \in \Lribbon,T \in \Mribbon:\\\abs{V(L \circ T \circ L'^\T)} \leq D_V}}\lam_{L\circ T \circ L'^\T}M_{L \circ T \circ L'^\T}.\]
However, this is not quite correct, as the ribbons on the right hand side may intersect. 

\begin{proposition}
\[ \sum_{R \in {\calS}}\lam_R M_R = \sum_{\substack{L,L' \in \Lribbon,T \in \Mribbon:\\\abs{V(L \circ T \circ L'^\T)} \leq D_V,\\ L,T,L'^\T\text{ are properly composable}}}\lam_{L\circ T \circ L'^\T}M_{L \circ T \circ L'^\T}.\]
\end{proposition}


To start, we uncorrelate the sizes of $L, T, L'$. Add to the moment matrix
\[ \sum_{\substack{L,L' \in \Lribbon,T \in \Mribbon:\\\abs{V(L \circ T \circ L'^\T)} > D_V,\\ L,T,L'^\T\text{ are properly composable}}}\lam_{L\circ T \circ L'^\T}M_{L \circ T \circ L'^\T}.\]
This term is added to $\text{truncation error}_\text{too many vertices}$.

We define matrices for intersection terms $I_k$, factored terms $F_k$, and truncation error $T_k$ at level $k$
recursively via the following process. We will maintain that the moment matrix satisfies for all $k$,
\[ \sum_{R \in {\calS}} \lam_R M_R = \left(\sum_{L \in \Lribbon} \lam_L M_L\right)\left(\sum_{i=1}^k (-1)^{i+1}F_i\right)\left(\sum_{L \in\calL}\lam_L M_L\right)^\T + (-1)^{k}I_k + \sum_{i=1}^k{(-1)^{i}T_i}.\]
In order to make this equation hold, given $I_k$ we will choose $F_{k+1}$, $T_{k+1}$, and $I_{k+1}$ so that 
\[
I_k = \left(\sum_{L \in \Lribbon} \lam_L M_L\right)F_{k+1}\left(\sum_{L \in\calL}\lam_L M_L\right)^\T - T_{k+1} -I_{k+1}
\]

At the start of the $(k+1)$th iteration, we will have that $I_k$ is equal to a sum over middle ribbons $T$, ``middle intersecting ribbons'' $(G_k , \ldots, G_1,{G'_1}^{\T}, \ldots, {G'_k}^{\T}) \in \midint^{(k)}_{T}$, and non-intersecting ribbons $L,L'$. Furthermore, the ribbons $G_k, \dots,T,\dots, G_k'$ will satisfy $|E_{mid}(R)| - |V(R)| \leq C\dsos$, in which case we say that ``edge bounds hold''.
Any time that a violation of this bound appears, we will throw the term
into the truncation error for too many edges in one part.
For example, initially we throw away $T$ with too many edges.
Initially,
\begin{align*}
I_0 &\defeq \sum_{\substack{L,L' \in \Lribbon,T \in \Mribbon:\\\abs{V(L)} \leq D_V, \abs{V(T)} \leq D_V, \abs{V({L'}^{\T})} \leq D_V,\\ L,T,L'^\T\text{ are properly composable},\\
\text{edge bounds hold}}}\lam_{L \odot T \odot L'^\T}M_{L \circ T \circ L'^\T}\\
T_0 &\defeq \sum_{\substack{L,L' \in \Lribbon,T \in \Mribbon:\\\abs{V(L)} \leq D_V, \abs{V(T)} \leq D_V, \abs{V({L'}^{\T})} \leq D_V,\\ 
|E_{mid}(T)| - |V(T)| > C \dsos,\\ 
L,T,L'^\T\text{ are properly composable}}}\lam_{L \odot T \odot L'^\T}M_{L \circ T \circ L'^\T}\\
F_0 &\defeq 0
\end{align*}

\begin{definition}
    For a middle ribbon $T$, let $\midint^{(k)}_T$ be the collection of tuples of ribbons
    $G_k, \dots, G_1, G_1', \dots, G_k'$
    such that the ribbons are composable and the induced intersection pattern of ${G_k \circ \cdots \circ G_1 \circ T \circ G_1'^\T \circ \cdots \circ G_k'^\T}$ is a middle intersection pattern.
\end{definition}
The inductive hypothesis for $I_k$ is, 
\[ I_k = \sum_{\substack{T \in \calM, (G_k ,\dots, {G'_k}) \in \midint^{(k)}_{T}, L,L' \in \calL:\\
\abs{V(L\odot G_k \odot \ldots \odot G_1)} \leq D_V, \abs{V(T)} \leq D_V, \abs{V({G'_1}^{\T} \odot \ldots \odot {G'_k}^{\T} \odot {L'}^{\T})} \leq D_V, \\
L, G_k \circ \ldots \circ G_1 \circ T \circ {G'_1}^{\T} \circ \ldots \circ {G'_k}^{\T}, L'^\T\text{ are properly composable},\\
\text{edge bounds satisfied}}}
\lam_{L \odot G_k \odot \ldots \odot G_1 \odot T \odot {G'_1}^{\T} \odot \ldots \odot {G'_k}^{\T} \odot {L'}^{\T}} M_{L \circ G_k \circ \ldots \circ G_1 \circ T \circ {G'_1}^{\T} \circ \ldots \circ {G'_k}^{\T} \circ {L'}^T}.\]
We can approximate $I_k$ by
\[
\left(\sum_{L \in \Lribbon} \lam_L M_L\right)F_{k+1}\left(\sum_{L \in\calL}\lam_L M_L\right)^\T
\]
where 
\[
F_{k+1} = \sum_{\substack{T \in \calM, (G_k ,\dots, {G'_k}) \in \midint^{(k)}_{T}:\\
\abs{V(G_k \odot \ldots \odot G_1)} \leq D_V, \abs{V(T)} \leq D_V, \abs{V({G'_1}^{\T} \odot \ldots \odot {G'_k}^{\T})} \leq D_V,\\
\text{edge bounds satisfied}}}
\lam_{G_k \odot \ldots \odot G_1 \odot T \odot {G'_1}^{\T} \odot \ldots \odot {G'_k}^{\T}} M_{G_k \circ \ldots \circ G_1 \circ T \circ {G'_1}^{\T} \circ \ldots \circ {G'_k}^{\T}}
\]
In order to make it so that $I_k = \left(\sum_{L \in \Lribbon} \lam_L M_L\right)F_{k+1}\left(\sum_{L \in\calL}\lam_L M_L\right)^\T - T_{k+1} -I_{k+1}$, we need to handle the following issues.
\begin{enumerate}
    \item $I_k$ only contains terms where $\abs{V(L\odot G_k \odot \ldots \odot G_1)} \leq D_V$ and ${\abs{V({G'_1}^{\T} \odot \ldots \odot {G'_k}^{\T} \odot {L'}^{\T})} \leq D_V}$, whereas some larger terms appear in the approximation. To handle this, we add 
    \[
    \sum_{\substack{T \in \calM, (G_k ,\dots,  {G'_k}) \in \midint^{(k)}_{T}, L,L' \in \calL:\\
    \abs{V(G_k \odot \ldots \odot G_1)} \leq D_V, \abs{V(T)} \leq D_V, \abs{V({G'_1}^{\T} \odot \ldots \odot {G'_k}^{\T})} \leq D_V, \\
        \abs{V(L\odot G_k \odot \ldots \odot G_1)} > D_V \text{ or } \abs{V({G'_1}^{\T} \odot \ldots \odot {G'_k}^{\T} \odot {L'}^{\T})} > D_V,\\
\text{edge bounds satisfied}}}
    \lam_{L \odot G_k \odot \ldots \odot G_1 \odot T \odot {G'_1}^{\T} \odot \ldots \odot {G'_k}^{\T} \odot {L'}^{\T}} {M_L}M_{G_k \circ \ldots \circ G_1 \circ T \circ {G'_1}^{\T} \circ \ldots \circ {G'_k}^{\T}}{M_{L'}}^{\T}
    \]
    to $T_{k+1}$. This is the source of the truncation error terms with too many vertices.
    \item $I_k$ only contains terms where $L, G_k \circ \ldots \circ G_1 \circ T \circ {G'_1}^{\T} \circ \ldots \circ {G'_k}^{\T}, L'^\T$ are properly composable. The remaining terms are $L, G_k \circ \ldots \circ G_1 \circ T \circ {G'_1}^{\T} \circ \ldots \circ {G'_k}^{\T}, L'^\T$ such that they are not properly composable, and these will mostly be put into $I_{k+1}$.
    \[
    \sum_{\substack{T \in \calM, (G_k ,\dots,  {G'_k}) \in \midint^{(k)}_{T}, L,L' \in \calL:\\
    \abs{V(L \odot G_k \odot \ldots \odot G_1)} \leq D_V, \abs{V(T)} \leq D_V, \abs{V({G'_1}^{\T} \odot \ldots \odot {G'_k}^{\T} \odot {L'}^\T)} \leq D_V, \\
        L, G_k \circ \cdots G_1 \circ T \circ G_1'^\T \circ \cdots G_k'^\T, L'^\T\text{ are not properly composable},\\
\text{edge bounds satisfied}}}
    \lam_{L \odot G_k \odot \ldots \odot {G'_k}^{\T} \odot {L'}^{\T}} {M_L}M_{G_k \circ \cdots \circ {G'_k}^{\T}}{M_{L'}}^{\T}
    \]
    
    These terms do not yet match our inductive hypothesis for $I_{k+1}$;
    for each $L, G_k \circ \ldots \circ G_1 \circ T \circ {G'_1}^{\T} \circ \ldots \circ {G'_k}^{\T}, L'^\T$ which are not properly composable, we must separate out the intersecting portions $G_{k+1}, G_{k+1}'$. To do this, 
    decompose $L$ as $L_2 \circ G_{k+1}$ where $B_{L_2} = A_{G_{k+1}}$ is the leftmost minimum vertex separator between $A_L$ and $B_{L} \cup \{\text{intersected vertices}\}$. We decompose $L'$ as $L'_2 \circ G'_{k+1}$ in a similar way.
    This is exactly the definition that $G_{k+1}, \dots, G_1,T,G_1',\dots,G_{k+1}$ are a middle intersection, \cref{def:middle-intersection}.
    That is, $(G_{k+1}, \dots, G_{k+1}') \in \midint^{(k+1)}.$
    
    \begin{claim}
    $L_2$ and $L_2'$ are left ribbons.
    \end{claim}
    \begin{proof}
        Definitionally, $B_{L_2}$ is the unique minimum vertex separator of $L_2$.
        Furthermore, reachability of all vertices in $L_2$ is inherited from $L$.
    \end{proof}
    We record a lemma that will be needed later:
    \begin{lemma}\label{lem:factoring-ignores-indicators}
        Factoring $L = L_2 \circ G_{k+1}$ is oblivious to edges inside $A_{L}, B_L$, edges incident to $A_L \cap B_L$, or edges not in $L$.
        \begin{proof}
            Edges inside $A_L$ or $B_L$ do not affect the connectivity between $A_L, B_L$; the same is true for edges incident to $A_L \cap B_L$ since
            all vertices of $A_L \cap B_L$ must be taken in any separator of $A_L$ and $B_L$.
            The last claim is clear because the factoring depends only on which vertices of $L$ intersected, as well as the structure of $L$. 
        \end{proof}
    \end{lemma}

    Observe that now instead of summing over $L, \midint^{(k)}, L'$, we may sum over $L_2, \midint^{(k+1)}, L_2'$. 
    That is, if we fix $T$ and the ribbons $G_k, \dots, G_{k}'$, then
$L$ determines the pair $L_2, G_{k+1}$,
and vice versa.

        
        We take $I_{k+1}$ to be the subset of $G_{k+1}$ that satisfy the edge bound ${|E_{mid}(G_{k+1})| - |V(G_{k+1})| \leq C\dsos}$. ($L_2$ and $L_2'$ are written as $L,L'$ here so as to more clearly match the inductive hypothesis.)
        \[
            \sum_{\substack{T \in \calM, (G_{k+1}, \ldots, {G'_{k+1}}) \in \midint^{(k+1)}_{T}, L,L' \in \calL:\\
            \abs{V(L\odot G_{k+1} \odot \ldots \odot G_1)} \leq D_V, \abs{V(T)} \leq D_V, \abs{V({G'_1}^{\T} \odot \ldots \odot {G'_{k+1}}^{\T} \odot {L'}^{\T})} \leq D_V, \\
            L, G_{k+1} \circ \ldots \circ G_1 \circ T \circ {G'_1}^{\T} \circ \ldots \circ {G'_{k+1}}^{\T}, L'^\T\text{ are properly composable},\\
            \text{edge bounds satisfied}}}
            \lam_{L \odot G_{k+1} \odot \ldots \odot {G'_{k+1}}^{\T} \odot {L'}^{\T}} M_{L \circ G_{k+1} \circ \ldots \circ {G'_{k+1}}^{\T} \circ {L'}^T}
        \]
    \item For technical reasons, we want to stop this process if $|E_{mid}(G_{k+1})| - |V(G_{k+1})| > C{\dsos}$ or $|E_{mid}(G'_{k+1})| - |V(G'_{k+1})| > C{\dsos}$. Thus we take such terms $G_{k+1}$ and ${G'_{k+1}}^{\T}$ and add them to $T_{k+1}$ instead of to $I_{k+1}$.
    \[
    \sum_{\substack{T \in \calM, (G_{k+1}, \ldots, {G'_{k+1}}) \in \midint^{(k+1)}_{T}, L,L' \in \calL:\\
    \abs{V(L\odot G_{k} \odot \ldots \odot G_1)} \leq D_V, |V(T)| \leq D_V, \abs{V({G'_1}^{\T} \odot \ldots \odot {G'_{k}}^{\T} \odot {L'}^{\T})} \leq D_V, \\
    |E_{mid}(G_{k+1})| - |V(G_{k+1})| > C{\dsos} \text{ or }|E_{mid}(G'_{k+1})| - |V(G'_{k+1})| > C{\dsos}, \\
    L, G_{k+1} \circ \ldots \circ G_1 \circ T \circ {G'_1}^{\T} \circ \ldots \circ {G'_{k+1}}^{\T}, {L'}^\T\text{ are properly composable},\\
    \text{edge bounds satisfied up to }k}}
    \lam_{L \odot G_{k+1} \odot \cdots \odot {G'_{k+1}}^{\T} \odot {L'}^{\T}} M_{L \circ G_{k+1} \circ \cdots \circ {G'_{k+1}}^{\T} \circ {L'}^T}
    \]
    This is the source of the truncation error terms with too many edges.
\end{enumerate}
Iteratively applying this procedure gives the result.
\begin{proposition}
    The recursion terminates within $2\dsos$ steps and $I_{2\dsos} = 0$.
\end{proposition}
\begin{proof}
Each additional intersection $G_{k+1},G_{k+1}'$ is nontrivial and increases $\abs{A_{G_{k+1}}} + \abs{A_{G_{k+1}'}}$.
Since $\abs{A_{G_{k}}}$ is upper bounded by $\dsos$,
the recursion ends within $2\dsos$ steps.
\end{proof}
\end{proof}

\end{document}